\documentclass[reqno]{amsart}

\usepackage{amsfonts,latexsym,enumerate}
\usepackage{amsmath}
\usepackage{amscd}
\usepackage{float,amsmath,amssymb,mathrsfs,bm,multirow,graphics}
\usepackage[dvips]{graphicx}
\usepackage[percent]{overpic}
\usepackage{amsaddr}
\usepackage[numbers,sort&compress]{natbib}
\usepackage{tikz}
\usepackage{epstopdf}
\usepackage{subfigure}
\usepackage{tkz-euclide}

\newcommand{\R}{{\Bbb R}}

\newcommand{\C}{{\Bbb C}}

\newcommand{\Z}{{\Bbb Z}}

\newtheorem{theorem}{Theorem}[section]
\newtheorem{proposition}[theorem]{Proposition}
\newtheorem{lemma}[theorem]{Lemma}

\newtheorem{assumption}[theorem]{Assumption}
\newtheorem{remark}[theorem]{Remark}

\newtheorem{RHproblem}[theorem]{RH problem}
\definecolor{mred}{rgb}{0.76,0.23,0.19}
\definecolor{mgreen}{rgb}{0.21,0.34,0.04}
\definecolor{mblue}{rgb}{0.0,0.29,0.63}
\numberwithin{equation}{section}

\begin{document}
	
	\title[Long-time asymptotics of the SK equation]
	{Long-time asymptotics of the Sawada-Kotera equation on the line}
	
	\author{Deng-Shan Wang$^1$,~Xiaodong Zhu$^{1,2}$}
	\address{$^1$Laboratory of Mathematics and Complex Systems (Ministry of Education), School of Mathematical Sciences, Beijing Normal University, Beijing 100875, China}
    \address{ $^2$SISSA, via Bonomea 265, 34136 Trieste, Italy, INFN Sezione di Trieste}
	\email{dswang@bnu.edu.cn, xdzbnu@mail.bnu.edu.cn}

	
	\date{February 3, 2025.}
	
	
	
	\begin{abstract}
    
		The Sawada-Kotera (SK) equation is an integrable system characterized by a third-order Lax operator and is related to the modified Sawada-Kotera (mSK) equation through a Miura transformation. This work formulates the Riemann-Hilbert problem associated with the SK and mSK equations by using direct and inverse scattering transforms. The long-time asymptotic behaviors of the solutions to these equations are then analyzed via the Deift-Zhou steepest descent method for Riemann-Hilbert problems. It is shown that the asymptotic solutions of the SK and mSK equations are categorized into four distinct regions: the decay region, the dispersive wave region, the Painlev\'{e} region, and the rapid decay region. Notably, the Painlev\'{e} region is governed by the F-XVIII equation in the Painlev\'{e} classification of fourth-order ordinary differential equations, a fourth-order analogue of the Painlev\'{e} transcendents. This connection is established through the Riemann-Hilbert formulation in this work. Similar to the KdV equation, the SK equation exhibits a transition region between the dispersive wave and Painlev\'{e} regions, arising from the special values of the reflection coefficients at the origin. Finally, numerical comparisons demonstrate that the asymptotic solutions agree excellently with results from direct numerical simulations.

	\end{abstract}
	
	\maketitle
	\tableofcontents

	\section{\bf Introduction}
	
The study of initial-value problems of integrable systems often involves developing inverse spectral theory of an ordinary differential operator
$$
\mathscr{L}=D^{n}+q_{n-2}D^{n-2}+\cdots+q_0,\quad n\geq 2,\quad D=d/dx,
$$
where the coefficients $q_j~(j=0,1,2,\cdots,n-2)$ are assumed to belong to the Schwartz class $\mathcal{S}(\R)$. Beals, Deift and Tomei \cite{Beals-1985,Beals-Coifman-1984,Beals-Deift-1988} investigated the direct and inverse scattering problem for this operator on the line. Subsequently, Deift and Zhou \cite{Deift-Zhou-1991} considered the case with arbitrary spectral singularities. For $n=2$, the one-dimensional Schr\"{o}dinger operator $\mathscr{L}=d^2/dx^2+q_0$ is related with the inverse scattering problem of the KdV equation, which was first established by Gardner, Greene, Kruskal and Miura \cite{GGKM-1967} in 1967, and then by Deift and Trubowitz \cite{Deift-Trubowitz}. For $n=3$, the third-order operator $\mathscr{L}=d^3/dx^3+q_1d/dx+q_0$ associates with the spectral problem of several famous nonlinear integrable systems \cite{Zakharov1973,Kaup-1980}. For example, the constrains $q_1=2q$ and $q_0=q_x+p$ correspond to the good Boussinesq equation \cite{McKean1981,Charlier-Lenells-2021,Charlier-Lenells-Wang-2021}
\begin{equation*}
p_t+\frac{1}{3}q_{xxx}+\frac{4}{3}(q^2)_x=0,\quad q_t=p_x,
\end{equation*}
and the constrains $q_1=6u$ and $q_0=0$ correspond to the Sawada-Kotera (SK) equation
\begin{equation}\label{SK}
		u_t+u_{x x x x x}+30\left(u u_{x x x}+u_x u_{x x}\right)+180 u^2 u_x=0,
\end{equation}
which was first proposed by Sawada and Kotera \cite{Sawada-Kotera-1974} in 1974 and then derived by Caudrey, Dodd and Gibbon \cite{Caudrey-Dodd-1976} independently, while the constrains $q_1=6v$ and $q_0=3v_x$ correspond to the Kaup-Kupershmidt (KK) equation
\begin{equation}\label{KK}
		v_t+v_{xx x x x}+30(v v_{x x x}+\frac{5}{2} v_x v_{x x})+180 v^2 v_x=0,
\end{equation}
which was given by Kaup \cite{Kaup-1980} and Kupershmidt \cite{Kupershmidt-1984}, respectively. In addition, Fordy and Gibbons \cite{Fordy-Gibbons-1980} found that both the SK equation (\ref{SK}) and the KK equation (\ref{KK}) were related with the modified SK (mSK) equation, also named Fordy-Gibbons-Jimbo-Miwa equation \cite{Fordy-Gibbons-1980,Jimbo-Miwa-1983}:
\begin{equation}\label{msk-equation}
		w_t+w_{x x x x x}-(5 w_x w_{x x}+5 w w_x^2+5 w^2 w_{x x}-w^5)_x=0,
\end{equation}
through the Miura transformations
\begin{equation}\label{miura-SK}
		u=\frac{1}{6}(w_x-w^2)\quad {\rm and} \quad v=\frac{1}{3}(w_x-\frac{w^2}{2}).
\end{equation}
\par
Both SK equation (\ref{SK}) and KK equation (\ref{KK}) are intriguing fifth-order nonlinear evolution equations that describe the dynamics of nonlinear waves in a liquid medium interspersed with gas bubbles \cite{Kudryashov-2014}. These equations stand out as completely integrable systems, each featuring a third-order Lax pair, solvable by inverse scattering transform, and owning a characteristic that distinguishes them from the fifth-order KdV equation \cite{Lax-1968}, which is associated with a second-order Lax pair. They have successfully passed the Painlev\'{e} test, a critical criterion for integrability, and exhibit bi-Hamiltonian structures, which are essential for understanding their rich mathematical properties. Moreover, they support multi-soliton solutions, a feature that is highly prized in the study of wave interactions. From a geometric perspective, the SK equation comes from a planar curve flow that is integrable within the context of affine geometry, while the KK equation arises from an integrable planar curve flow in projective geometry. Furthermore, nontrivial Liouville correspondences exist, linking the Novikov equation \cite{Boutet-de-Monvel-2} to the SK equation, as well as the Degasperis-Procesi equation \cite{Constantin-4} to the KK equation. In addition, by means of group-invariant reduction, the SK equation (\ref{SK}), KK equation (\ref{KK}) and mSK equation (\ref{msk-equation}) are intricately connected to the fourth-order analogues of Painlev\'{e} transcendent
\begin{equation}\label{4th-Painleve}
p^{(4)}=5p(p')^2+5p'p''+sp+5p^2p''-p^5,\quad p=p(s),
\end{equation}
which is the F-XVIII equation in the Painlev\'{e} classification of the fourth-order ordinary differential equations in polynomial class \cite{Cosgrove-2000,Cosgrove-2006}, i.e., the fourth-order analogues of the Painlev\'{e} transcendent. For example, take the self-similar transformation $w(x,t)=(5t)^{-\frac{1}{5}}p(s)$ with $s=\frac{x}{(5t)^{\frac{1}{5}}}$, then the mSK equation (\ref{msk-equation}) is reduced to the ordinary differential equation
\begin{equation}\label{self transformation}
		p^{(5)}-5 p'^3-10p p' p''-5 p''^2-5 p' p^{(3)}-s p'-p-5 p^{(3)} p^2-10p p' p''+5 p^4 p'=0.
\end{equation}
Integrating (\ref{self transformation}) equation once and setting the integral constant to be zero, yields Painlev\'{e} transcendent equation (\ref{4th-Painleve}).
\par
In 1993, Deift and Zhou \cite{Deift-Zhou-1993} introduced a potent nonlinear steepest-descent approach to investigate the oscillatory Riemann-Hilbert (RH) problems associated with the modified KdV (mKdV) equation, which features initial conditions of the Schwartz class. Notably, they discovered that the central region of the problem is a self-similar region, elegantly captured by the unique solution of the Painlev\'{e} II equation. It is significant to highlight that numerous other integrable equations, characterized by vanishing boundary conditions, also exhibit self-similar regions, including the KdV equation \cite{Ablowitz-1977,Deift-Zhou-1994,Constantin-2,Tamara} and Camassa-Holm equation \cite{Ann-Teschl-2009,Constantin-0,Constantin-1,Constantin-3}. Moreover, during the conference ``Integrable Systems, Random Matrices, and Applications,” held at the Courant Institute in May 2006, Deift \cite{Deift-2008} was asked to present a list of unsolved problems, in which he presented that the study of long-time asymptotic behavior of integrable systems with third-order spectral problem is an extremely challenging issue. {Recent research has further explored the so-called ``bad'' or ill-posed Boussinesq equation and its long-time dynamics. This equation has been studied extensively from both analytical and numerical perspectives, including investigations of long-time asymptotics, blow-up phenomena, and the soliton resolution conjecture \cite{bad4,bad2,bad3,bad1}. These results provide valuable context for integrable with higher-order Lax structures and situate the present analysis of $3\times 3$ Riemann--Hilbert problems and soliton resolution within a broader framework of dispersive models.}

The current work will demonstrate that self-similar regions also emerge in the long-time asymptotic behavior of the SK equation and mSK equation. These regions are encapsulated by the fourth-order analogues of the Painlev\'{e} transcendent. Moreover, the rapid decay region and similar dispersive wave region (also called Zakharov-Manakov region) are also formulated by deforming the RH problem based on the nonlinear steepest-descent approach.
\par
The Lax pair of the SK equation (\ref{SK}) in matrix form is
\begin{equation}\label{SK-lax-pair}
		\left\{\begin{array}{l}
			\Phi_x=L \Phi, \\
			\Phi_t=Z \Phi,
		\end{array}\right.
\end{equation}
where
\begin{equation}\label{SKlaxspace}
		L=\left(\begin{array}{ccc}
			0 & 1 & 0 \\
			0 & 0 & 1 \\
			k^3 & -6 u & 0
		\end{array}\right),\\	
	\end{equation}
	\begin{equation}\label{SKlaxtime}
		Z=\left(\begin{array}{ccc}
			36 k^3 u & 6 u_{x x}-36 u^2 & 9 k^3-18 u_x \\
			18 k^3 u_x+9 k^6 & 6 u_{x x x}-18 k^3 u+36 u u_x & -12 u_{x x}-36 u^2 \\
			6 k^3 u_{x x}-36 u^2 k^3 & Z_{32} & -6 u_{x x x}-18 k^3 u-36 u u_x
		\end{array}\right),
	\end{equation}
	with spectral parameter $k$ and $Z_{32}=36 u_x{ }^2+108 u u_{x x}+9 k^6+216 u^3+6 u_{x x x x}$.
	\par
	The mSK equation (\ref{msk-equation}) has Lax pair
	\begin{equation}\label{msk-lax-pair}
		\left\{\begin{array}{l}
			\Phi_x=\mathcal{M} \Phi, \\
			\Phi_t=\mathcal{N} \Phi,
		\end{array}\right.
	\end{equation}
	where
	\begin{align*}
		\mathcal{M} & =\left(\begin{array}{ccc}
			0 & 1 & 0 \\
			0 & -w & 1 \\
			k & 0 & w
		\end{array}\right), \\
		\mathcal{N} & =\left(\begin{array}{lll}
			-6 k w^2+6 k w_x & N_{12} & -3 w_{x x}+9 k+6 w_x w \\
			-6 k w w_x+9 k^2+3 k w_{x x} & N_{22} & N_{23} \\
			N_{31} & 9 k^2 & N_{33}
		\end{array}\right),
	\end{align*}
	with $N_{12}=-3w^2_x-w^4+w_{xxx}-9k w-4w_x w^2+w_{xx}w, N_{22}=-3k w_{x}+3k w^2+w_{xxxx}w+w^5-5w_{xx}w^2-5w_xw_{xx}-5w_{x}^2w, N_{23}=-w^4+3w_{x}^2-2w_{xxx}+2w_{x}w^2+4w_{xx}w, N_{31}=-3k w_{x}^2-k w^4+k w_{xxx}+9w^2k-4k w_{x}w^2+k w_{xx}w$, $N_{33}=-w_{xxxx}+3k w^2-3k w_{x} -w^5+5w_{xx}w^2+5w_{x}w_{xx}+5w_x^2w$.
\par
The investigation of long-time asymptotics of SK equation (\ref{SK}) not only deepens our understanding of complex nonlinear systems but also stimulate the development of innovative mathematical techniques and theories. Notice that the spectral problem of the SK equation (\ref{SK}) has singularity at $k=0$ after diagonalization, while the singularity at $k=0$ is absent within the spectral problem of the mSK equation (\ref{msk-equation}). Thus it is practicable to study the long-time asymptotics of Painlev\'{e} region for the SK equation (\ref{SK}) by the examining the asymptotic behavior of the mSK equation (\ref{msk-equation}).
\par
This work is organized as follows: In Section \ref{Main-Results}, the RH problems associated with the SK equation (\ref{SK}) and the mSK equation (\ref{msk-equation}) are proposed. Moreover, the main results concerning the long-time asymptotics of the SK and mSK equations are presented in Theorem \ref{msk Thm} and Theorem \ref{thm for sk}, respectively. It is shown that as $t \to \infty$, the solutions of the SK equation (\ref{SK}) and the mSK equation (\ref{msk-equation}) can be categorized into four distinct regions: the decay region, the dispersive wave region, the Painlev\'e region, and the rapid decay region. It is worth noting that, analogous to the KdV equation, the SK equation  (\ref{SK}) features a transition region between the dispersive wave region and the Painlev\'e region due to the special values of the reflection coefficients at the original point $k=0$. Numerical comparisons reveal that the asymptotic solutions are in remarkably close agreement with the results obtained by direct numerical simulations. The inverse scattering transform of the SK and mSK equations is studied in Section \ref{RHPandmiura}, and the corresponding the RH problems are formulated. Additionally, the Miura transformation connecting the SK equation (\ref{SK}) and the mSK equation (\ref{msk-equation})  is established. In Section \ref{Sector I II}, the Deift-Zhou steepest descent method is applied to analyze the dispersive wave region of the SK and mSK equations, revealing that the long-time behavior can be expressed as the sum of two modulated cosine traveling waves decaying as $1/\sqrt{t}$. Furthermore, for $x \sim t^{\frac{1}{5}}$ as $t \to \infty$, the leading-order term of the long-time asymptotics is described by the fourth-order analogues of the Painlev\'{e} transcendent (\ref{4th-Painleve}), as detailed in Section \ref{painleve region}. Finally, the rapid decay region is analyzed in Section \ref{rapid region}.

\section{\bf Main Results}\label{Main-Results}
	
This section presents the primary findings of the current work. Similar to the relationship between the KdV and mKdV equations, the Miura transformation establishes a strong connection between the SK equation (\ref{SK}) and the mSK equation (\ref{msk-equation}), as detailed in Theorem \ref{Miura theorem}. For the initial value problem of the SK equation \eqref{SK}, direct scattering analysis enables the definition of the scattering matrices \(s(k) = (s_{ij}(k))_{3 \times 3}\) and \(s^{A}(k) = (s^{A}_{ij}(k))_{3 \times 3}\) in \eqref{Scattering-Ma} and \eqref{Scattering-Ma-A}, respectively. For the mSK equation (\ref{msk-equation}), denote the scattering matrices as 
\(\tilde{s}(k)=(\tilde{s}_{ij}(k))_{3 \times 3}\) and \(\tilde{s}^A(k)=(\tilde{s}_{ij}^A(k))_{3 \times 3}\), defined in (\ref{s Ma msk}) below. Firstly, we present some results regarding the mSK equation (\ref{msk-equation}), particularly focusing on the long-time asymptotic analysis, based on the scattering data and the RH problem discussed in \cite{procA}. Then, this paper details its key contributions through the formulation and proof of the main theorems related to the SK equation (\ref{SK}). These theorems emerge from a foundational theoretical framework based on a set of basic assumptions. Below, we outline the essential assumptions, which are crucial for deriving the main results:
\par		
\begin{assumption}\label{solitonless}For the initial value problem of the SK equation \eqref{SK}, assume that the elements 
\(s_{11}(k)\) and \(s_{11}^A(k)\) are nonzero for \(k \in \bar{\Omega}_1 \setminus \{0\}\) and \(k \in \bar{\Omega}_4 \setminus \{0\}\), respectively. This assumption also holds for the mSK equation (\ref{msk-equation}), as specified in Assumption 3.5 of Ref. \cite{procA}. In essence, we posit the absence of solitons in the initial value problems of the equations \eqref{SK} and \eqref{msk-equation}, focusing solely on their pure radiation solutions.
\end{assumption}
\par
The following discussion will validate the assumption by selecting a specific initial value and performing numerical calculations. Our preliminary discovery unveils two spectral functions, \(r_1(k)\) and \(r_2(k)\), which are derived from the initial conditions of the SK equation \eqref{SK} and serve as reflection coefficients. These functions play a pivotal role in formulating the RH problem and in accurately reconstructing the solution within this framework. Similarly, denote the reflection coefficients for the mSK equation (\ref{msk-equation}) as \(\tilde{r}_1(k)\) and \(\tilde{r}_2(k)\).
		\par
		To be specific, the reflection coefficients \(r_1(k)\) and \(r_2(k)\) are defined as:
		\begin{align}\label{r1r2def}
			\begin{cases}
				r_1(k):=\frac{s_{12}(k)}{s_{11}(k)},\quad k \in \mathbb{R}_+,\\
				r_2(k):=\frac{s_{12}^A(k)}{s_{11}^A(k)},\quad k \in \mathbb{R}_-.
			\end{cases}
		\end{align}
		Similarly, for the mSK equation (\ref{msk-equation}), the reflection coefficients \(\tilde{r}_1(k)\) and \(\tilde{r}_2(k)\) are defined as:
		\begin{align}\label{r1r2mskdef}
			\begin{cases}
				\tilde{r}_1(k):=\frac{\tilde{s}_{12}(k)}{\tilde{s}_{11}(k)},\quad k \in \mathbb{R}_+,\\
				\tilde{r}_2(k):=\frac{\tilde{s}_{12}^A(k)}{\tilde{s}_{11}^A(k)},\quad k \in \mathbb{R}_-.
			\end{cases}
		\end{align}
		In Proposition \ref{sprop} below, we demonstrate that the matrix entries \(s_{11}(k)\) and \(s_{12}(k)\) of the scattering matrix $s(k)$, are smooth functions over the interval \(k \in (0,\infty)\), except for \(k=0\), which is a simple pole. In contrast, the scattering matrix \(\tilde{s}(k)\) for the mSK equation (\ref{msk-equation}) is regular at \(k=0\). Consequently, the reflection coefficients \(r_j(k)\) and \(\tilde{r}_j(k)\) \((j=1,2)\), exhibit different properties at the origin. Next, we will recall some key facts about the mSK equation (\ref{msk-equation}) and discuss results related to its long-time behaviors. Subsequently, we will illustrate corresponding results regarding the SK equation \eqref{SK}.
		
		\subsection{The modified Sawada-Kotera equation}
		The formulation of our main result entails two  scattering matrices, \(\tilde{s}(k)\) and \(\tilde{s}^A(k)\), defined as follows (see \cite{procA} for details). Let \(w_0(x)=w(x,0)\) be a real-valued function in \(\mathcal{S}(\mathbb{R})\), and denote \(\omega := \mathrm{e}^{\frac{2\pi i}{3}}\). Suppose that
		$$
		V_1(x;k)= G(k)^{-1}\begin{pmatrix}
			0 & 0 &0\\
			0 & -w_0(x) &0\\
			0 & 0 &w_0(x)
		\end{pmatrix} G(k),
		$$
		where \(G(k)\) is defined in (\ref{gauge-matrix}). Further define the \(3\times3\) matrix-valued eigenfunctions by the following Volterra integral equations:
		$$
		\begin{aligned}
			&\tilde{J}_{+}(x;k)=I-\int_x^{\infty} \mathrm{e}^{(x-s) \widehat{k\Lambda}}\left(V_1 \tilde{J}_{+}\right)(s;k) \mathrm{d} s,\\
			&\tilde{J}_{+}^A(x;k)=I+\int_x^{\infty} \mathrm{e}^{-(x-s) \widehat{k\Lambda}}\left(V_1^T \tilde{J}_{+}^A\right)(s;k) \mathrm{d} s,
		\end{aligned}
		$$
		with \(\Lambda:=\mathrm{diag}\{\omega,\omega^2,1\}\), \(\widehat{k\Lambda}\) is an operator, where \(\mathrm{e}^{\widehat{k\Lambda}}A=\mathrm{e}^{k\Lambda}A\mathrm{e}^{-k\Lambda}\), and \(V_1^T\) denotes the transpose of \(V_1\). Now, the scattering matrices \(\tilde{s}(k)\) and \(\tilde{s}^A(k)\) are defined by
		\begin{equation}\label{s Ma msk}
			\begin{aligned}
				&\tilde{s}(k)=I-\int_{-\infty}^{\infty} \mathrm{e}^{-x \widehat{k\Lambda}}\left(V_1 \tilde{J}_{+}\right)(x;k) \mathrm{d} x,\\
				&\tilde{s}^A(k)=I+\int_{-\infty}^{\infty} \mathrm{e}^{x \widehat{k\Lambda}}\left(V_1 \tilde{J}_{+}^A\right)(x;k) \mathrm{d} x.
			\end{aligned}
		\end{equation}
        \par
        The following theorem can be proved by standard way.
		\begin{theorem}
			Suppose that $w_0(x) \in \mathcal{S}(\mathbb{R})$, then the reflection coefficients $\tilde{r}_1(k)$ and $\tilde{r}_2(k)$ are well-defined for $k \in\R_+$ and $k \in\R_-$, respectively, and satisfy the following properties:
			\begin{enumerate}
				\item The functions $\tilde{r}_1(k)$ and $\tilde{r}_2(k)$ are smooth for $k$ in their domain and decay rapidly as $k \rightarrow \infty$.
				\item The reflection coefficients satisfy \(|r_j(k)|\leq 1\) for \(k\) belonging to their respective domains. Meanwhile, for potential function $w_0(x)$ with compact support, \(|r_j(k)|<1\) for \(j=1,2\).
			\end{enumerate}
		\end{theorem}
		\begin{assumption}\label{Painelve assumption}
			Assume that the reflection coefficients $\tilde{r}_j(k)~(j=1,2)$ are strictly less than 1. In particular, suppose the reflection coefficients at $k=0$ satisfy the relations:
			$$		\tilde{r}_2(0)=\frac{\tilde{r}_2^*(0)^2-\tilde{r}_1^*(0)}{\tilde{r}_2^*(0)\tilde{r}_1^*(0)-1},\quad
			\tilde{r}_1(0)=\frac{\tilde{r}_1^*(0)^2-\tilde{r}_2^*(0)}{\tilde{r}_2^*(0)\tilde{r}_1^*(0)-1},
			$$
			which is related with the Painlev$\acute{e}$ model in Appendix \ref{appendix painleve}.
		\end{assumption}
		\begin{RHproblem}\label{msk RHP}
			Given the reflection coefficients $\tilde{r}_{1}(k)$ and $\tilde{r}_{2}(k)$ associated with the mSK equation (\ref{msk-equation}), find a $3 \times 3$ matrix-valued function $m(x, t; k)=m_n(x, t; k)$ for $k\in\Omega_n,~n=1,\cdots,6$ in Figure \ref{sigma0}
			with the following properties:

			\begin{enumerate}[(a)]
				\item  \(m_n(x, t; k): \mathbb{C} \setminus \Sigma \rightarrow \mathbb{C}^{3 \times 3}\) is analytic for \(k \in \mathbb{C} \setminus \Sigma\), where \(\Sigma = \bigcup_{j=1}^3 \mathrm{e}^{(j-1)\pi i/3}\mathbb{R}\) (see Figure \ref{sigma0}).
				
				\item As $k$ approaches $\Sigma  $ from the left (+) and right (-), the limits of $ m(x, t; k)$  exist, are continuous on $\Sigma  $, and are related by
				$$
				m_{+}(x, t; k)=m_{-}(x, t; k) v(x, t; k), \quad k \in \Sigma,
				$$
				where, if $k\in \mathrm{e}^{(j-1)\pi i/3}\R_+$ for $j=1,2,\cdots,6$, then $v(x, t; k)=v_n(x, t; k)$, where $v_n(x, t; k)~(n=1,2,\cdots,6)$ are defined in terms of $\tilde{r}_{1}(k)$ and $\tilde{r}_{2}(k)$ by $(\ref{vn-jump-matrix})$.
				
				\item The matrix-valued functions $m_{n}(x,t; k)$ exhibit the following symmetries
				\begin{equation}\label{m symmeties}
					m_n(x,t; k)=\mathcal{A} m_n(x, t;\omega k) \mathcal{A}^{-1}={\mathcal{B} m_n^*(x, t;k^*)} \mathcal{B},
				\end{equation}
              where the matrices $\mathcal{A}$ and $\mathcal{B}$ are
		\begin{equation}\label{AB-Matrices}
		\mathcal{A}:=\left(\begin{array}{ccc}
			0 & 0 & 1 \\
			1 & 0 & 0 \\
			0 & 1 & 0
		\end{array}\right) \quad \text { and } \quad \mathcal{B}:=\left(\begin{array}{lll}
			0 & 1 & 0 \\
			1 & 0 & 0 \\
			0 & 0 & 1
		\end{array}\right).
		\end{equation}
				
				\item $m(x, t; k)=I+\frac{m^{(1)}_{\infty}(x,t)}{k}+\mathcal{O}\left(k^{-2}\right)$ as $k \rightarrow \infty,~ k \notin \Sigma $, with
				\begin{equation}\label{expansion of m}
					m^{(1)}_{\infty}(x,t)=\frac{w(x,t)}{3}\left(\begin{array}{ccc}
						0 & \omega & 1 \\
						\omega^2 & 0 & 1 \\
						\omega^2 & \omega & 0
					\end{array}\right)+\frac{1}{3} \int_{\infty}^x w(x^{\prime},t)^2 \mathrm{~d} x^{\prime}\left(\begin{array}{ccc}
						\omega^2 & 0 & 0 \\
						0 & \omega & 0 \\
						0 & 0 & 1
					\end{array}\right) .
				\end{equation}
				
				\item $m(x, t; k)=\sum_{l=0}^{p} {m}^{(l)}_0(x,t) k^{l}+\mathcal{O}(k^{p+1})$ as $k \rightarrow 0$.
			\end{enumerate}
		\end{RHproblem}
		\begin{theorem}\label{thm for msk}
			Suppose the initial data \(w_0(x) \in \mathcal{S}(\mathbb{R})\) and the scattering data satisfy Assumption \ref{solitonless}. Define the reflection coefficients \(\tilde{r}_1(k)\) and \(\tilde{r}_2(k)\) with respect to \(w_0(x)\) as per (\ref{r1r2mskdef}). It is then established that the RH problem \ref{msk RHP} admits a unique solution 
            \(m(x, t; k)\) whenever it exists, for each point in the domain \((x, t) \in \mathbb{R} \times [0, T)\). Furthermore, the solution \(w(x, t)\) of the mSK equation (\ref{msk-equation}) for all \((x, t) \in \mathbb{R} \times [0, T)\) can be expressed by
			\begin{equation}\label{recover-formula mSK}
				w(x, t)= 3\lim _{k \rightarrow \infty}(k m(x, t; k))_{13}.
			\end{equation}
		\end{theorem}
		The above results were proven in \cite{procA} and can also be found in \cite{Charlier-Lenells-2021} and \cite{Charlier-Lenells-Wang-2021}. Based on the intricate link between the solutions of the mSK equation (\ref{msk-equation})  with Schwartz class initial conditions and the RH problem \ref{msk RHP}, the long-time asymptotics of the solution to the mSK equation (\ref{msk-equation}) is formulated in the theorem below.
        
		\begin{theorem}\label{msk Thm}
			Let \(w_0(x) \in \mathcal{S}(\mathbb{R})\) satisfy the assumptions of Theorem \ref{thm for msk}. Then for $\xi=\frac{x}{t}$, the solution \(w(x,t)\) of the initial value problem for mSK equation \eqref{msk-equation} exhibits the following asymptotic behaviors as 
            \((x,t) \to \infty\) in the $(x,t)$-half plane (see Figure \ref{categoryformsk}):
			\par
			\noindent {\bf Sector \rm I:}
			$
			w(x, t)=\frac{\tilde{A}_1(\xi)}{\sqrt{t}} \cos \tilde{\alpha}_1(\xi, t)+\frac{\tilde{A}_2(\xi)}{\sqrt{t}} \cos \tilde{\alpha}_2(\xi, t)+$
			$\mathcal{O}\left(\frac{1}{x^N}+\frac{\mathrm{C}_N(\xi)\ln(x)}{x}\right),\  M\le\xi<\infty;
			$
			\par
			\noindent {\bf Sector \rm II:}
			$
			w(x, t)=\frac{\tilde{A}_1(\xi)}{\sqrt{t}} \cos \tilde{\alpha}_1(\xi, t)+\frac{\tilde{A}_2(\xi)}{\sqrt{t}} \cos \tilde{\alpha}_2(\xi, t)+
			\mathcal{O}\left(\frac{\log t}{t}\right),\ \frac{1}{M}\le\xi\le M;
			$
			\par
			\noindent {\bf Sector \rm III:}
			$
			w(x,t)=(5t)^{-\frac{1}{5}}p(s)
			$
			{$+{\mathcal{O}((5t)^{-\frac{2}{5}})}$}
			$
			,\quad |\xi|\le {M}{t^{-\frac{4}{5}}},
			$
			where $p=p(s)$ satisfies the fourth-order analogues of the Painlev\'{e} transcendent \cite{Cosgrove-2000,Cosgrove-2006} for $s=\frac{x}{(5t)^{\frac{1}{5}}}$:
			\begin{equation}\label{Painleve_equation_first}
				p^{(4)}=5p(p')^2+5p'p''+sp+5p^2p''-p^5;
			\end{equation}
			\par
			\noindent {\bf Sector \rm IV:}
			{$
				w(x,t)=\mathcal{O}((|x|+t)^{-j})~(j\geq1),\quad \frac{1}{M}\le|\xi|,x<0.
				$}
			\par
			\noindent Here $M>1$ and $\mathrm{C}_N(\xi)$ is rapidly decreasing as $\xi\to\infty$ for each $N\in \Z_+$. Moreover,
			$$
			\begin{aligned}
				&\tilde{A}_1(\xi):=-\frac{\sqrt{\tilde{\nu}_1}}{3^{\frac{1}{4}} 2 \sqrt{5 } k_0^{\frac{3}{2}}},\  \tilde{A}_2(\xi):=-\frac{\sqrt{\tilde{\nu}_4}}{3^{\frac{1}{4}} 2 \sqrt{5 } k_0^{\frac{3}{2}}},\\
				&\tilde{\alpha}_1(\xi, t):=\frac{19 \pi}{12}-\left(\arg \tilde{r}_1(k_0)+\arg \Gamma\left(i \tilde{\nu}_1\right)\right)-\left(36 \sqrt{3} t k_0^5\right)+\tilde{\nu}_1 \ln \left(3^{\frac{7}{2}} 20 t k_0^5\right)+\tilde{s}_1,\\
				&\tilde{\alpha}_2(\xi, t):=\frac{11 \pi}{12}-\left(\arg \tilde{r}_2(-k_0)+\arg \Gamma\left(i \nu_4\right)\right)-\left(36 \sqrt{3} t k_0^5\right)+\tilde{\nu}_4 \ln \left(3^{\frac{7}{2}} 20 t k_0^5\right)+\tilde{s}_2,
			\end{aligned}
			$$
			with $k_0=\sqrt[4]{\xi/45}$, $\Gamma(k)$ denotes the Gamma function, and
			$$
			\begin{aligned}
				&\tilde{\nu}_1:=-\frac{1}{2 \pi} \ln \left(1-\left|\tilde{r}_1\left(k_0\right)\right|^2\right),\ \tilde{\nu}_4=-\frac{1}{2 \pi} \ln \left(1-\left|\tilde{r}_2\left(-k_0\right)\right|^2\right),\\
				&\tilde{s}_1=\tilde{\nu}_4 \ln (4)+\frac{1}{\pi} \int_{-k_0}^{-\infty} \log _\pi \frac{\left|s-\omega k_0\right|}{\left|s-k_0\right|} \mathrm{d}\ln \left(1-\left|\tilde{r}_2(s)\right|^2\right)+\frac{1}{\pi} \int_{k_0}^{\infty} \log _0 \frac{\left|s-k_0\right|}{\left|s-\omega k_0\right|} \mathrm{d}\ln \left(1-\left|\tilde{r}_1(s)\right|^2\right),\\
				&\tilde{s}_2=\tilde{\nu}_1 \ln (4)+\frac{1}{\pi} \int_{k_0}^{\infty} \log _0 \frac{\left|s+\omega k_0\right|}{\left|s+k_0\right|} \mathrm{d}\ln \left(1-\left|\tilde{r}_1(s)\right|^2\right)+\frac{1}{\pi} \int_{-k_0}^{-\infty} \log _\pi \frac{\left|s+k_0\right|}{\left|s+\omega k_0\right|} \mathrm{d}\ln \left(1-\left|\tilde{r}_2(s)\right|^2\right).
			\end{aligned}
			$$
			{Furthermore, the asymptotic formula in Sector II holds uniformly with respect to $\xi=x / t$ in compact subset of the stated interval.}
			
			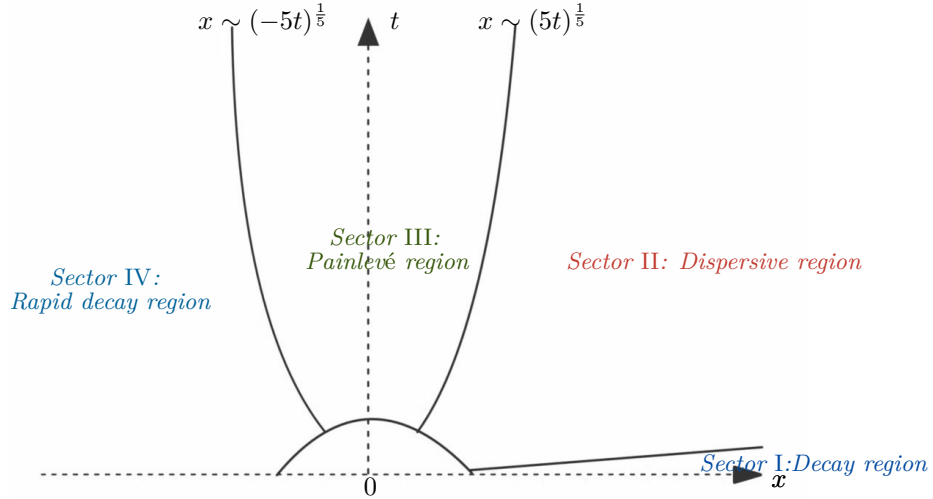
\begin{figure}[H]
				\centering
				\begin{overpic}[width=0.7\textwidth]{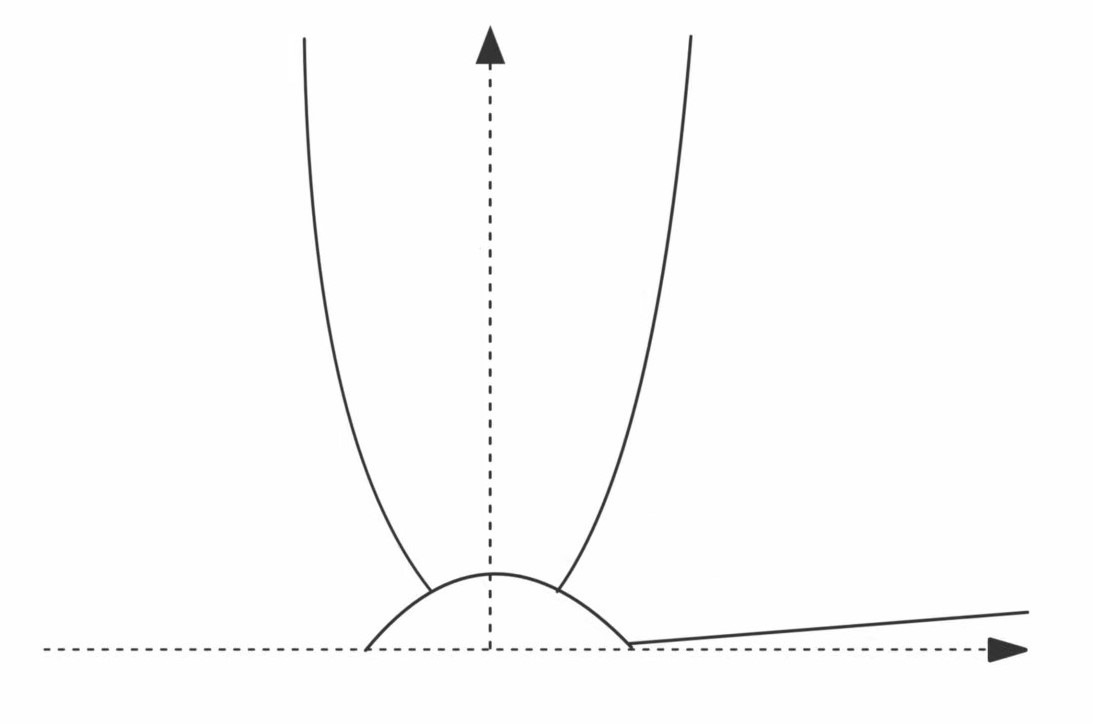}
					\put(69.6,32.5){\textcolor[rgb]{0.76,0.23,0.19}{\small Sector $\rm II$: Dispersive region }}
					\put(86.3,7.3){\textcolor[rgb]{0.0,0.29,0.63}{\small Sector $\rm I$:Decay region}}
					\put(40,35.5){\textcolor[rgb]{0.21,0.34,0.04}{\small Sector $\rm III$:}}
					\put(37,32.5){\textcolor[rgb]{0.21,0.34,0.04}{\small Painlev${\rm\acute{e}}$ region}}
					\put(5,30.5){\textcolor[rgb]{0.0,0.38,0.6}{\small Sector $\rm IV$: }}
                    
                    \put(0,27.5){\textcolor[rgb]{0.0,0.38,0.6}{\small Rapid decay region}}
					\put(58.6,62.5){\bf $x\sim (5t)^{\frac{1}{5}}$}
					\put(23.6,62.5){\bf $x\sim (-5t)^{\frac{1}{5}}$}
					\put(95.3,5.2){\bf $x$}
					\put(95.3,5.2){\bf $x$}
					\put(44.3,4.2){\bf $0$}
					\put(47.6,62.5){\bf $t$}
				\end{overpic}
				\caption{{\protect\small
						The asymptotic regions $\rm I$-$\rm IV$ in the $(x,t)$-half plane.}}
				\label{categoryformsk}
			\end{figure}
		\end{theorem}
		\begin{proof}
			The proof of Sectors $\rm I$ and $\rm II$ is illustrated in Section \ref{Sector I II}, the proof of Sector $\rm III$ is provided in Section \ref{painleve region}, and the proof of Sector $\rm IV$ is detailed in Section \ref{rapid region}.
		\end{proof}
		\subsection{The Sawada-Kotera equation}
		The scattering matrices \( s(k) \) and \( s^A(k) \) of the SK equation \eqref{SK} exhibit different properties compared to that of the mSK equation (\ref{msk-equation}). Notably, \( s(k) \) and \( s^A(k) \) have a simple pole at \( k=0 \). The details regarding these properties are provided in (\ref{sprop}). Moreover, in order to apply the Deift-Zhou steepest-descent method to the RH problem for the SK equation \eqref{SK},  the following assumption should be imposed.
		\begin{assumption}\label{r less 1}$($Generic behavior at $k=0$$)$. Generically, assume that
			$$
			\lim_{k\to0}k(s_{11}(k))\neq0,\ \lim_{k\to0}k(s^A_{11}(k))\neq0.
			$$
		\end{assumption}
		\noindent When \(u_0(x)=u(x,0) \in \mathcal{S}(\mathbb{R}) \) satisfies the Assumption \ref{solitonless} and Assumption \ref{r less 1}, the reflection coefficients \( r_j(k) \) for \( j=1,2 \) satisfy the properties stated in Theorem \ref{prop for r} below, with \( |r_j(k)| < 1~(j=1,2) \) except at \( k = 0 \).
		\begin{theorem}\label{prop for r}
			Suppose $u_{0}(x) \in \mathcal{S}(\mathbb{R})$, then  $r_{1}(k): (0, \infty) \rightarrow \mathbb{C}$ and $r_{2}(k): (-\infty, 0) \rightarrow \mathbb{C}$ have the following properties:
			$r_1(k)$ and $r_2(k)$ are smooth functions, rapidly decay as $|k|\to\infty $ in their domains and can be extended to $k=0$ in the way below
			$$
			r_{1}(k)=r_{1}(0)+r_{1}^{\prime}(0) k+\frac{1}{2} r_{1}^{\prime \prime}(0) k^{2}+\cdots, \quad k \rightarrow 0,\quad k>0,
			$$
			and
			$$
			r_{2}(k)=r_{2}(0)+r_{2}^{\prime}(0) k+\frac{1}{2} r_{2}^{\prime \prime}(0) k^{2}+\cdots, \quad k \rightarrow 0,\quad k<0,
			$$
			where $r_1(0)=\omega^2$ and $r_2(0)=1$.
		\end{theorem}
		\begin{remark}
			Notice that after gauge transformation (\ref{gauge-matrix}), the isospectral problem of the KK equation (\ref{KK}) has a double pole. In this case, the functions $J_{\pm}(x;k)$, $s(k)$, $J^{A}_{\pm}(x;k)$, $s^{A}(k)$ also have a double pole at $k=0$. Consequently, the behaviors of reflection coefficients $\breve{r}_1(k)$ and $\breve{r}_2(k)$ associated with the KK equation (\ref{KK})
			have different values with that of the SK equation (\ref{SK}), which are $\breve{r}_1(0)=\omega$ and $\breve{r}_2(0)=1$.
		\end{remark}
		\begin{RHproblem}\label{rhp SK}
			Given the reflection coefficients $r_{1}(k)$ and $r_{2}(k)$ associated with the SK equation (\ref{SK}), find a $3 \times 3$ matrix-valued function $M(x, t; k)=M_n(x, t; k)$ for $k\in\Omega_n,~n=1,\cdots,6$ with the following properties:
			
			(a) $M_n(x, t; k):~ \mathbb{C} \backslash \Sigma \rightarrow \mathbb{C}^{3 \times 3}$ is analytic for $k\in \mathbb{C} \backslash \Sigma$.
			
			(b) As $k$ approaches $\Sigma $ from the left (+) and right (-), the limits of $ M(x, t; k)$  exist, are continuous on $\Sigma  $ and are related by
			$$
			M_{+}(x, t; k)=M_{-}(x, t; k) v(x, t; k), \quad k \in \Sigma,
			$$
			where $v(x, t; k)=v_n(x, t; k)~(n=1,2,\cdots,6)$ are defined in terms of $r_{1}(k)$ and $r_{2}(k)$ by $(\ref{vn-jump-matrix})$.
			
			(c)The matrix-valued functions $M_{n}(x;k)$ follow  the symmetries
			\begin{equation}\label{M symmeties}
				M_n(x;k)=\mathcal{A} M_n(x; \omega k) \mathcal{A}^{-1}={\mathcal{B} M_n^*(x;k^*)} \mathcal{B}.
			\end{equation}
			
			(d) $M(x, t; k)=I+\frac{M_{\infty}^{(1)}(x,t)}{k}+\mathcal{O}\left(k^{-2}\right)$ as $k \rightarrow \infty,~ k \notin \Sigma $, with
			$$
			M_{\infty}^{(1)}(x,t)=\int_x^{\infty}2u(y,t)\mathrm{d}y\left(\begin{array}{ccc}
				\omega^2  & 0 & 0\\
				0 & \omega& 0\\
				0 & 0 & 1\\
			\end{array}
			\right).
			$$
			
			(e) $M(x, t; k)=\sum_{l=-1}^{p} {M_{0}}^{(l)}(x,t) k^{l}+\mathcal{O}(k^{p+1})$ as $k \rightarrow 0$ with
			$$
			M_0^{(-1)}(x,t)={a_+(x,t)}\begin{pmatrix}
				\omega^2 & 0 & 0\\
				\omega^2 & 0 &0\\
				\omega^2 & 0 &0
			\end{pmatrix},
			$$
			where $a_+(x,t)$ is a real valued function and rapidly decreases as $x\to\infty$.
		\end{RHproblem}
		\begin{theorem}\label{thm for sk}
			Suppose the initial data \(u_0(x) \in \mathcal{S}(\mathbb{R})\) and the scattering data satisfy Assumption \ref{solitonless} and Assumption \ref{r less 1}. Define the reflection coefficients \(r_1(k)\) and \(r_2(k)\) with respect to \(u_0(x)\) as per (\ref{r1r2def}). Then it is established that the RH problem \ref{rhp SK} admits a unique solution \(M(x, t; k)\) whenever it exists, for each point in the domain \((x, t) \in \mathbb{R} \times [0, T)\). Furthermore, the solution \(u(x, t)\) of the SK equation (\ref{SK}) for all \((x, t) \in \mathbb{R} \times [0, T)\) can be expressed by
			\begin{equation}\label{recover-formula SK}
				u(x, t)=-\frac{1}{2} \frac{\partial}{\partial x} \lim _{k \rightarrow \infty} k\left(M_{33}(x, t; k)-1\right).
			\end{equation}
		\end{theorem}
		\begin{proof}
			The proof follows a standard approach. Please refer to \cite{Charlier-Lenells-2021} or \cite{procA} for details.
		\end{proof}
		The Miura transformation \cite{Fordy-Gibbons-1980} establishes a connection between the SK equation (\ref{SK}) and its modified version. In fact, this transformation can be derived directly from their corresponding RH problems.
		\begin{theorem}\label{Miura theorem}
			Assume the reflection coefficients $r_1(k)$ and $r_2(k)$ satisfy the Theorem \ref{prop for r}. Suppose that the $3\times3$ jump matrices $v_n(x, t; k)$ are formulated in (\ref{vn-jump-matrix})
			in terms of $r_1(k)$ and $r_2(k)$. For $x\in\R$ and $t\in[0,\infty)$, the solutions $M(x, t; k)$ and $m(x, t; k)$ for RH problems of the SK equation (\ref{SK}) and the modified SK equation (\ref{msk-equation}) satisfy the following correspondence. Moreover, the Miura transformation between the SK equation (\ref{SK}) and the modified SK equation (\ref{msk-equation}) is reconstructed. They are
			\par

			\begin{enumerate}[(a)]
				\item Define $A(x,t)$ as
				\begin{equation}\label{miura-A}
					A(x,t)=-\frac{w(x,t)}{3}\begin{pmatrix}
						\omega^2 & \omega  & 1\\
						\omega^2 & \omega  & 1\\
						\omega^2 & \omega  & 1\\
					\end{pmatrix},
				\end{equation}	
				then the $3\times3$ matrix-valued function $M(x,t;k)$  defined by
				\begin{equation}\label{miura-relation}
					M(x,t;k)=\left(I+\frac{A(x,t)}{k}\right)m(x,t;k),
				\end{equation}
				solves the RH problem \ref{rhp SK} for the SK equation (\ref{SK}).
				
				\item The solutions $u(x,t)$ and $w(x,t)$ of the SK equation (\ref{SK}) and the mSK equation (\ref{msk-equation}) are related by the Miura transformation for $x\in \mathbb{R},0\le t<\infty$, that is
				$$
				u(x,t)=\frac{1}{6}(w_x(x,t)-w(x,t)^2).
				$$	
			\end{enumerate}
		\end{theorem}
		\begin{proof}
			See Section \ref{Miura rhp}.
		\end{proof}
		\begin{theorem}\label{main-theorem-SK}
			Let \(u_0(x) \in \mathcal{S}(\mathbb{R})\) satisfy the assumptions in Theorem \ref{thm for sk}. Then, the solution \(u(x,t)\) of the initial value problem for the SK equation \eqref{SK} exhibits the following asymptotic behaviors as \((x,t) \to \infty\) in the $(x,t)$-half plane (see Figure \ref{categoryforsk}):
			\par
			\noindent {\bf Sector \rm I:}
			$
			u(x, t)=\frac{A_1(\xi)}{\sqrt{t}} \sin \alpha_1(\xi, t)+\frac{A_2(\xi)}{\sqrt{t}} \sin \alpha_2(\xi, t)+
			$
			{$
				\mathcal{O}\left(\frac{1}{x^N}+\frac{\mathrm{C}_N(\xi)}{x}\right)$}
			$,\quad M\le\xi<\infty;
			$
			\par
			\noindent {\bf Sector \rm II:}
			$
			u(x, t)=\frac{A_1(\xi)}{\sqrt{t}} \sin \alpha_1(\xi, t)+\frac{A_2(\xi)}{\sqrt{t}} \sin \alpha_2(\xi, t)+\mathcal{O}\left(\frac{\log t}{t}\right),\quad \frac{1}{M}\le\xi\le M;
			$
			\par
			\noindent {\bf Sector \rm III:}
			This is a transition region that arises due to $|r_j(0)|=1$ for $j=1,2$.
			\par
			\noindent {\bf Sector \rm IV:} This is Painlev\'{e} region and the leading-order term is
			$u(x,t)\sim \frac{1}{6}(5t)^{-\frac{2}{5}}\left(p'(s)-p^2(s)\right)$ with $s=\frac{x}{(5t)^{\frac{1}{5}}}$ and $p(s)$ solves the fourth-order analogues of the Painlev\'{e} transcendent (\ref{Painleve_equation_first}), $\quad|\xi|\le {M}{t^{-\frac{4}{5}}};$
			\par
			\noindent {\bf Sector \rm V:}
			{$
				u(x,t)=\mathcal{O}((|x|+t)^{-j})~j\geq 1,\quad \frac{1}{M}\le|\xi|,x<0.
				$}
			\par
			\noindent Here $\xi=\frac{x}{t}$,  $M>1$, and $\mathrm{C}_N(\xi)$ is rapidly decreasing as $\xi\to\infty$ for each $N\in \Z_+$. Moreover,
			\[
			\begin{aligned}
				&A_1(\xi):=-\frac{\sqrt{\nu_1}}{3^{\frac{3}{4}} 2 \sqrt{5k_0 } },\  A_2(\xi):=-\frac{\sqrt{\nu_4}}{3^{\frac{3}{4}} 2 \sqrt{5k_0 } },\\
				&\alpha_1(\xi, t):=\frac{19 \pi}{12}-\left(\arg r_1(k_0)+\arg \Gamma\left(i \nu_1\right)\right)-\left(36 \sqrt{3} t k_0^5\right)+\nu_1 \ln \left(3^{\frac{7}{2}} 20 t k_0^5\right)+s_1,\\
				&\alpha_2(\xi, t):=\frac{11 \pi}{12}-\left(\arg r_2(-k_0)+\arg \Gamma\left(i \nu_4\right)\right)-\left(36 \sqrt{3} t k_0^5\right)+\nu_4 \ln \left(3^{\frac{7}{2}} 20 t k_0^5\right)+s_2,
			\end{aligned}
			\]
			with $k_0=\sqrt[4]{\xi/45}$, $\Gamma(k)$ denotes the Gamma function, and	
			\[
			\begin{aligned}
				&\nu_1:=-\frac{1}{2 \pi} \ln \left(1-\left|r_1\left(k_0\right)\right|^2\right),\ \nu_4=-\frac{1}{2 \pi} \ln \left(1-\left|r_2\left(-k_0\right)\right|^2\right),\\
				&s_1=\nu_4 \ln (4)+\frac{1}{\pi} \int_{-k_0}^{-\infty} \log _\pi \frac{\left|s-\omega k_0\right|}{\left|s-k_0\right|} \mathrm{d}\ln \left(1-\left|r_2(s)\right|^2\right)+\frac{1}{\pi} \int_{k_0}^{\infty} \log _0 \frac{\left|s-k_0\right|}{\left|s-\omega k_0\right|} \mathrm{d}\ln \left(1-\left|r_1(s)\right|^2\right),\\
				&s_2=\nu_1 \ln (4)+\frac{1}{\pi} \int_{k_0}^{\infty} \log _0 \frac{\left|s+\omega k_0\right|}{\left|s+k_0\right|} \mathrm{d}\ln \left(1-\left|r_1(s)\right|^2\right)+\frac{1}{\pi} \int_{-k_0}^{-\infty} \log _\pi \frac{\left|s+k_0\right|}{\left|s+\omega k_0\right|} \mathrm{d}\ln \left(1-\left|r_2(s)\right|^2\right).
			\end{aligned}
			\]
			{Furthermore, the formula in Sector II holds uniformly with respect to $\xi=x / t$ in compact subset of the stated interval.}
		\end{theorem}
        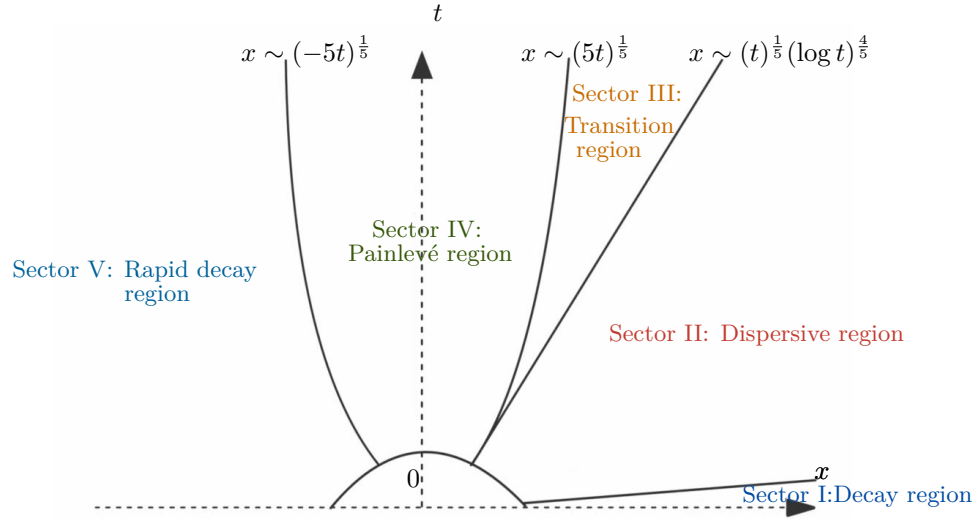
\begin{figure}[!ht]
				\centering
				\begin{overpic}[width=0.7\textwidth]{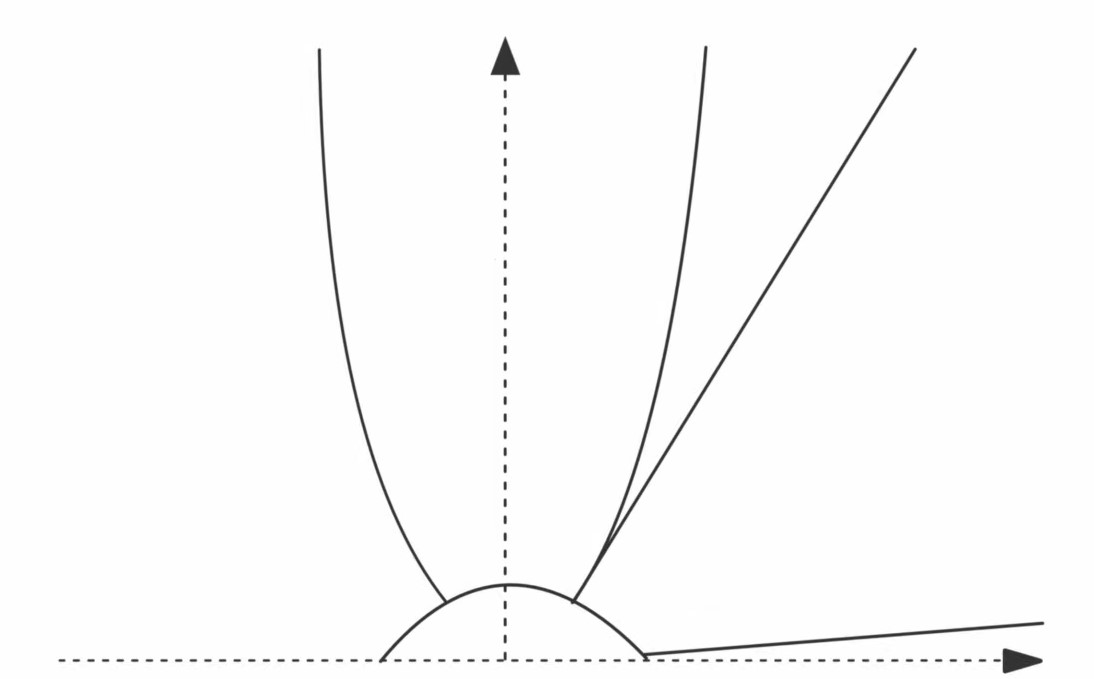}
					\put(69.6,22.5){\textcolor[rgb]{0.76,0.23,0.19}{\small Sector $\rm II$: Dispersive region }}
					\put(86.3,2.3){\textcolor[rgb]{0.0,0.29,0.63}{\small Sector $\rm I$:Decay region}}
					\put(65,52.5){\textcolor[rgb]{0.78,0.42,0}{\small Sector $\rm III$:}}
					\put(64,48.5){\textcolor[rgb]{0.78,0.42,0}{\small Transition}}
                    \put(65.5,45.5){\textcolor[rgb]{0.78,0.42,0}{\small region}}
					\put(40,35.5){\textcolor[rgb]{0.21,0.34,0.04}{\small Sector $\rm IV$:}}
					\put(37,32.5){\textcolor[rgb]{0.21,0.34,0.04}{\small Painlev${\rm\acute{e}}$ region}}
					\put(-5,30.5){\textcolor[rgb]{0.0,0.38,0.6}{\small Sector $\rm V$: Rapid decay}}
                    \put(9,27.5){\textcolor[rgb]{0.0,0.38,0.6}{\small  region}}
					\put(58.6,57.5){\bf $x\sim (5t)^{\frac{1}{5}}$}
					\put(23.6,57.5){\bf $x\sim (-5t)^{\frac{1}{5}}$}
					\put(79.6,57.5){\bf ${x}\sim{(t)^{\frac{1}{5}}(\log t)^{\frac{4}{5}}}$}
					\put(95.3,5.2){\bf $x$}
					\put(95.3,5.2){\bf $x$}
					\put(44.3,4.2){\bf $0$}
					\put(47.6,62.5){\bf $t$}
				\end{overpic}
				\caption{{\protect\small
						The asymptotic sectors $\rm I$-$\rm V$ in the $(x,t)$-half plane.}}
				\label{categoryforsk}
			\end{figure}
		\begin{proof}
			The proof of Sectors $\rm I$ and $\rm II$ is illustrated in Section \ref{Sector I II}, the proof of Sector $\rm IV$ is provided in Section \ref{painleve region}, and the proof of Sector $\rm V$ is detailed in Section \ref{rapid region}.
		\end{proof}
		\begin{remark}
			Similar to the KdV equation which generically has $r(0)=-1$ (see \cite{Deift-Zhou-1994}), it is conjectured that the SK equation (\ref{SK}) features a transition region referred to as the ``collisionless shock region", which serves as a bridge between the dispersive wave region and Painlev\'{e} region. The occurrence of this phenomenon stems from the fact that $|r_j(0)|=1$ for $j=1,2$. However, delving into the analysis of this region lies beyond the scope of the present work and surpasses the expertise of the authors.
		\end{remark}

\subsection{Numerical results}
\subsubsection{\bf Numerical verifications for the modified SK equation (\ref{msk-equation})} To demonstrate the validity of Theorem \ref{msk Thm}, the following initial condition in term of Gaussian wave pattern is specified as: 
		\begin{equation}\label{initial-msk-equation}
			w(x,0)=w_0(x)=-\frac{1}{10}\mathrm{e}^{-\frac{x^2}{20}}.
		\end{equation}
		This ensures that the reflection coefficients comply with the Assumption 3.5 in Ref. \cite{procA}, which requires that $\tilde{s}_{11}(k) \neq 0$ and $\tilde{s}_{11}^A(k) \neq 0$ for all \(k \in \bar{\Omega}_1 \setminus \{0\}\) and \(k \in \bar{\Omega}_4 \setminus \{0\}\), respectively.
		\par		
		Figures \ref{fig:mskt50} and \ref{fig:mskt100} show the comparison between the leading-order terms of asymptotic solutions given in Theorem \ref{msk Thm} and the results obtained by numerical simulations with the initial condition specified in \eqref{initial-msk-equation} at times $t = 50$ and $t = 100$, respectively. In these figures, the leading-order term in Sector $\rm I$ and $\rm II$, i.e., $\frac{\tilde{A}_1(\xi)}{\sqrt{t}} \cos \tilde{\alpha}_1(\xi, t)+\frac{\tilde{A}_2(\xi)}{\sqrt{t}} \cos \tilde{\alpha}_2(\xi, t)$ is depicted with dashed red lines, while the numerical results are shown with solid blue lines. On the other hand, the dashed purple line illustrates the numerical result for the fourth-order analogues of the Painlev\'{e} transcendent \cite{Cosgrove-2000,Cosgrove-2006} in (\ref{Painleve_equation_first}).  These visual comparisons demonstrate that the large-time asymptotic solutions closely approximate the numerical results, which validates the accuracy of the asymptotic predictions in Theorem \ref{msk Thm}.
\begin{figure}[h!]
    \centering
    \subfigure{
        \includegraphics[width=12cm]{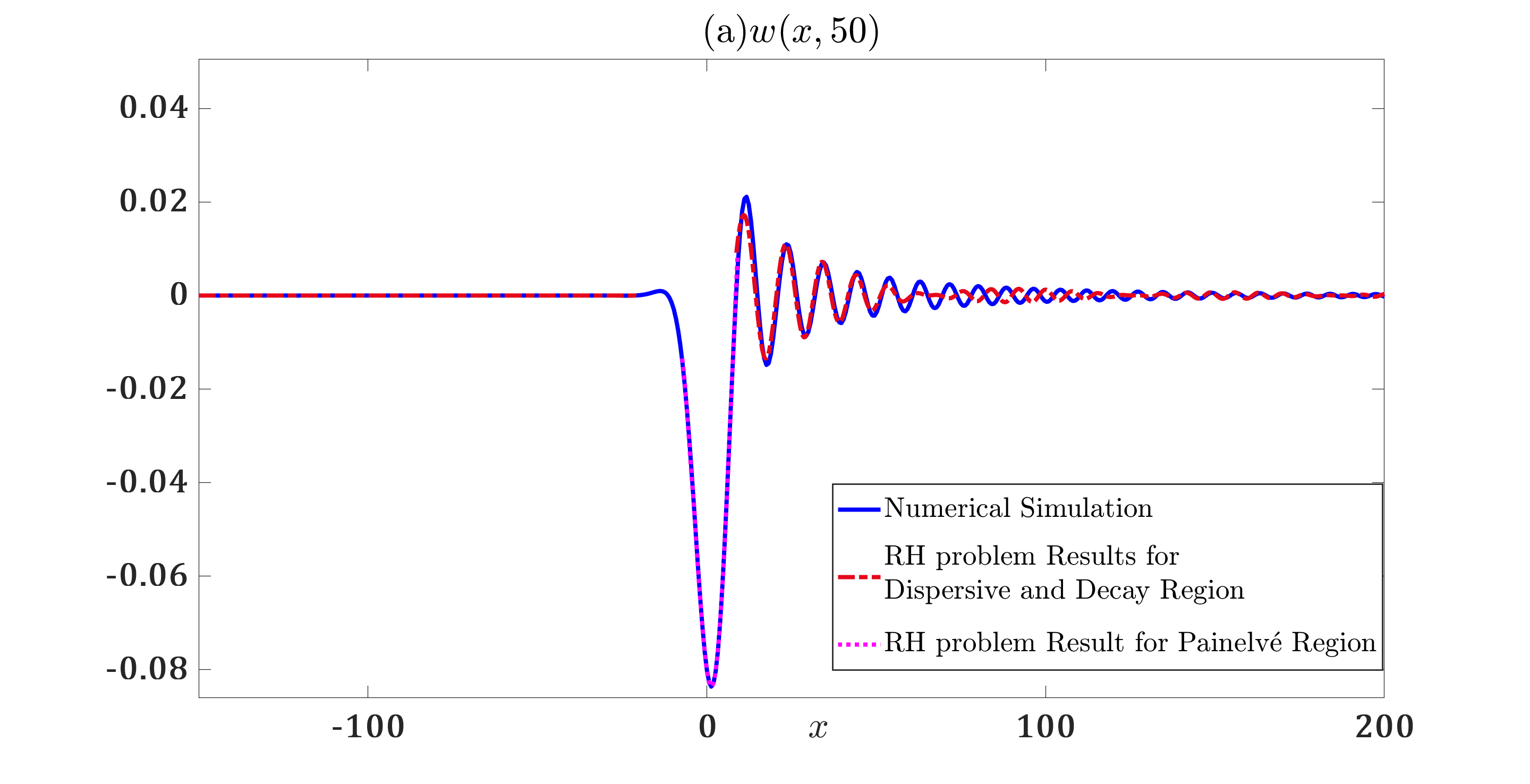}
        \label{fig:mskt50}
    }
    \hfill
    \subfigure{
        \includegraphics[width=12cm]{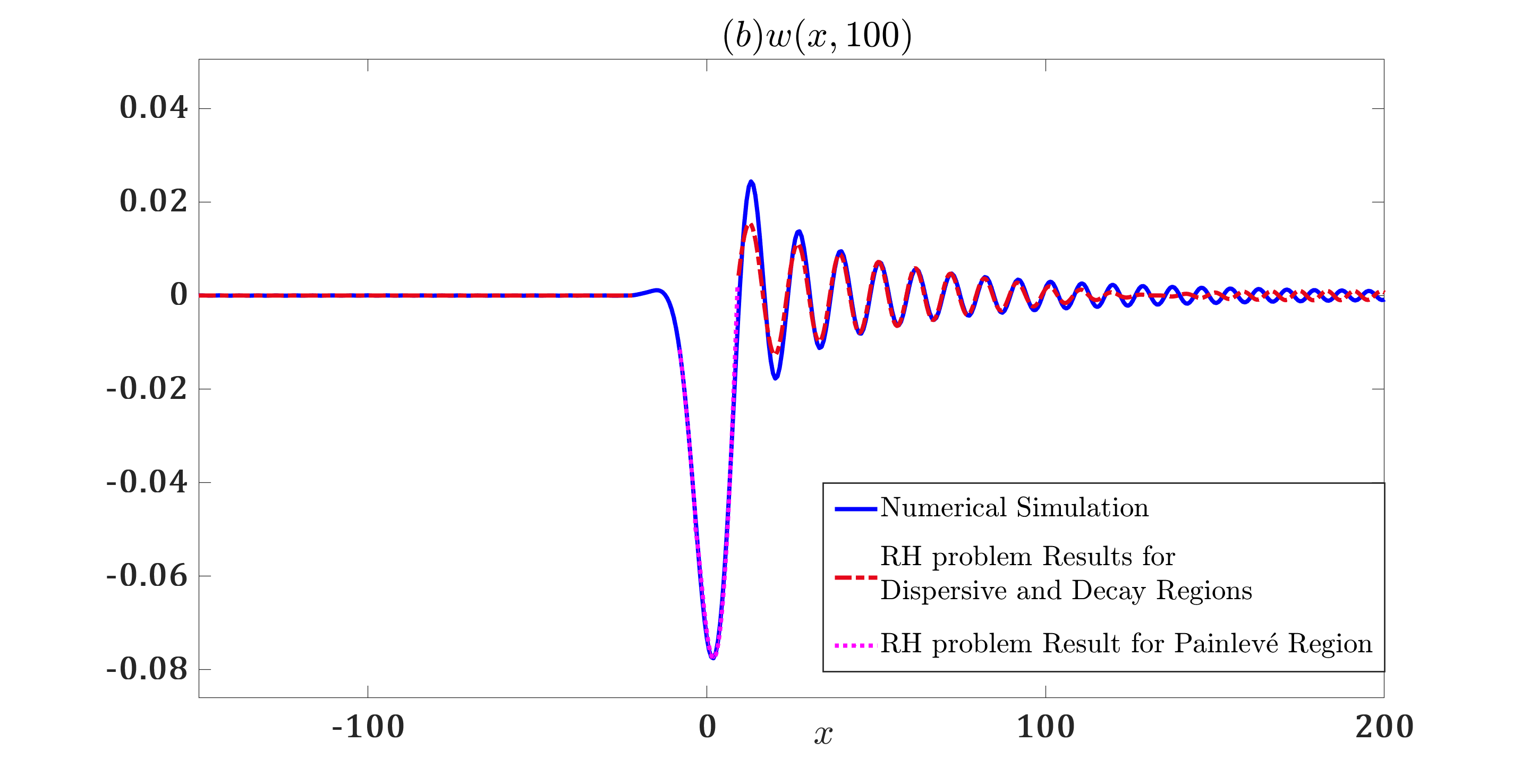}
        \label{fig:mskt100}
    }
\caption{{\protect\small 
Comparison of the leading-order asymptotic approximations in Theorem~\ref{msk Thm} with direct numerical simulations of the modified SK equation (\ref{msk-equation}) under the initial Gaussian wave packet (\ref{initial-msk-equation}) at different times. 
{{The observed increase of the error for} {small positive $x$ reflects transition effects between the Painlev\'e and dispersive regimes.}}}
}
    \label{msk-comparisons}
\end{figure}
{\begin{remark} It is observed from Figure \ref{msk-comparisons} that, when moving from the Painlev\'{e} region into the dispersive region, there is a small region where the error increases. This is related to the asymptotics of the fourth-order analogues of Painlev\'{e} transcendent \eqref{4th-Painleve}. Since the asymptotic behavior of the equation \eqref{4th-Painleve} is not yet fully understood, the detailed description of the transition from the Painlev\'{e} region to the dispersive region remains to be clarified in the future work.
\end{remark}}
		\par

		\subsubsection{\bf Numerical results for SK equation}
		Similarly, take the initial-value condition of the form
		\begin{equation}\label{initial-SK-equation}
			u(x,0)=u_0(x)=\frac{1}{600}(x\mathrm{e}^{-\frac{x^2}{20}}-\mathrm{e}^{-\frac{x^2}{10}}).
		\end{equation}
		This choice of initial condition ensures that the reflection coefficients comply with Assumption \ref{solitonless} and Assumption \ref{r less 1}.
		\par
		Figure \ref{SK-comparisons} demonstrates the evolutions of the solution $u(x, t)$ to the SK equation (\ref{SK}) with initial data (\ref{initial-SK-equation}) at time $t= 50$ and $t=100$ by two different ways, where the dashed red line shows the leading-order asymptotics from the Riemann-Hilbert formulation and the solid blue line shows the wave profile obtained by numerical simulation. The convergence is weak for small values of $x$, which is consistent with the fact that the asymptotic estimate (\ref{longtime-solution}) is not uniform near $x=0$.
		\begin{figure}[h!]
			\centering
			\includegraphics[width=12cm]{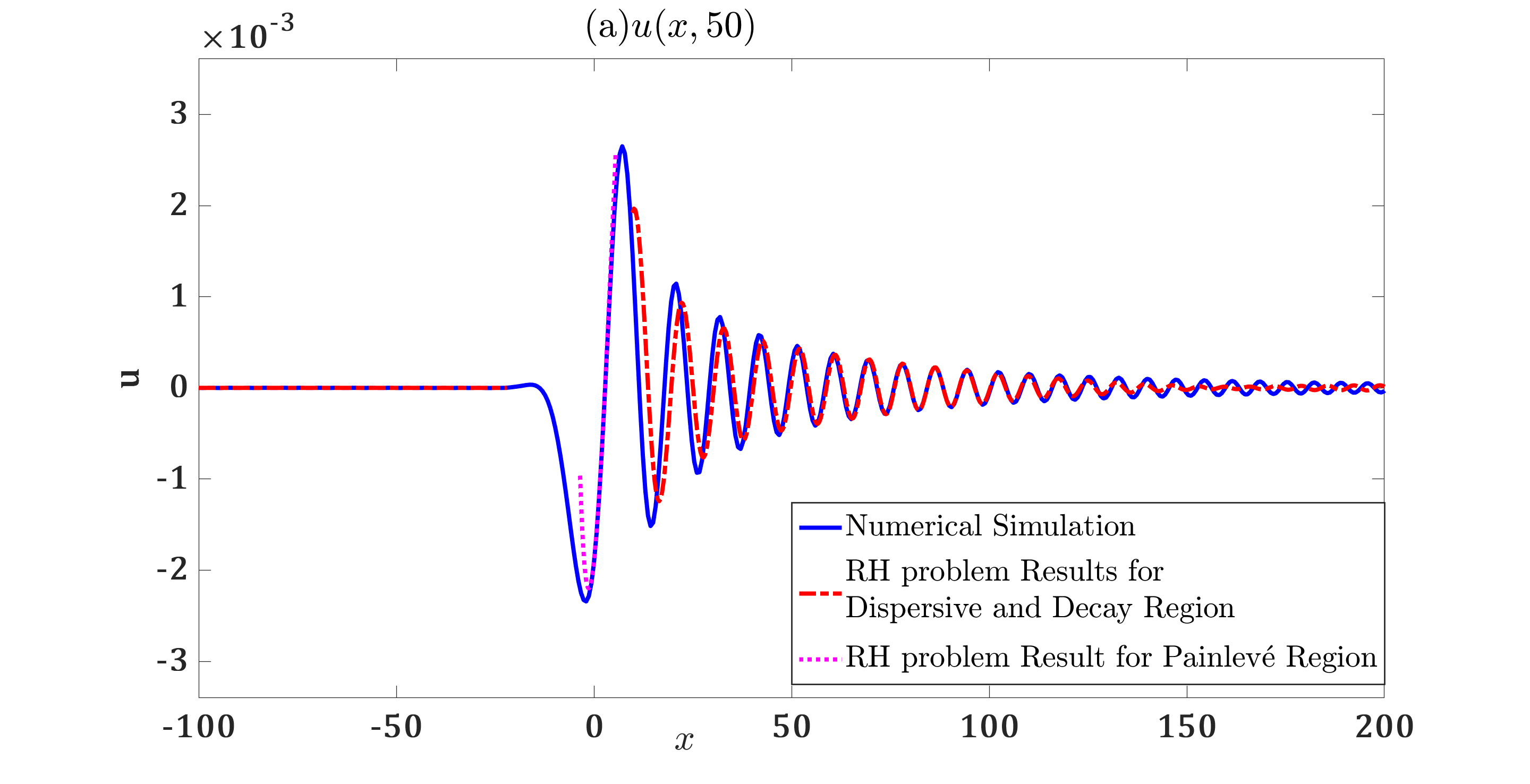}
			\includegraphics[width=12cm]{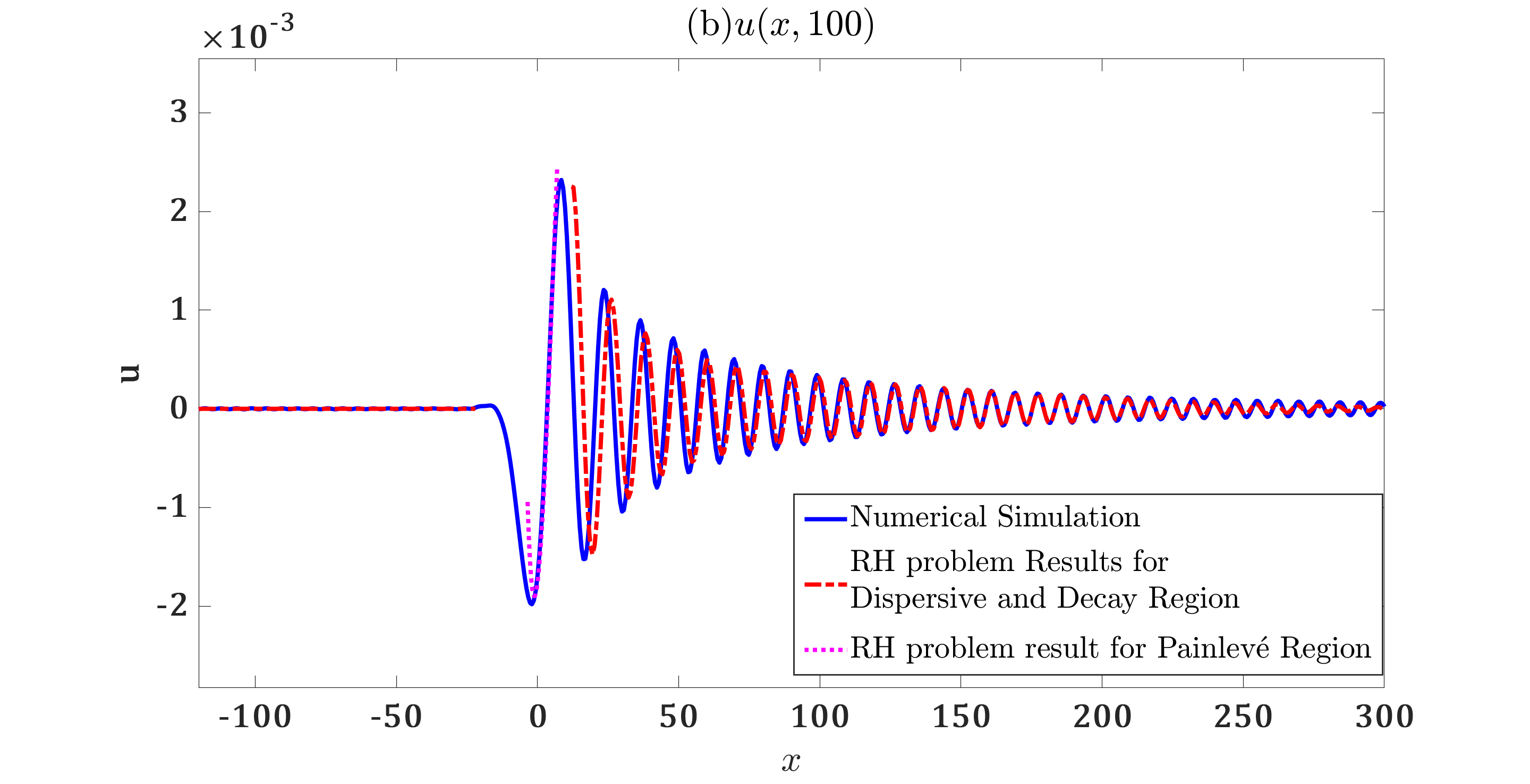}
            \caption{{\protect\small 
Comparison of the leading-order asymptotic approximation from the RH problem and direct numerical simulations of the SK equation (\ref{SK}) with initial data (\ref{initial-SK-equation}) at $t=50$ and $t=100$. 
{{A transition region between the Painlev\'{e} and dispersive regimes} {is observed in the numerical comparison plots.}}
}}

			\label{SK-comparisons}
		\end{figure}

\section{\bf The Riemman-Hilbert problem and Miura transformation}\label{RHPandmiura}
This section performs the direct and inverse scattering transforms \cite{GGKM-1967} to formulate the RH problems associated with the SK equation (\ref{SK}) and the modified SK equation (\ref{msk-equation}), and reconstructs the Miura transformation between the two equations based on the relationship of the RH problems.
	\par	
	Introduce the gauge transformation
	\begin{equation}\label{gauge-matrix}
		\Phi(x,t;k)=G(k)\Psi(x,t;k) \quad {\rm with}\quad G(k)=\left(
		\begin{array}{ccc}
			{\omega} & {\omega^2} & {1} \\
			{\omega^2}{ k} & {\omega}{ k} & {k} \\
			k^2 & k^2 & k^2 \\
		\end{array}
		\right),\quad \omega=\mathrm{e}^{\frac{2\pi i}{3}},
	\end{equation}
	then the space part of spectral problem (\ref{SK-lax-pair}) with (\ref{SKlaxspace}) is converted into
	\begin{equation}\label{lax-pair-new-x-part}
		\Psi_x=\mathcal{L} \Psi,
	\end{equation}
	where $\mathcal{L}=G^{-1}LG=k\Lambda+Q(x,t;k)$ with
	$$
	\Lambda=\left(\begin{array}{ccc}
		\omega & 0 & 0\\
		0 & \omega^2 & 0\\
		0 & 0 & 1\\
	\end{array}
	\right),\quad
	Q(x,t;k)=-\frac{2u(x,t)}{k}\left(\begin{array}{ccc}
		\omega^2 & \omega & 1\\
		\omega^2 & \omega & 1\\
		\omega^2 & \omega & 1\\
	\end{array}
	\right):=\frac{Q_1}{k}.
	$$
	\begin{remark}
		Similarly, one can take the same gauge transformation on the Lax pair of the KK equation (\ref{KK}). However, in the new spectral problem, which is denoted as $\widetilde{\Psi}_x=\widetilde{\mathcal{L}} \widetilde{\Psi}$, a second-order pole emerges at $k=0$, i.e., $\widetilde{\mathcal L}=k\Lambda+\frac{\widetilde{Q}_1}{k}+\frac{\widetilde{Q}_2}{k^2}$ with $\mathop{{\rm lim}}\limits_{|x|\to \infty}\widetilde{Q}_1=\mathop{{\rm lim}}\limits_{|x|\to \infty}\widetilde{Q}_2=0$, in contrast to the simple pole in the case of the SK equation. This difference arises from the discrepancy between $q_0=0$ (for the SK equation (\ref{SK})) and $q_0=q_x+p$ (for the KK equation (\ref{KK})) in the third-order operator $\mathscr{L}=d^3/dx^3+q_1d/dx+q_0$.
	\end{remark}
	
	\par
	The gauge transformation (\ref{gauge-matrix}) can map the temporal part of spectral problem (\ref{SK-lax-pair}) with (\ref{SKlaxtime}) into
	\begin{equation}\label{lax-pair-new-t-part}
		\Psi_t=\mathcal{Z} \Psi,
	\end{equation}
	where $\mathcal{Z}=G^{-1}ZG:=9k^5\Lambda^2+P(x,t;k)$     
    with $P(x,t;k)\to0$ as $|x|\to \infty$.
	\par
	Thus the gauge transformation (\ref{gauge-matrix}) transforms the Lax pairs (\ref{SK-lax-pair}) into
	\begin{equation}\label{New-Lax-SK-KK}
		\left\{\begin{array}{l}
			\Psi_x=(k\Lambda+Q) \Psi, \\
			\Psi_t=(9k^5\Lambda^2+P)\Psi.
		\end{array}\right.
	\end{equation}
	Furthermore, the transformation $\Psi=J \mathrm{e}^{(k\Lambda x+9k^5\Lambda^2 t)}$ indicates that
	\begin{equation}\label{Lax-equation}
		\left\{\begin{array}{l}
			J_x-[k\Lambda,J]=Q J, \\
			J_t-[9k^5\Lambda^2,J]=P J. \\
		\end{array}\right.
	\end{equation}
	\par
	In what follows, we only focus on the $x$-variable and take $t$-variable as a dump variable. According to the equation $J_x-[k\Lambda,J]=Q J$, the Volterra integral equations of the Jost functions $J_{+}(x; k)$ and $J_{-}(x; k)$ are written as
	\begin{equation}\label{Jost-functions}
		\begin{aligned}
			& J_{+}(x;k)=I-\int_x^{\infty} \mathrm{e}^{(x-y) \widehat{k\Lambda}}\left(Q(y;k) J_{+}(y;k)\right) d y, \\
			& J_{-}(x;k)=I+\int_{-\infty}^x \mathrm{e}^{(x-y) \widehat{k\Lambda}}\left(Q(y;k) J_{-}(y;k)\right) d y,
		\end{aligned}
	\end{equation}
	which show that the singular set is
	$$
	\Sigma:=\{k\in\mathbb{C}|{\rm Re}(\omega^nk)={\rm Re}(\omega^m k),\quad 0\le n<m<3\},
	$$
	which divides the complex plane into six parts (see Figure \ref{sigma0}), i.e.,
	$$
	\Omega_n:=\left\{k\in\mathbb{C}\bigg|\frac{(n-1)\pi}{3}<{\rm arg}(k)<\frac{n\pi}{3},n=1,2,\cdots,6\right\}.
	$$
	\begin{figure}[!h]
		\centering
		\begin{overpic}[width=.55\textwidth]{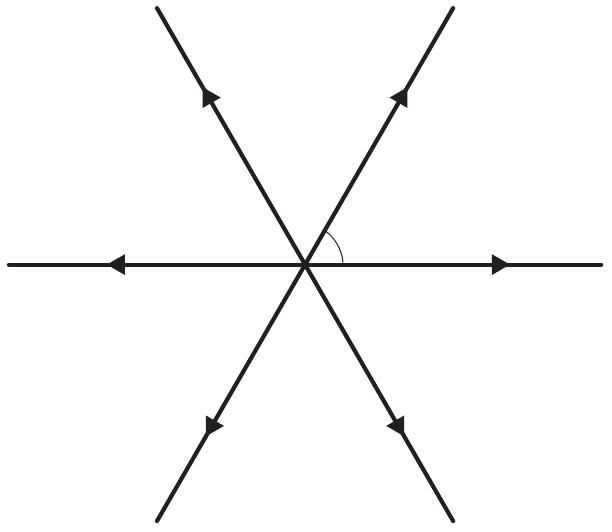}
			\put(57,48){$\frac{\pi}{3}$}
			\put(103,42){$\Sigma$}
			\put(75,55){${\Omega}_1$}
			\put(48,70){${\Omega}_2$}
			\put(23,55){${\Omega}_3$}
			\put(23,25){${\Omega}_4$}
			\put(48,15){${\Omega}_5$}
			\put(75,25){${\Omega}_6$}
		\end{overpic}
		\caption{{\protect\small
				The contour $\Sigma$ decomposes the complex $k$-plane into six parts: $\Omega_n~(n=1,2,\cdots,6)$.}}
		\label{sigma0}
	\end{figure}
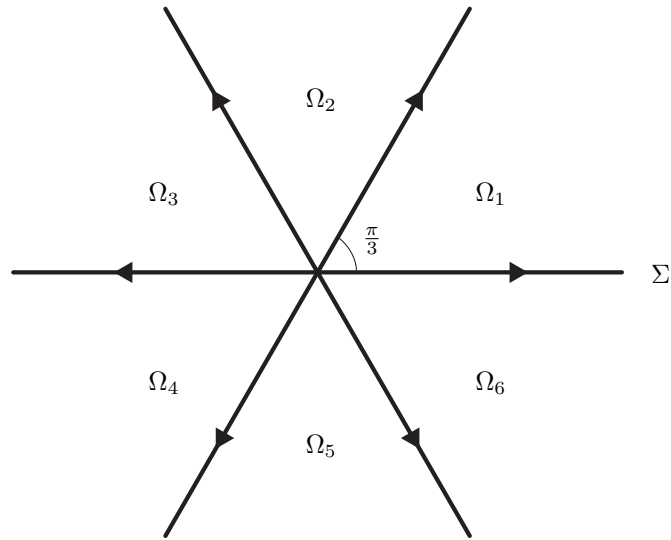
	
	The following way to construct the RH problem \cite{Boutet-de-Monvel-1,Beals-Coifman-1984}\cite{Lenells-2018}
	is standard, so the proofs of the propositions below are omitted, see \cite{Lin-Wang-Zhu} for details.

	\subsection{Basic properties of the Jost functions}
	
	\begin{proposition} Let $\mathcal{S}(\R)$ be the Schwartz space.
		Suppose that the SK equation (\ref{SK}) has initial value $u_0(x)\in\mathcal{S}(\R)$ and denote $S=\Omega_3\cup\Omega_4$, then the matrix-valued Jost functions $J_+(x;k)$ and $J_-(x;k)$ possess the following properties:
		
		(1). $J_+(x;k)$ is well defined in the closure of  $(\omega^2S,\omega S,S)\setminus\{0\}$, and $J_-(x;k)$ is well defined in the closure of  $(-\omega^2S,-\omega S,-S)\setminus\{0\}$. Moreover, the determinants of $J_{\pm}(x;k)$ are always equal to $1$.
		
		(2). $J_+(\cdot;k)$ and $J_-(\cdot;k)$ are smooth and rapidly decay in the closure of their domains $($except for $\{0\}$$)$.
		
		(3). $J_+(x;\cdot)$ and $J_-(x;\cdot)$ are analytic in the interior of  their domains, but any order partial derivative of $k$ can be continuous to the closure of their domains $($except for $\{0\}$$)$.
		
		(4). The functions $J_+(x;k)$ and $J_-(x;k)$ satisfy the following symmetries:
		$$
		\begin{aligned}
			& J_+(x;k)=\mathcal{A} J_+(x; \omega k) \mathcal{A}^{-1}=\mathcal{B} {J_+^*(x; {k}^*)} \mathcal{B},\\
			& J_-(x;k)=\mathcal{A} J_-(x; \omega k) \mathcal{A}^{-1}=\mathcal{B} {J_-^*(x; {k}^*)} \mathcal{B},
		\end{aligned}
		$$
		where $k$ is located in their domains and the matrices $\mathcal{A}$ and $\mathcal{B}$ are given in (\ref{AB-Matrices}).
		
		(5). When $u_0(x)$ is compact support, $J_+(x;k)$ and $J_-(x;k)$ are well defined and analytic for $k\in\C\setminus\{0\}$. 	
	\end{proposition}
	
	\noindent{\bf The behavior of Jost functions for $k\to\infty$.}
	Let the WKB expansions of the Jost functions $J_{\pm}(x;k)$ be
	$$
	J_{\pm}(x;k)=I+\frac{J_{\pm}^{(1)}}{k}+\frac{J_{\pm}^{(2)}}{k^2}+\cdots.
	$$
	Taking into account of the equation ({\ref{Lax-equation}}), one has
	\begin{equation}\label{WkB-infty}
		\left\{\begin{array}{l}
			{\left[\Lambda, J_{\pm}^{(n+1)}\right]=(\partial_x J_{\pm}^{(n)})^{(o)}-\left(\mathrm{Q}_1 J_{\pm}^{(n-1)}\right)^{(o)}}, \\
			(\partial_x J_{\pm}^{(n+1)})^{(d)}=\left(\mathrm{Q}_1 J_{\pm}^{(n)}\right)^{(d)},
		\end{array}\right.
	\end{equation}
	with $Q_1(x,t;k)={-2u}\left(\begin{array}{ccc}
		\omega^2 & \omega & 1\\
		\omega^2 & \omega & 1\\
		\omega^2 & \omega & 1\\
	\end{array}
	\right)$, in which the notation $(o)$ means the off-diagonal part of the matrix and $(d)$ denotes the diagonal part. 
	Furthermore, the other expansion coefficients are
	$$
	J_+^{(1)}=\int_x^{\infty}2u(y)dy\left(\begin{array}{ccc}
		\omega^2  & 0 & 0\\
		0 & \omega& 0\\
		0 & 0 & 1\\
	\end{array}
	\right),
	$$
	\begin{equation}\label{J1-infty}
		J_+^{(2)}=\int_x^{\infty}{2u(y)(J_{1+})_{33}}dy\left(\begin{array}{ccc}
			\omega  & 0 & 0\\
			0 & \omega^2 & 0\\
			0 & 0 & 1\\
		\end{array}
		\right)+\frac{2u(x)}{1-\omega}\left(\begin{array}{ccc}
			0  & 1 & -1\\
			-\omega & 0 & \omega\\
			\omega^2 & -\omega^2 & 0\\
		\end{array}
		\right).
	\end{equation}
	
	\begin{proposition}
		Suppose $u_0(x)\in\mathcal{S}(\R)$, there exist bounded smooth functions $f_{\pm}(x)$, which rapidly decay as $x\to\infty$ and $x\to-\infty$, respectively. Letting $m\ge 0$ be an integer and for each integer $n\ge 0$, it follows
		$$
		\left|\frac{\partial^n}{\partial{k^n}}\left[ J_{\pm}-\left(I+\frac{J_{\pm}^{(1)}}{k}+\cdots+\frac{J_{\pm}^{(m)}}{k^m}\right)\right]\right|\le \frac{f_{\pm}(x)}{k^{m+1}},
		$$
		where $k$ is located in the domains of $J_{+}(x;k)$ and $J_{-}(x;k)$, respectively and is large enough.
	\end{proposition}

	\noindent{\bf The behavior of Jost functions for  $k\to0$.}
	
	Since the kernel matrix function $Q(x;k)$ has a simple pole at $k=0$, it is necessary to illustrate the asymptotics of $J_{\pm}(x;k)$ as $k\to 0$.
	
	\begin{proposition}
		Suppose $u_0(x)\in\mathcal{S}(\R)$, there exist bounded smooth functions $g_{\pm}(x)$, which rapidly decay as $x\to\infty$ and $x\to-\infty$, respectively. Let $m\ge 0$ be an integer and for each integer $n\ge 0$, then the Jost function $J_{\pm}(x;k)$ have the asymptotic expansions of the forms:
		$$
		\left|\frac{\partial^n}{\partial{k^n}}\left[ J_{\pm}(x;k)-\left(\frac{\mathcal{J}_{\pm}^{(-1)}}{k}+I+{\mathcal{J}_{\pm}^{(1)}}{k}+\cdots+{\mathcal{J}_{\pm}^{(m)}}{k^m}\right)\right]\right|\le {g_{\pm}(x)}{k^{m+1}},
		$$
		where $k$ is small enough.
		Furthermore, the terms $\mathcal{J}_{\pm}^{(-1)}$ are
		$$
		\mathcal{J}_{\pm}^{(-1)}(x)=a_{\pm}(x)\left(\begin{array}{ccc}
			\omega^2 & \omega & 1 \\
			\omega^2 & \omega & 1 \\
			\omega^2 & \omega & 1
		\end{array}\right),
		$$
		where $a_{\pm}(x)$ are real valued functions and are dominated by $g_{\pm}(x)$ with rapidly decay as $x\to \infty$ and $x\to-\infty$, respectively.
	\end{proposition}
	
	\subsection{The scattering matrix}
	
	Define the scattering matrix as
	\begin{equation}\label{Scattering-Ma}
		s(k)=I-\int_{\mathbb{R}} \mathrm{e}^{-xk \widehat{\Lambda}}(Q J)(x;k) d x.
	\end{equation}
	\par
	When the initial potential function $u_0(x)$ is compact support, the scattering matrix $s(k)$ satisfies
	$$
	J_+(x;k)=J_-(x;k) \mathrm{e}^{x k \widehat{\Lambda}} s(k), \quad k \in \mathbb{C}\setminus\{0\}.
	$$
	\begin{proposition}\label{sprop}
		Suppose $u_0(x)\in\mathcal{S}(\R)$, then the scattering function $s(k)$ defined in (\ref{Scattering-Ma}) has the following properties:
		
		(a) The domain of scattering matrix $s(k)$ is
		$$
		s(k) \in \left(\begin{array}{ccc}
			\omega^{2}\overline{S} & \mathbb{R}_{+} & \omega \mathbb{R}_{+} \\
			\mathbb{R}_{+} & \omega \overline{S} & \omega^{2} \mathbb{R}_{+} \\
			\omega \mathbb{R}_{+} & \omega^{2} \mathbb{R}_{+} & \overline{S}
		\end{array}\right)\setminus\{0\},
		$$
		where $\overline{S}$ means the closure of set $S$ and $s(k)$ is continuous to the boundary of domain but is analytic in the interior of its domain.
		
		(b) The matrix-valued function $s(k)$ has the following expansions as $k\to\infty$ and $k\to 0$, respectively, that are
		$$
		s(k)=I-\sum_{j=1}^{N} \frac{s_{j}}{k^{j}}+\mathcal{O}\left(\frac{1}{k^{N+1}}\right),\quad k \rightarrow \infty,
		$$
		and
		$$
		s(k)=\frac{s^{(-1)}}{k}+s^{(0)}+s^{(1)} k+\cdots, \quad k \rightarrow 0,
		$$
		with
		$$
		s^{(-1)}=\mathbf{s}^{(-1)}\begin{pmatrix}
			\omega^2 &\omega &1\\
			\omega^2 &\omega &1\\
			\omega^2 &\omega &1\\
		\end{pmatrix},
		$$
		where $\mathbf{s}^{(-1)}$ is a constant in form of integral about the potential function $u_0(x)$.
		
		(c) The matrix-valued function $s(k)$ satisfies the symmetries:
		$$
		s(k)=\mathcal{A} s(\omega k) \mathcal{A}^{-1}=\mathcal{B} s^*({k}^*) \mathcal{B}.
		$$
	\end{proposition}
	
	\noindent{
    \bf The cofactor Jost functions}
	Define $M^{A}=(M^{-1})^T$, then the adjoint equation associated with the equation $J_x-[k\Lambda,J]=Q J$ is
	\begin{equation}\label{cofactor-equation}
		\left(J^{A}\right)_{x}+\left[k\Lambda, J^{A}\right]=-Q^{T}J^{A}.
	\end{equation}	
	By the same procedure, one can also get the cofactor Jost functions $J_{\pm}^A(x;k)$ and cofactor scattering matrix $s^A(k)$. Furthermore, the properties of $J_{\pm}^A(x;k)$ and $s^A(k)$ can also be given similarly. Moreover, we have
	\begin{equation}\label{Scattering-Ma-A}
		s^A(k)=I+\int_{\mathbb{R}} \mathrm{e}^{-xk \widehat{\Lambda}}(Q^T J^A)(x;k) d x.
	\end{equation}
	
	\subsection{The eigenfunctions $M_n$ }
	
	Define the eigenfunctions for the equation (\ref{Lax-equation}) in each $k\in \Omega_n\setminus\{0\}~(n=1,2,\cdots,6)$ by the following Fredholm integral
	
	\begin{equation}\label{M_n-eigenfunction}
		\left(M_{n}\right)_{i j}(x;k)=\delta_{i j}+\int_{\gamma_{i j}^{n}}\left(\mathrm{e}^{\left(x-y\right) k\widehat{\Lambda}}\left(Q M_{n}\right)\left(y;k\right)\right)_{i j} d y, \quad i, j=1,2,3,
	\end{equation}
	where $\gamma_{ij}^n=(x,\infty)\ \text{or}\ (-\infty,x)$, which is determined by the exponential part and $\delta_{ij}$ is the Kronecker delta. Notice that there are zeros in Fredholm determinants on the complex plane that is denoted by $\mathcal{Z}$, which is a finite set. However, the solution of (\ref{M_n-eigenfunction}) can be analytic continuation to $\mathcal{Z}$.
	
	\begin{proposition}
		Suppose $u_0(x)\in\mathcal{S}(\R)$, then the integral equation (\ref{M_n-eigenfunction}) uniquely defines six $3 \times 3$ matrix-valued solutions $\left\{M_{n}\right\}_{n=1}^{6}$ of (\ref{Lax-equation}) with the following properties:
		
		(a) The eigenfunctions $M_{n}(x;k)$ are defined for $x \in \mathbb{R}$ and $k \in \bar{\Omega}_{n} \backslash{(\mathcal{Z}\cup\{0\})}$. Moreover, the functions $M_{n}(x;k)$ are smooth for $x\in\R$, continuous to $k \in \bar{\Omega}_{n} \backslash{(\mathcal{Z}\cup\{0\})}$ and analytic in the interior of its domain. Except for $k\in\mathcal{Z}\cup\{0\}$, the functions $M_n(x;k)$ are bounded.
		
		(b) The eigenfunctions $M_{n}(x;k)$ follow  the symmetries
		\begin{equation}\label{Mnsymmeties}
			M_n(x;k)=\mathcal{A} M_n(x; \omega k) \mathcal{A}^{-1}={\mathcal{B} M_n^*(x;k^*)} \mathcal{B},
		\end{equation}
		where $k\in\bar\Omega_k\setminus{(\mathcal{Z}\cup\{0\})}$.
		
		(c) The determinants of  eigenfunctions $M_{n}(x;k)$ identically equal to one for each $k\in\bar\Omega_k\setminus{(\mathcal{Z}\cup\{0\})}$.
	\end{proposition}

	\noindent{\bf The properties of eigenfunctions $M_n(x;k)$ as $k\to\infty$.}
	\begin{proposition}
		Suppose $u_0(x)\in\mathcal{S}(\R)$ and $u_0(x)$ is not identically equal to zero. Given an integer $m\ge 1$ and for $k$ large enough in its domain, the eigenfunction $M_n(x;k)$ can be approached by the expansion of $J_{+}(x;k)$ as
		\begin{equation}\label{M_n-infty}
			\left| M_{n}(x;k)-\left(I+\frac{J_{+}^{(1)}}{k}+\cdots+\frac{J_{+}^{(m)}}{k^m}\right)\right|\le \frac{C}{k^{m+1}},\quad C\in \R_{+}.
		\end{equation}
	\end{proposition}
	Now, assuming  $u_0(x)\in\mathcal{S}(\R)$ is compact support, then one can get the relationship between $M_n(x;k)$ and $J_{\pm}(x;k)$ for $k \in \bar{\Omega}_{n} \backslash \mathcal{Z}$ and $ x \in \mathbb{R}$ by
	\begin{equation}\label{Sn-Tn}
		\begin{aligned}
			M_{n}(x;k) &=J_-(x;k) \mathrm{e}^{x \widehat{\mathcal{L}(k})} S_{n}(k) \\
			&=J_+(x;k) \mathrm{e}^{x \widehat{\mathcal{L}(k})} T_{n}(k), \quad n=1, 2, \ldots, 6.
		\end{aligned}
	\end{equation}
	Combining the relationship between $J_+(x;k)$ and $J_-(x;k)$, the $S_n(k)$ and $T_n(k)$ can be linked by
	$$
	s(k)=S_{n}(k) T_{n}^{-1}(k), \quad k \in \bar{\Omega}_{n} \backslash (\mathcal{Z}\cup\{0\}) .
	$$
	Since the Schwartz functions with compact support are dense in $\mathcal{S}(\R)$ with respect to the $L^{\infty}$ norm, one can asymptotically express the functions $M_n(x;k)$, $J_{\pm}(x;k)$ and $s(k)$ under generically Schwartz initial potentials by the ones generated from potentials with compact support.

    \noindent{\bf The jump matrices $v_n(x;k)$}\begin{lemma}
		Suppose $u_0(x)\in\mathcal{S}(\R)$, then the matrix-valued functions $M_n(x;k)$ satisfies the boundary condition
		$$
		M_{+}(x;k)=M_{-}(x;k) v(x;  k), \quad k \in \Sigma \backslash(\mathcal{Z}\cup\{0\}),
		$$
		where $v(x;  k)$ is the jump matrix to be determined below.
		\par		
		In particular, when $u_0(x)\in\mathcal{S}(\R)$ is compact support, there exists a matrix $v_1(k)$ such that
		$$
		M_1(x;k)=M_6(x;k) \mathrm{e}^{k x \widehat{\Lambda}} v_1(k).
		$$
		One has $M_n(x;k)=\mathrm{e}^{x \widehat{\mathcal{L}(k)}} S_n(k)$ when $x$ is out of the support of $u_0(x)$ and $x\to-\infty$. Hence, the jump matrix $v_1(k)$ can be calculated by
		$$
		v_1(k)=S_6(k)^{-1} S_1(k).
		$$
	\end{lemma}
	By the same procedure, all the jump functions $v_n(k)~(n=1,2,\cdots,6)$ can be gotten.
	
	\begin{lemma}
		Let $u_{0}(x) \in \mathcal{S}(\mathbb{R})$, the eigenfunctions $M_{1}(x;k)$ can be expressed in terms of the entries of $J_{\pm}(x;k), J_{\pm}^{A}(x;k), s(k)$, and $s^{A}(k)$ as follows:
		$$
		M_{1}=\left(\begin{array}{clc}
			J^{+}_{11} & \frac{(J^{-}_{31})^{A} (J^{+}_{23})^{A}-(J^{-}_{21})^{A} (J^{+}_{33})^{A}}{s_{11}} & \frac{J^{-}_{13}}{s_{33}^{A}} \\
			J^{+}_{21} & \frac{(J^{-}_{11})^{A} (J^{+}_{33})^{A}-(J^{-}_{31})^{A} (J^{+}_{13})^{A}}{s_{11}} & \frac{J^{-}_{23}}{s_{33}^{A}} \\
			J^{+}_{31} & \frac{(J^{-}_{21})^{A} (J^{+}_{13})^{A}-(J^{-}_{11})^{A} (J^{+}_{23})^{A}}{s_{11}} & \frac{J^{-}_{33}}{s_{33}^{A}}
		\end{array}\right).
		$$
		Furthermore, for $|k|$ small enough, the following property holds
		$$
		\left|M_{n}(x;k)-\sum_{l=-1}^{p} {M}_{n}^{(l)}(x) k^{l}\right| \leq C|k|^{p+1}, \quad k \in \bar{\Omega}_{n}.
		$$
		
	\end{lemma}
	\par

	Define the jump matrices $v_n(x,t;k)~(n=1,2,\cdots,6)$ for $k\in\Sigma$ (see Figure \ref{sigma0}) as
	\begin{equation}\label{vn-jump-matrix}
		\begin{aligned}
			&\small v_1=\left(\begin{array}{ccc}
				1 & -r_1(k) \mathrm{e}^{-\theta_{21}} & 0 \\
				r_1^*(k) \mathrm{e}^{\theta_{21}} & 1-\left|r_1(k)\right|^2 & 0 \\
				0 & 0 & 1
			\end{array}\right),
			v_2=\left(\begin{array}{ccc}
				1 & 0 & 0 \\
				0 & 1-r_2(\omega k) r_2^*(\omega k) & -r_2^*(\omega k) \mathrm{e}^{-\theta_{32}} \\
				0 & r_2(\omega k) \mathrm{e}^{\theta_{32}} & 1
			\end{array}\right), \\
			&\small v_3=\left(\begin{array}{ccc}
				1-r_1\left(\omega^2 k\right) r_1^*\left(\omega^2 k\right) & 0 & r_1^*\left(\omega^2 k\right) \mathrm{e}^{-\theta_{31}} \\
				0 & 1 & 0 \\
				-r_1\left(\omega^2 k\right) \mathrm{e}^{\theta_{31}} & 0 & 1
			\end{array}\right),
			v_4=\left(\begin{array}{ccc}
				1-\left|r_2(k)\right|^2 & -r_2^*(k) \mathrm{e}^{-\theta_{21}} & 0 \\
				r_2(k) \mathrm{e}^{\theta_{21}} & 1 & 0 \\
				0 & 0 & 1
			\end{array}\right) \text {, } \\
			&\small v_5=\left(\begin{array}{ccc}
				1 & 0 & 0 \\
				0 & 1 & -r_1(\omega k) \mathrm{e}^{-\theta_{32}} \\
				0 & r_1^*(\omega k) \mathrm{e}^{\theta_{32}} & 1-r_1(\omega k) r_1^*(\omega k)
			\end{array}\right),
			v_6=\begin{pmatrix}
				1 & 0 & r_2\left(\omega^2 k\right) \mathrm{e}^{-\theta_{31}} \\
				0 & 1 & 0 \\
				-r_2^*\left(\omega^2 k\right) \mathrm{e}^{\theta_{31}} & 0 & 1-r_2\left(\omega^2 k\right) r_2^*\left(\omega^2 k\right)
			\end{pmatrix},
		\end{aligned}
	\end{equation}
	where the terms $\theta_{i j}=\left(l_i-l_j\right) x+$ $\left(z_i-z_j\right) t$~$(1 \leq j<i \leq 3)$ with $l_1(k)= \omega k,l_2=\omega^2k,l_3=k$ and $z_1(k)=9 \omega^2 k^5,z_2(k)=9 \omega k^5,z_3(k)=9  k^5$.
	
	Consequently, we can construct the RH problem \ref{rhp SK} for the SK equation (\ref{SK}), which has a singularity at \( k = 0 \). Redeemingly, we can rewrite the RH problem \( M(x,t;k) \) as \( N(x,t;k) := \begin{pmatrix} \omega & \omega^2 & 1 \end{pmatrix} M(x,t;k) \) which is a regular RH problem at $k=0$. In particular, the RH problem for \( N(x,t;k) \) obeys the following properties.
	
	\begin{RHproblem}
		Given the reflection coefficients $r_{1}(k)$ and $r_{2}(k)$, find a $1 \times 3$ vector-valued function $N(x, t; k)=N_n(x, t; k)$ for $k\in\Omega_n$ with the following properties:
		
		(a) $N_n(x, t; k):~ \mathbb{C} \backslash \Sigma \rightarrow \mathbb{C}^{3 \times 3}$ is analytic for $k\in \mathbb{C} \backslash \Sigma$.
		
		(b) The limits of $ N(x, t; k)$ as $k$ approaches $\Sigma $ from the left (+) and right (-) exist, are continuous on $\Sigma  $, and are related by
		$$
		N_{+}(x, t; k)=N_{-}(x, t; k) v(x, t; k), \quad k \in \Sigma,
		$$
		where $v(x, t; k)=v_n(x, t; k)$ for $n=1,2,\cdots,6$ are defined in terms of $r_{1}(k)$ and $r_{2}(k)$ by $(\ref{vn-jump-matrix})$.
		
		(c) $N(x, t; k)=\begin{pmatrix} \omega & \omega^2 & 1 \end{pmatrix}+\mathcal{O}\left(k^{-1}\right)$ as $k \rightarrow \infty,~ k \notin \Sigma $ and $N(x, t; k)=\mathcal{O}(1)$ as $k \rightarrow 0$.
		The reconstruction formula for the potential function of the SK equation (\ref{SK}) is
		\begin{equation}\label{recover-formula SK RHP N}
			u(x, t)= -\frac{1}{2}\frac{\partial}{\partial x}\lim _{k \rightarrow \infty}k (N(x, t; k)_{3}-1).
		\end{equation}
	\end{RHproblem}

    \begin{remark}
    By the similar way, the RH problem for matrix-valued $m(x,t;k)$ and reconstruction formula of the mSK equation (\ref{msk-equation}) can also be obtained, which are given in RH problem \ref{msk RHP} and Theorem \ref{thm for msk}, see also Ref. \cite{procA}.
    \end{remark}
	
	\subsection{Miura transformation between the SK equation and mSK equation}\label{Miura rhp}

	At first glance, the relationship between RH problems for $M(x,t;k)$ and $m(x,t;k)$ seems profound. Indeed, one can establish the Miura transformation between the SK equation (\ref{SK}) and the mSK equation (\ref{msk-equation}), as shown in (\ref{miura-SK}), akin to the relationship between the KdV and mKdV equations \cite{Charlier-Lenells-miura}. The proof of Theorem \ref{Miura theorem} is proposed below.
	\begin{proof}
		Suppose
		$$
		M(x,t;k)=\left(I+\frac{A_1(x,t)}{k}\right)m(x,t;k),
		$$
        where the matrix-valued function $A_1(x,t)$ is to be determined. By the symmetries in (\ref{Mnsymmeties}), we have
		$$
A_1(x,t)=\omega^2\mathcal{A}A_1(x,t)\mathcal{A}^{-1},
		$$
		which indicates that
		$$
		A_1(x)=\begin{pmatrix}
			\omega^2c_3 & \omega^2 c_1 & \omega^2 c_2\\
			\omega c_2 & \omega c_3 &\omega c_1\\
			c_1 & c_2 & c_3\\
		\end{pmatrix}.
		$$
		Here the functions $c_1, c_2$ and $c_3$ are determined by considering the limit of $k\to0$. Recall that $r_1(0)=\omega^2$ and $r_2(0)=1$ and thus
		$$
		v_1(0)=\begin{pmatrix}
			1 & -\omega^2 & 0\\
			\omega & 0 & 0\\
			0 & 0 & 1
		\end{pmatrix},\quad
		v_6(0)v_1(0)=\begin{pmatrix}
			1 & -\omega^2 & 1\\
			\omega & 0 & 0\\
			-1 & \omega^2 & 0\\
		\end{pmatrix}.
		$$
		For the RH problem associated with the mSK equation, one has
		$$
		m_{1}(x, t; k)=\mathcal{A}m_1(x,t;\omega k)\mathcal{A}^{-1}(v_6v_1)(x, t; k),
		$$
		thus taking $k=0$ yields
		$$
		m^{(1)}_0=\mathcal{A}m_0^{(1)}\mathcal{A}^{-1}(v_6v_1)(0)
		=\left(\begin{array}{ccc}
			\mathfrak{m}_{11}^{(0)} & \mathfrak{m}_{12}^{(0)} & \mathfrak{m}_{33}^{(0)} \\
			\omega \mathfrak{m}_{11}^{(0)}+\mathfrak{m}_{33}^{(0)}-\mathfrak{m}_{12}^{(0)} & \omega^2\left(\mathfrak{m}_{12}^{(0)}-\mathfrak{m}_{33}^{(0)}\right) & \mathfrak{m}_{33}^{(0)} \\
			\omega^2 \mathfrak{m}_{11}^{(0)}+\mathfrak{m}_{12}^{(0)} & \omega \mathfrak{m}_{12}^{(0)}+\mathfrak{m}_{33}^{(0)} & \mathfrak{m}_{33}^{(0)}
		\end{array}\right).
		$$
		Comparing with the asymptotic expansion of the RH problem for $M(x,t;k)$ of the SK equation (\ref{SK}) at $k=0$, we have
		$$
		M_0^{(-1)}={a_+(x)}\begin{pmatrix}
			\omega^2 & 0 & 0\\
			\omega^2 & 0 &0\\
			\omega^2 & 0 &0
		\end{pmatrix},
		$$
		and
		$$
		M_0^{(-1)}=A_1m^{(1)}_0,
		$$
		thus it follows
		$$
		c_3=-\omega^2c_1-\omega c_2.
		$$
        \par
		Moreover, taking the asymptotics  (\ref{J1-infty}) and (\ref{expansion of m}) as $k\to\infty$ into account, one can deduce that
		$$
		c_1=-\frac{\omega^2}{3}w(x,t),c_2=-\frac{\omega}{3}w(x,t),c_3=-\frac{w(x,t)}{3}.
		$$
		Finally, combining the reconstruction formula (\ref{recover-formula SK}) and (\ref{recover-formula mSK}), it follows that
		$$
		u(x, t)=-\frac{1}{2} \frac{\partial}{\partial x} \lim _{k \rightarrow \infty} k\left(M(x, t; k)_{33}-1\right)\\
		=-\frac{1}{2}\frac{\partial}{\partial x}\left(\frac{1}{3} \int_{\infty}^xw^2-\frac{w(x,t)}{3}\right)=\frac{1}{6}(w_x-w^2).
		$$
	\end{proof}

	\par	
	\par	
	\section{\bf  Asymptotic analysis for Sectors $\rm I$ and $\rm II$}\label{Sector I II}
	
	This section investigates the long-time asymptotics of the SK equation (\ref{SK}) and the mSK equation (\ref{msk-equation}) in Sectors $\rm I$ and $\rm II$ by Deift-Zhou steepest-descent method \cite{Deift-Zhou-1993}. In the subsequent sections, the analysis of the RH problems for the equations  (\ref{SK})  and (\ref{msk-equation})  is similar. Therefore, unless necessary, we will not distinguish between them and will abuse the same notation. Denote $\xi:=\frac{x}{t}$ and $\zeta:=\frac{t}{x}=\frac{1}{\xi}$, as parameters in Sector $\rm II$ and Sector $\rm I$, respectively. Moreover, the phase functions $\theta_{ij}$~$(1 \leq j<i \leq 3)$ can be rewritten as:
	\begin{equation}\label{phase function}
    \theta_{i j}(x,t;k)=
    \begin{cases}
        \begin{aligned}
			&t\left[\left(l_i-l_j\right) \xi+ \left(z_i-z_j\right) \right]:=t\Phi_{ij}(\xi;k),\\
			&x\left[\left(l_i-l_j\right) + \left(z_i-z_j\right)\zeta \right]:=x\tilde{\Phi}_{ij}(\zeta;k),
		\end{aligned}
    \end{cases}		
	\end{equation}
	with $l_j(k)= \omega^j k$ and $z_j(k)=9 \omega^{2j} k^5$ for $j=1,2,3$. Indeed, our main results provide the asymptotic formulas for \(u(x,t)\) in Sectors \(\mathrm{I}\) and \(\mathrm{II}\). In Sector \(\mathrm{I}\), these are given by \(\zeta := \frac{t}{x} \in [0, \zeta_{\rm{max}}]\), where \(0 < \zeta_{\rm{max}} < 1\) is a constant, and in Sector \(\mathrm{II}\), by \(\xi = \frac{x}{t}\) in compact subsets of \((0, \infty)\). Furthermore, introduce the saddle points $\pm k_0$ of $\Phi_{21}(\xi;k)$ and $\tilde{\Phi}_{21}(\zeta;k)$ for $x>0$, which are given by
	\begin{equation}\label{saddle points}
		k_0:=\sqrt[4]{\frac{x}{45t}}=\sqrt[4]{\frac{\xi}{45}}=\sqrt[4]{\frac{1}{ 45\zeta}}.
	\end{equation}
	Since $\theta_{21}(x,t;k)=-\theta_{31}(x,t;\omega k)=\theta_{32}(x,t;\omega^2 k)$, it follows that the saddle points of $\Phi_{31}(\tilde{\Phi}_{31})$ and $\Phi_{32}(\tilde{\Phi}_{32})$ are $\{\pm\omega k_0\}$ and $\{\pm\omega^2 k_0\}$, respectively. The saddle points on $\Sigma$ and the signature tables for $\Phi_{ij}(\tilde{\Phi}_{ij})$ are dicipted in Figure \ref{sign for phi}.
	
	
	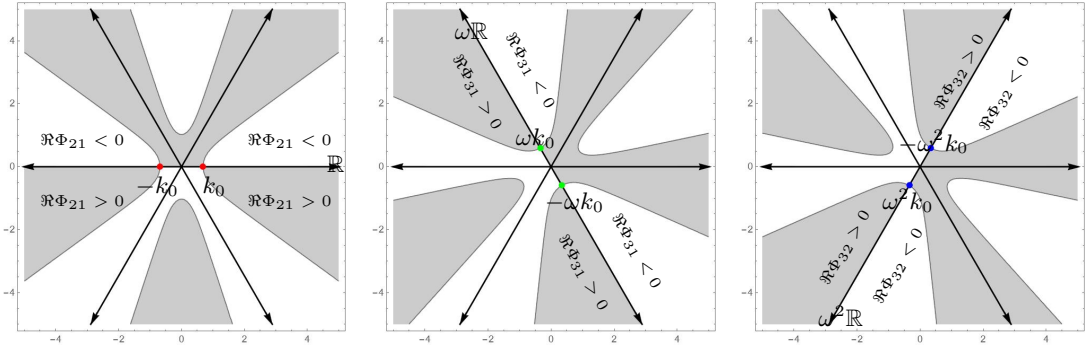
\begin{figure}[h]
		\centering
		\subfigure{
			\includegraphics[width=0.3\textwidth]{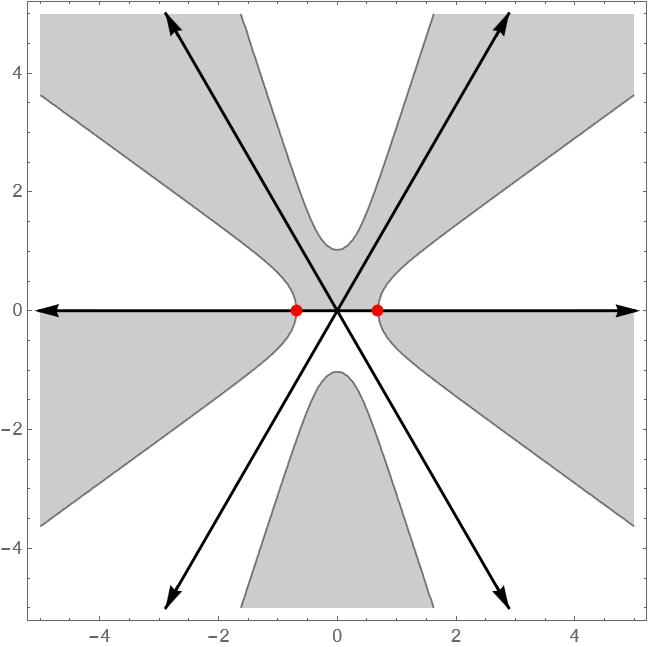}
			\put(-7,65){\small $\R$}
			\put(-54,57){\small $k_0$}
			\put(-79,57){\small $-k_0$}
			\put(-37,75){\fontsize{6}{7}\selectfont  $\Re\Phi_{21}<0$}
			\put(-115,75){\fontsize{6}{7}\selectfont $\Re\Phi_{21}<0$}
			\put(-115,52){\fontsize{6}{7}\selectfont $\Re\Phi_{21}>0$}
			\put(-37,52){\fontsize{6}{7}\selectfont $\Re\Phi_{21}>0$}
		}
		\hspace{0\textwidth}
		\subfigure{
			\includegraphics[width=0.3\textwidth]{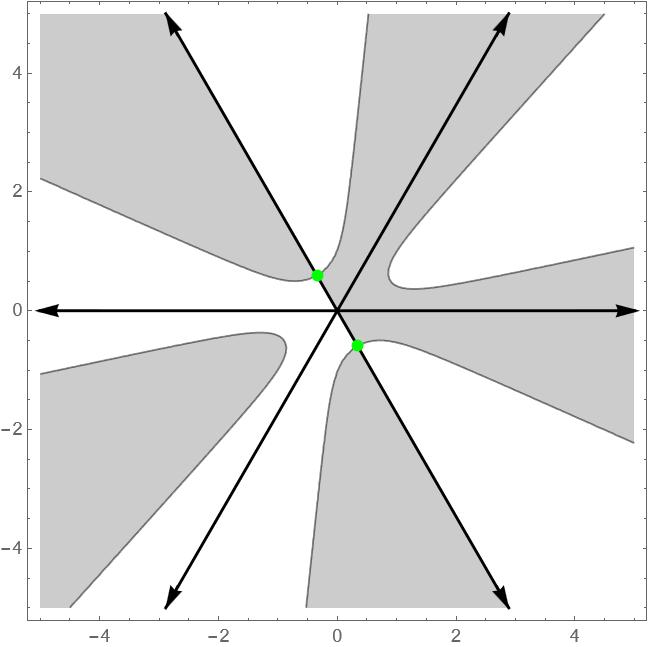}
			\put(-98,115){\small $\omega\R$}
			\put(-64,51){\small $-\omega k_0$}
			\put(-75,75){\small $\omega k_0$}
			\put(-40,48){\rotatebox{-60}{\fontsize{6}{6}\selectfont  $\Re\Phi_{31}<0$}}
			\put(-80,116){\rotatebox{-60}{\fontsize{6}{6}\selectfont $\Re\Phi_{31}<0$}}
			\put(-60,40){\rotatebox{-60}{\fontsize{6}{6}\selectfont $\Re\Phi_{31}>0$}}
			\put(-100,108){\rotatebox{-60}{\fontsize{6}{6}\selectfont $\Re\Phi_{31}>0$}}
		}
		\hspace{0\textwidth}
		\subfigure{
			\includegraphics[width=0.3\textwidth]{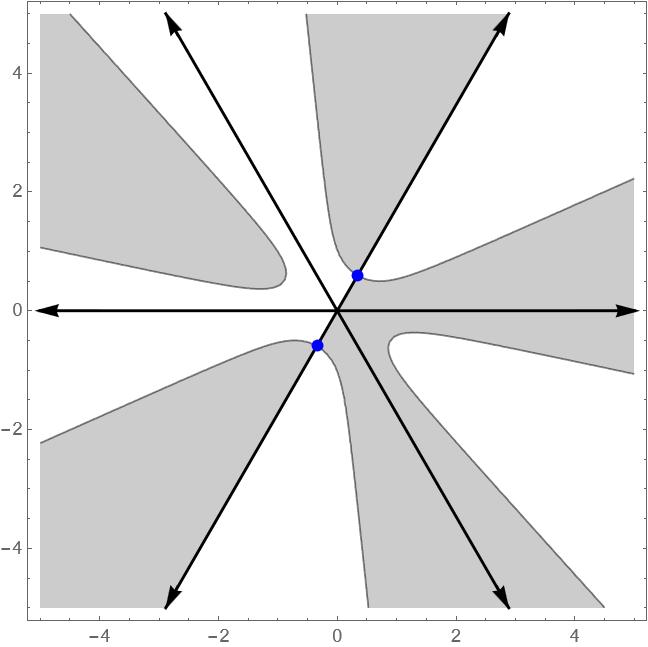}
			\put(-100,7){\small $\omega^2\R$}
			\put(-76,51){\small {$\omega^2 k_0$}}
			\put(-70,73){\small $-\omega^2 k_0$}
			\put(-80,15){\rotatebox{60}{\fontsize{6}{6}\selectfont  {$\Re\Phi_{32}<0$}}}
			\put(-40,80){\rotatebox{60}{\fontsize{6}{6}\selectfont $\Re\Phi_{32}<0$}}
			\put(-100,22){\rotatebox{60}{\fontsize{6}{6}\selectfont $\Re\Phi_{32}>0$}}
			\put(-57,88){\rotatebox{60}{\fontsize{6}{6}\selectfont $\Re\Phi_{32}>0$}}
		}
		\caption{From left to right: the signatures and saddle points of the functions $\Phi_{21}$, $\Phi_{31}$, and $\Phi_{32}$  for $\xi=10$ or $\zeta=\frac{1}{10}$. The grey regions correspond to $\{k \mid \Re \Phi_{ij} > 0\}$, while the white regions correspond to $\{k \mid \Re \Phi_{ij} < 0\}$. }
		\label{sign for phi}	
	\end{figure}
	The Deift-Zhou steepest-descent method is adopted through a series of transformations. We would like to denote $M^{(j)}$ as the RH problem after the $j$-th transformations (denote $M^{(0)}=M$), and let $\Sigma^{(j)}$ and $v^{(j)}$ represent the corresponding jump contours and jump matrices of the RH problem for $M^{(j)}$. The contributions to the leading-order term in asymptotic formular due to the local parametrix near the six saddle points $\pm\omega^jk_0~(j=0,1,2)$ and the global parametrix $\Delta(k)$ defined in (\ref{Delta}) below. Thanks to the symmetries outlined in (\ref{m symmeties}) and (\ref{M symmeties}), it is enough to focus on demonstrating the transformations restricted to $\mathbb{R}$ and the analysis in the vicinity of the point $k_0$. Indeed, transformations maintain the symmetries of the RH problems, i.e.,
	\begin{equation}\label{symmetry}
		\begin{aligned}
			& v^{(j)}(x, t; k)=\mathcal{A} v^{(j)}(x, t; \omega k) \mathcal{A}^{-1}=\mathcal{B} \overline{v^{(j)}\left(x, t; \bar k\right)} \mathcal{B}, \quad k \in \Sigma^{(j)} , \\
			& M^{(j)}(x, t; k)=\mathcal{A}M^{(j)}(x, t; \omega k) \mathcal{A}^{-1}=\mathcal{B}\overline{M^{(j)}\left(x, t;\bar  k\right)} \mathcal{B}, \quad k \in \mathbb{C} \backslash \Sigma^{(j)}. \\
		\end{aligned}
	\end{equation}

	\subsection{Global parametrix $\bf \Delta$ and  the first deformation}
	To implement the transformations of the RH problem for $M(x,t;k)$, introduce the global parametrix $\Delta(k)$. For each $\zeta\in [0, \zeta_{\rm{max}}]$ and $\xi$ in some compact subset of $(0,\infty)$, let $\delta_1(\xi;k)$ or $\delta_1(\zeta;k)$: $\C\setminus[k_0,\infty)$ be a solution of the following scalar RH problem
	$$
	\begin{aligned}
		\delta_{1+}(k)  =\delta_{1-}(k)\left(1-|r_1(k)|^2\right), & & k\in[k_0,\infty),
	\end{aligned}
	$$
	while $\delta_4(\xi;k)$ or $\delta_4(\zeta;k)$: $\C\setminus(-\infty,-k_0]$ obeys the jump condition
	$$
	\begin{aligned}
		\delta_{4+}(k) & =\delta_{4-}(k)\left(1-|r_2(k)|^2\right), & & k\in(-\infty,-k_0],
	\end{aligned}
	$$
	where both $\delta_1(k)$ and $\delta_4(k)$ satisfy the normalization condition $\delta_j(k)=1+\mathcal{O}\left(\frac{1}{k}\right)$, as $k\to\infty$ for $j=1,4$. Thanks to Assumptions $\ref{solitonless}$ and $\ref{r less 1}$, it is concluded that $1-|r_j(k)|^2>0$ for $j=1,2$, respectively, when $|k|\geq k_0$. Thus the $\delta_j(k)$ are well-defined and by using the Plemelj's formula, it is derived that
	\begin{equation}\label{delta1}
		\delta_1(k)=\exp \left\{\frac{1}{2 \pi i} \int_{\left[k_0, \infty\right)} \frac{\ln \left(1-\left|r_1(s)\right|^2\right)}{s-k} d s\right\}, \quad k \in \mathbb{C} \backslash [k_0,\infty),
	\end{equation}
	and
	\begin{equation}\label{delta4}
		\delta_4(k)=\exp \left\{\frac{1}{2 \pi i} \int_{[-k_0,-\infty)} \frac{\ln \left(1-\left|r_2(s)\right|^2\right)}{s-k} d s\right\}, \quad k \in \mathbb{C} \backslash (-\infty,-k_0].	
	\end{equation}
	
	Let $\log_{\theta}(k)$ represents the logarithm of $k$ with the branch cut along $\arg k = \theta$, that is,
	$
	\log_0(k) = \ln|k| + \arg_0(k)$ for $\arg_0(k) \in (0,2\pi)$, and
	$\log_{\pi}(k) = \ln|k| + \arg_{\pi}(k)$ for $\arg_{\pi}(k) \in (-\pi,\pi).$
    
	\begin{proposition}\label{properties of delta}
		The basic properties of functions $\delta_j(k)$ for $j=1,4$ are given below:
		\par	
		\begin{enumerate}
			\item On the one hand, $\delta_1(k)$ can be rewritten as
			$$
			\delta_1( k)=\mathrm{e}^{-i \nu_1 \log_0\left(k-k_0\right)} \mathrm{e}^{-\chi_1( k)}
			$$
			where
			$
			\nu_1=-\frac{1}{2 \pi} \ln \left(1-\left|r_1\left(k_0\right)\right|^2\right),\ \chi_1(\xi;k)=\frac{1}{2 \pi i} \int_{k_0}^{\infty} \log_0(k-s) \mathrm{d}\ln \left(1-\left|r_1(s)\right|^2\right).
			$
			
			\noindent One the other hand, one has
			$$
			\delta_4( k)=\mathrm{e}^{-i \nu_4 \log_{\pi}\left(k+k_0\right)} \mathrm{e}^{-\chi_4( k)}
			$$
			where
			$
			\nu_4=-\frac{1}{2 \pi} \ln \left(1-\left|r_2\left(-k_0\right)\right|^2\right),\ \chi_4( \xi;k)=\frac{1}{2 \pi i} \int_{-k_0}^{-\infty} \log_{\pi}(k-s) \mathrm{d}\ln \left(1-\left|r_2(s)\right|^2\right).
			$
			\item  The $\delta_{1\pm}(k)$ and $\delta_{4\pm}(k)$ satisfy the conjugate symmetries and are bounded, for $k>k_0$ and $k<-k_0$, respectively, such that
			$$
			\delta_1( k)={(\overline{\delta_1( \bar{k})})^{-1}}, \  k \in \mathbb{C} \backslash [k_0,\infty),\quad
			\delta_4( k)={(\overline{\delta_4( \bar{k})})^{-1}}, \  k \in \mathbb{C} \backslash (-\infty,-k_0];
			$$
			and
			$
			|\delta_{1}^{\pm 1}(k)|<\infty\ {\rm for}~ k \in \mathbb{C} \backslash [k_0,\infty);\
			|\delta_{4}^{\pm 1}(k)|<\infty\  {\rm for}~~ \mathbb{C} \backslash (-\infty,-k_0].
			$
			
			\item  As $k\to \pm k_0$ along a path non-tangential to $|k|\ge k_0$, it follows
			$$
			\begin{aligned}
				& \left|\chi_1(\xi;k)-\chi_1\left(\xi;k_0\right)\right| \leq C\left|k-k_0\right|\left(1+|\ln | k-k_0||\right), \\
				& \left|\chi_4(\xi;k)-\chi_4\left(\xi; -k_0\right)\right| \leq C\left|k+k_0\right|\left(1+|\ln | k+k_0||\right), \\
			\end{aligned}
			$$
			where $C$ is a constant independent of $\xi$ and $\zeta$. Especially, for $\xi$ in some subset of $\R_+$, one has
			$$
			\begin{aligned}
				& \left|\partial_x\left(\chi_1(\xi;k)-\chi_1\left(\xi;k_0\right)\right)\right| \leq \frac{C}{t}\left(1+|\ln | k-k_0||\right),\\
				& \left|\partial_x\left(\chi_4(\xi;k)-\chi_4\left(\xi; -k_0\right)\right)\right| \leq \frac{C}{t}\left(1+|\ln | k+k_0||\right),
			\end{aligned}
			$$
			where
			$
			\left|\partial_x \chi_j\left(\xi;k_0\right)\right|\leq \frac{C}{t} 
		$, and
        $	\partial_x\left(\delta_j(\xi;k)^{\pm 1}\right)=\frac{\pm i \nu_j}{180 tk_0^3\left(k-k^*\right)} \delta_j(\xi;k)^{\pm 1}
			$, for $k^*=k_0,j=1$, $k^*=-k_0,j=4$.
		\end{enumerate}
	\end{proposition}
	\begin{proof}
		We focus on proving the properties of $\delta_1(k)$, with the properties of $\delta_4(k)$ being analogous. Using the technique of integration by parts, it is immediate to derive (1) from the expression in (\ref{delta1}). Note that we choose $\log_0$ for $\delta_1(k)$ and $\log_{\pi}$ for $\delta_4(k)$ based on their respective jump conditions. By leveraging the uniqueness of the RH problem associated with $\delta_1(k)$, it can be inferred that $\delta_1(k) = \overline {\delta_1( \bar k)}^{-1}$, which states the second property of $\delta_1(k)$. Based on the representation of $\chi_1$ and the properties of $r_1(k)$, the inequalities in item (3) directly follow from some standard estimates, see \cite{Charlier-Lenells-Wang-2021} and \cite{Charlier-Lenells-2022}.
	\end{proof}
	\par
	Reminding the symmetries in (\ref{symmetry}), define $\delta_j(\xi;k)\ \text{or}\ \delta_j(\zeta;k)$ for $j=2,3,5,6$ as follows:
	$$
	\begin{aligned}
		&\delta_3( k)=\delta_1\left( \omega^2 k\right),\ k \in \mathbb{C} \backslash [\omega k_0,\omega \infty),
		&&\delta_5( k)=\delta_1( \omega k),\ k \in \mathbb{C} \backslash [\omega^2k_0,\omega^2\infty),\\
		&\delta_2( k)=\delta_4\left( \omega k\right),\ k \in \mathbb{C} \backslash (-\omega^2 \infty,-\omega^2 k_0],
		&&\delta_6( k)=\delta_4( \omega^2 k),\ k \in \mathbb{C} \backslash (-\omega \infty,-\omega k_0],
	\end{aligned}
	$$
	which satisfy the jump conditions
	$$
	\begin{aligned}
		&\delta_{3+}( k)=\delta_{3-}( k)\left(1-\left|r_1\left(\omega^2 k\right)\right|^2\right),\ \omega^2 k>k_0 ,
		&&\delta_{5+}( k)=\delta_{5-}( k)\left(1-\left|r_1(\omega k)\right|^2\right), \ \omega k>k_0 ,\\
		&\delta_{2+}( k)=\delta_{2-}( k)\left(1-\left|r_2\left(\omega k\right)\right|^2\right), \ \omega k<-k_0 ,
		&&\delta_{6+}( k)=\delta_{6-}( k)\left(1-\left|r_2(\omega^2 k)\right|^2\right), \ \omega^2 k<-k_0.
	\end{aligned}
	$$

	\begin{remark}
		The expressions for the functions $\delta_n(k)$  include $\log_{\frac{(n-1)\pi}{3}}(k)$ and are respectively defined in the intervals $\frac{(n-1)\pi}{3} < \arg(k) < 2\pi + \frac{(n-1)\pi}{3}$ for $n=1,2,3$. For $n=4,5,6$, the functions $\delta_n(k)$ also involve $\log_{\frac{(n-1)\pi}{3}}(k)$ but are defined in the intervals $-\frac{(7-n)\pi}{3} < \arg(k) < 2\pi - \frac{(7-n)\pi}{3}$.
	\end{remark}
	
	Now it is ready to define the global parametrix $\Delta(k)$ as 
	\begin{equation}\label{Delta}
		\Delta(k)=\left(\begin{array}{ccc}
			\frac{\delta_1( k)\delta_6( k)}{\delta_3( k)\delta_4( k)} & 0 & 0 \\
			0 & \frac{\delta_5( k)\delta_4( k)}{\delta_1( k)\delta_2( k)} & 0 \\
			0 & 0 & \frac{\delta_3( k)\delta_2( k)}{\delta_5( k)\delta_6( k)}
		\end{array}\right).
	\end{equation}
    \par
	Furthermore, take the first transformation by
	$$
	M^{(1)}(x, t; k)=M(x, t; k) \Delta(k),
	$$
	then the jump matrix is $v^{(1)}(x,t;k)=\Delta^{-1}_-v(x,t;k)\Delta_+$, and the corresponding contour $\Sigma^{(1)}$ is decipited in Figure \ref{Sigma1}. More explicitly, for $|k|>k_0$  the jump matrices $v_1^{(1)}$ and $v_4^{(1)}$ are
	
	\begin{equation}\label{firstv14}
		\begin{aligned}
			&v_1^{(1)}=
			\begin{pmatrix}
				1-|r_1(k)|^2 & -\frac{\tilde\delta_{v_1}}{\delta_{1-}^2}\frac{r_1(k)}{1-|r_1(k)|^2} \mathrm{e}^{-t \Phi_{21}} & 0 \\
				\frac{\delta_{1+}^2 }{\tilde\delta_{v_1}} \frac{r_1^*(k)}{1-|r_1(k)|^2} \mathrm{e}^{t \Phi_{21}} & 1 & 0 \\
				0 & 0 & 1
			\end{pmatrix},&&k\in\Sigma_{1}^{(1)},\\
			&v_{4}^{(1)}=
			\begin{pmatrix}
				1 & -\frac{\delta_{4+}^2}{\tilde\delta_{v_4}}\frac{r_{2}^{*}(k)}{1-|r_2(k)|^2} \mathrm{e}^{-t\Phi_{21}} & 0 \\
				\frac{\tilde\delta_{v_4}}{\delta_{4-}^2}\frac{r_{2}(k)}{1-|r_2(k)|^2} \mathrm{e}^{t\Phi_{21}} & 1-\left|r_{2}(k)\right|^{2} & 0 \\
				0 & 0 & 1
			\end{pmatrix},&&k\in\Sigma_{4}^{(1)},
		\end{aligned}
	\end{equation}
where $\tilde\delta_{v_1}=\frac{\delta_3\delta^2_4\delta_5}{\delta_6\delta_2}$ and $\tilde\delta_{v_4}=\frac{\delta_1^2\delta_2\delta_6}{\delta_5\delta_3}$.
	On the other hand, the functions $\delta_j(k)$ for $j=1,4$ have no jumps between $-k_0<k<k_0$, thus the jump matrices $v_{7}^{(1)}$ and $v_{10}^{(1)}$ are written as
	\begin{equation}\label{firstv710}
		\begin{aligned}
			&v_7^{(1)}=
			\begin{pmatrix}
				1 & -\frac{\tilde\delta_{v1} }{\delta_{1}^2} r_1(k) \mathrm{e}^{-t \Phi_{21}} & 0 \\
				\frac{\delta_{1}^2 }{\tilde\delta_{v1}} r_1^*(k) \mathrm{e}^{t \Phi_{21}} & 1-r_1(k) r_1^*(k)& 0 \\
				0 & 0 & 1
			\end{pmatrix},&&k\in\Sigma_7^{(1)},\\
			&v_{10}^{(1)}=\begin{pmatrix}
				1-\left|r_{2}(k)\right|^{2} & -\frac{\delta_{4}^2}{\tilde\delta_{v_4}}{r_{2}^{*}(k)} \mathrm{e}^{-t\Phi_{21}} & 0 \\
				\frac{\tilde\delta_{v_4}}{\delta_{4}^2}{r_{2}(k)} \mathrm{e}^{t\Phi_{21}} & 1 & 0 \\
				0 & 0 & 1
			\end{pmatrix},&&k\in\Sigma_{10}^{(1)}.
		\end{aligned}
	\end{equation}
	Furthermore, based on the symmetries in (\ref{symmetry}), the other jump matrices can be derived from (\ref{firstv14}) and (\ref{firstv710}), and they are omitted for brevity.
	
	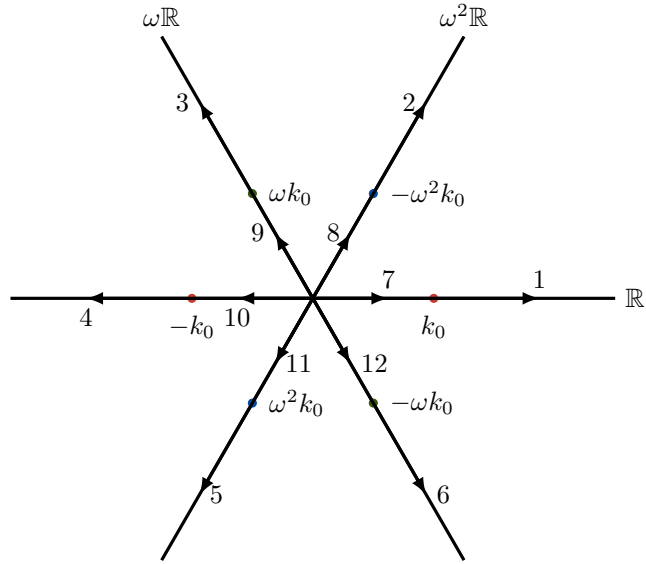
\begin{figure}[h]
		\centering
		\begin{tikzpicture}[>=latex]
			\draw[very thick] (-4,0) to (4,0) node[black,right]{$\mathbb{R}$};
			\draw[very thick] (-2,-1.732*2) to (2,1.732*2)  node[black,above]{$\omega^2\mathbb{R}$};
			\draw[very thick] (2,-1.732*2) to (-2,1.732*2)    node[black,above]{$\omega\mathbb{R}$};
			\filldraw[mred] (1.6,0) node[black,below=1mm]{$k_{0}$} circle (1.5pt);
			\filldraw[mred] (-1.6,0) node[black,below=1mm]{$-k_{0}$} circle (1.5pt);
			\filldraw[mblue] (.8,1.732*0.8) node[black,right=1mm]{$-\omega^{2}k_{0}$} circle (1.5pt);
			\filldraw[mblue] (-.8,-1.732*0.8) node[black,right=1mm]{$\omega^{2}k_{0}$} circle (1.5pt);
			\filldraw[mgreen] (.8,-1.732*0.8) node[black,right=1mm]{$-\omega k_{0}$} circle (1.5pt);
			\filldraw[mgreen] (-.8,1.732*0.8) node[black,right=1mm]{$\omega k_{0}$} circle (1.5pt);
			\draw[<->,very thick] (-1,0) node[below] {$10$} to (1,0) node[above] {$7$};
			\draw[<->,very thick] (-3,0) node[below] {$4$} to (3,0) node[above] {$1$};
			\draw[<->,very thick] (-.5,-1.732*.5) node[right] {$11$} to (.5,1.732*.5)node[left] {$8$};
			\draw[<->,very thick] (-1.5,-1.732*1.5) node[right] {$5$} to (1.5,1.732*1.5) node[left] {$2$};
			\draw[<->,very thick] (-.5,1.732*.5) node[left] {$9$} to (.5,-1.732*.5) node[right] {$12$};
			\draw[<->,very thick] (-1.5,1.732*1.5) node[left] {$3$} to (1.5,-1.732*1.5) node[right] {$6$};
		\end{tikzpicture}
		\caption{{\protect\small
				The jump contour $\Sigma^{(1)}$ and saddle points $\pm\omega^jk_0$ for $j=0,1,2$.}}
		\label{Sigma1}
	\end{figure}
	\subsection{The second deformation}
	The purpose of the second deformation is to expand the jumps \( v^{(1)}_{\{1,4,7,10\}} \) into regions where the $\Re{\Phi_{21}(\xi;k)}$ keeps decaying as \( t \rightarrow \infty \) for $\xi$ in some compact subset of $\R_+$, or $\Re{\tilde{\Phi}_{21}(\zeta;k)}$ keeps the decay properties as $x\to\infty$ for $\zeta\in[0,\zeta_{\rm{max}}]$. Naturally, let $U_1,U_2,\cdots,U_6$ be the open sets  defined in Figure \ref{regionU}, which are coincided with the signature of $\Re{\Phi_{21}}$.
	\begin{figure}[!h]
		\centering
		\begin{overpic}[width=.5\textwidth]{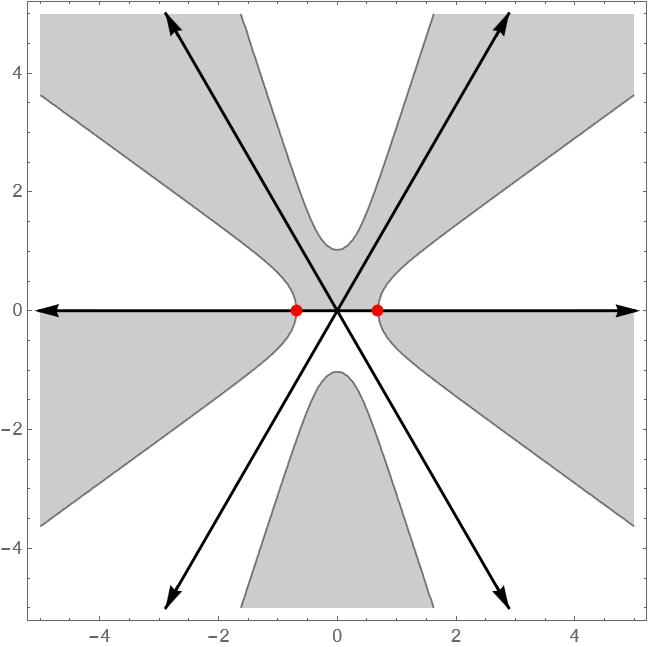}
			\put(73.6,58.9){\small $U_1$}
			\put(49.5,55.9){\small $U_2$}
			\put(23.5,58.9){\small $U_3$}
			\put(73.6,40.3){\small $U_6$}
			\put(49.5,43.3){\small $U_5$}
			\put(23.5,40.3){\small $U_4$}
			\put(57.5,48.3){\small $k_0$}
			\put(41.5,48.3){\small $-k_0$}
		\end{overpic}
		\caption{{\protect\small
				The open subsets $U_j~(j=1,2,\cdots,6)$ and the saddle points $\pm k_0$ (red points).  The gray regions correspond to $\{k \mid \Re \Phi_{21} > 0\}$, while the white regions correspond to $\{k \mid \Re \Phi_{21} < 0\}$.}}
		\label{regionU}
	\end{figure}
	
	Note that the non-diagonal parts in $v^{(1)}_{1,4}$ involve $\frac{r_j(k)}{1-r_j(k)r_j^{*}(k)}$ for $j=1,2$, and we also need to decompose them. Suppose that
	$$
	\rho_1(k)=\frac{r_1(k)}{1-r_1(k)r_1^{*}(k)},\quad \rho_2(k)=\frac{r_2(k)}{1-r_2(k)r_2^{*}(k)}.
	$$
	\begin{lemma}\label{Lemma-analytic-extension}	
		For any integer $N\ge1$, the functions $r_j(k)$ and $\rho_j(k)~(j=1,2)$ have the following decompositions
		$$
		\begin{array}{ll}
			r_1(k)=r_{1, a}(x, t; k)+r_{1, r}(x, t; k), & k \in\left[0, k_0\right), \\
			r_2(k)=r_{2, a}(x, t; k)+r_{2, r}(x, t; k), & k \in(-k_0,0], \\
			\rho_1(k)=\rho_{1, a}(x, t; k)+\rho_{1, r}(x, t; k), & k \in\left[k_0, \infty\right),\\
			\rho_2(k)=\rho_{2, a}(x, t; k)+\rho_{2, r}(x, t; k), & k \in\left( -\infty,-k_0\right].
		\end{array}
		$$
		Furthermore, the decomposition functions have the properties as follow:
		
		\begin{enumerate}
			\item For each $t\ge1$ and $\xi$ in some compact subset of $\R_+$ or $x\ge1$ and $\zeta\in[0,\zeta_{\rm{max}}]$, the functions \( r_{1,a} \) and \( r_{2,a} \) are defined and continuous on \(\bar{U}_2 \cap \{ k \mid 0 \le \Re(k) \le k_0 \}\) and \(\bar{U}_5 \cap \{ k \mid -k_0 \le \Re(k) \le 0 \}\), respectively, and are analytic in the interior of their respective domains. While the functions \( \rho_{1,a} \) and \( \rho_{2,a} \) are defined and continuous on \( \bar{U}_6 \) and \( \bar{U}_3 \), respectively, and are analytic for \( k\in U_6 \) and \( k\in U_3 \), respectively.
			
			\item For $t\ge1$ and $\xi$ in some compact subset of $\R_+$, the functions $r_{j,a}$ and $\rho_{j,a}$ for $j=1,2$  satisfy the following estimates:
			$$
			\left|r_{j, a}(x, t; k)-\sum_{i=0}^N\frac{r_j^{(i)}(k_*)(k-k_*)^i}{i!}\right| \leq C|k-k_*|^{N+1}{\mathrm{e}^{t|\Re \Phi_{21}(\xi;k)|/4}},
			$$
			$$
			\left|\rho_{j, a}(x, t; k)-\sum_{i=0}^N\frac{\rho_j^{(i)}(k_*)( k-k_*)^i}{i!}\right| \leq C|k-k_*|^{N+1}{\mathrm{e}^{t|\Re \Phi_{21}(\xi;k)|/4}},
			$$
			and
			$$
			\left|\rho_{j, a}(x, t; k)\right| \leq \frac{C}{1+|k|^{N+1}}{\mathrm{e}^{t|\Re \Phi_{21}(\xi;k)|/4}}.
			$$
			Meanwhile, the first inequality holds for \( j=1 \) when \( k_* \in \{0, k_0\} \) and \( k \) is in \( \bar{U}_2 \) such that \( 0 \le \Re(k) \le k_0 \), and for \( j=2 \) when \( k_* \in \{0, -k_0\} \) and \( k \) is in \( \bar{U}_5 \) such that \( -k_0 \le \Re(k) \le 0 \). The inequalities involving \( \rho_j(k) \) are established for \( j=1 \) when \( k \) is in \( U_6 \) and \( k=k_0 \), and for \( j=2 \) when \( k \) is in \( U_3 \) and \( k=k_0 \).
			\item Similarly, for $x\ge1$ and $\zeta\in[0,\zeta_{\rm{max}}]$, the functions $r_{j,a}$ and $\rho_{j,a}$ for $j=1,2$ obey
			$$
			\left|r_{j, a}(x, t; k)-\sum_{i=0}^N\frac{r_j^{(i)}(k_*)(k-k_*)^i}{i!}\right| \leq C|k-k_*|^{N+1}{\mathrm{e}^{x|\Re \tilde{\Phi}_{21}(\zeta; k)|/4}},
			$$
			$$
			\left|\rho_{j, a}(x, t; k)-\sum_{i=0}^N\frac{\rho_j^{(i)}(k_*)( k-k_*)^i}{i!}\right| \leq C|k-k_*|^{N+1}{\mathrm{e}^{x|\Re \tilde{\Phi}_{21}(\zeta; k)|/4}},
			$$
			and
			$$
			\left|\rho_{j, a}(x, t; k)\right| \leq \frac{C}{1+|k|^{N+1}}{\mathrm{e}^{x|\Re \tilde{\Phi}_{21}(\zeta; k)|/4}}.
			$$
			Especially, for $k_*\in\{\pm k_0\}$ and $\zeta$ near $0$, we have the following stronger estimates:
			$$
			\left|r_{j, a}(x, t; k)-\sum_{i=0}^N\frac{r_j^{(i)}(k_*)(k-k_*)^i}{i!}\right| \leq \mathrm{C}_N(\zeta)|k-k_*|^{N+1}{\mathrm{e}^{x|\Re \tilde{\Phi}_{21}(\zeta; k)|/4}},
			$$
			$$
			\left|\rho_{j, a}(x, t; k)-\sum_{i=0}^N\frac{\rho_j^{(i)}(k_*)( k-k_*)^i}{i!}\right| \leq \mathrm{C}_N(\zeta)|k-k_*|^{N+1}{\mathrm{e}^{x|\Re \tilde{\Phi}_{21}(\zeta; k)|/4}},
			$$
			where $\mathrm{C}_N(\zeta)\ge0$ is a smooth function of $\zeta$ which vanishes to any order at $\zeta=0$.
			\item For each \(1 \leq p \leq \infty\), the \(L^p\)-norm of \(r_{j,r}\) and \(\rho_{j,r}\) for \(j=1,2\), on their respective domains is \(\mathcal{O}(t^{-N-\frac{1}{2}})\) as \(t \to \infty\) for \(\xi\) in some compact subset of \(\mathbb{R}_+\), and \(\mathcal{O}(x^{-N-\frac{1}{2}})\) as \(x \to \infty\) for \(\zeta \in [0, \zeta_{\rm{max}}]\).
			%
			%
			%
		\end{enumerate}
		\begin{remark}
		By the Schwartz reflection,	the functions $r_{j}^*(k)$ and $\rho_{j}^*(k)$ can be decomposed in the same procedure. Furthermore, the symmetries in (\ref{symmetry}) indicate the decompositions of other matrices.
		\end{remark}
	\end{lemma}
	
	\begin{proof}
		The proof follows standard techniques outlined in \cite{Deift-Zhou-1993}. Therefore, we only provide a proof of the third property about $\rho_1(k)$ for brevity. Suppose that $M\ge N+1$ is an positive integer, then there exists a rational function $h_0(k)$ which has no poles in $U_6$ and such that $h_0(k)$ is coincided with $\rho_1(k)$ at $k_0$ for $4M$-order, and $h_0(k)=\mathcal{O}(k^{-4M})$, as $k\to\infty$ for $k\in[k_0,\infty)$. Denote $h(k):=\rho_1(k)-h_0(k)$, and notice that $-i\tilde{\Phi}_{21}(\zeta;k):={9\sqrt{3}k^5\zeta-\sqrt{3}k}:=\phi(k)$ is a monotonic increasing function from $[k_0,\infty)\rightarrow[0,\infty)$, and thus define
		$$
		H(\phi):=\begin{cases}
			\begin{aligned}
				&\frac{k^{2M}h(k)}{(k-k_0)^M},&&\phi\ge0,\\
				&0,&&\phi<0.
			\end{aligned}
		\end{cases}
		$$
		It is seen that $H(\phi)$ is a smooth function for $k\in\R\setminus\{k_0\}$, and for $n\ge1$, we have
		$$
		F^{(n)}(\phi)=\left(\frac{1}{(k^4-k_0^4)}\frac{d}{\mathrm{d}k}\right)^n\frac{k^{2M}h(k)}{(k-k_0)^M},\ \phi\ge0.
		$$
		Consequently, for $M$ large enough, it is immediate that $H\in H^{N+1}(\R)$. Introduce
		$$
		\hat H(s)=\frac{1}{\sqrt{2\pi}}\int_{\R}H(\phi)\mathrm{e}^{-i\phi s} d\phi, \quad H(\phi)=\frac{1}{\sqrt{2\pi}}\int_{\R}\hat H(s)\mathrm{e}^{i\phi s} \mathrm{ds},
		$$
		and by the Plancherel's Theorem, $\|s^{N+1}\hat H(s)\|_{L^2(\R)}=\| H^{N+1}(\phi)\|_{L^2(\R)}$, it follows that
		$$
		h(k)=\frac{(k-k_0)^M}{k^{2M}}\frac{1}{\sqrt{2\pi}}\int_{\R}\hat H(s)\mathrm{e}^{\tilde{\Phi}_{21}(\zeta;k)s } \mathrm{ds}.
		$$
		For $x\ge1$, decompose $h(k)$ as $h(k):=h_1(x;k)+h_2(x;k)$ with
		$$
		h_1(x;k)=\frac{(k-k_0)^M}{k^{2M}}\frac{1}{\sqrt{2\pi}}\int_{-\infty}^{\frac{x}{4}}\hat H(s)\mathrm{e}^{\tilde{\Phi}_{21}(\zeta;k)s } \mathrm{ds},
		$$
		and
		$$
		h_2(x;k)=\frac{(k-k_0)^M}{k^{2M}}\frac{1}{\sqrt{2\pi}}\int^{\infty}_{\frac{x}{4}}\hat H(s)\mathrm{e}^{\tilde{\Phi}_{21}(\zeta;k)s } \mathrm{ds}.
		$$
		Since $\Re \tilde{\Phi}_{21}=0$ for $k>k_0$, it states that
		$$
		|h_2(x;k)|\le\frac{C}{1+|k|^N}\|s^{N+1}\hat H(s)\|_{L^2(\R)}x^{-N-\frac{1}{2}},
		$$
		and
		$$
		|h_1(x;k)|\le\frac{(k-k_0)^M}{k^{2M}}\|\hat H(s)\|_{L^1(\R)}\mathrm{e}^{\frac{x}{4}|\Re\tilde{\Phi}_{21}(\zeta;k)|}.
		$$
		Let $\rho_{1,a}(x, t; k):=h_0(x, t; k)+h_1(x, t; k)$ for $k\in\bar U_6$ and $\rho_{1,r}(x, t; k):=h_2(x, t; k)$ for $k\ge k_0$, then the properties in item (3) of $\rho_{1,a}$ and $\rho_{1,r}$ hold. Moreover, since $r_1(k)$ tends to $0$, rapidly as $k\to\infty$, it follows that $k_0=\infty$ as $\zeta=0$, and $r_1(k_0)$ and $\rho_1(k_0)$ vanish.
	\end{proof}
	\par
	As a result, the matrices $v^{(1)}_{\{1,4,7,10\}}$ can be decomposed into
	$$
	v_1^{(1)}(x,t;k)=v^{(1)}_{1,lower}\ v_{1,r}^{(1)}\ v^{(1)}_{1,upper},
	$$
	where
	$$
	v^{(1)}_{1,lower}=
	\begin{pmatrix}
		1 & -\frac{\tilde\delta_{v_1}}{\delta_{1-}^2}\rho_{1,a} \mathrm{e}^{-t \Phi_{21}} & 0 \\
		0 & 1 & 0 \\
		0 & 0 & 1
	\end{pmatrix},\quad
	v^{(1)}_{1,upper}=
	\begin{pmatrix}
		1 & 0 & 0 \\
		\frac{\delta_{1+}^2 }{\tilde\delta_{v_1}} \rho^*_{1,a} \mathrm{e}^{t \Phi_{21}} & 1 & 0 \\
		0 & 0 & 1
	\end{pmatrix},
	$$
	and
	$$
	v_{1,r}^{(1)}(x,t;k)=
	\begin{pmatrix}		1-\frac{\delta_{+}^2}{\delta_{1-}^2}\rho_{1,r}(k)\rho^*_{1,r}(k) & -\frac{\tilde\delta_{v_1}}{\delta_{1-}^2}\rho_{1,r} \mathrm{e}^{-t \Phi_{21}} & 0 \\
		\frac{\delta_{1+}^2 }{\tilde\delta_{v_1}} \rho_{1,r}^*  \mathrm{e}^{t \Phi_{21}} & 1 & 0 \\
		0 & 0 & 1
	\end{pmatrix}.
	$$
	Moreover, we have
	$$
	v_{7}^{(1)}=v^{(1)}_{7,lower}\ v_{7,r}^{(1)}\ v^{(1)}_{7,upper},
	$$
	where
	$$
	v^{(1)}_{7,lower}=\begin{pmatrix}
		1 & 0 & 0 \\
		\frac{\delta_{1}^2 }{\tilde\delta_{v1}}r^*_{1,a} \mathrm{e}^{t \Phi_{21}} & 1 & 0 \\
		0 & 0 & 1
	\end{pmatrix},\quad
	v^{(1)}_{7,upper}=\begin{pmatrix}
		1 & -\frac{\tilde\delta_{v1} }{\delta_{1}^2}r_{1,a} \mathrm{e}^{-t \Phi_{21}} & 0 \\
		0 & 1 & 0 \\
		0 & 0 & 1
	\end{pmatrix},
	$$
	and
	$$
	v_{7,r}^{(1)}=\begin{pmatrix}
		1 & -\frac{\tilde\delta_{v1} }{\delta_{1}^2} r_{1,r}(k) \mathrm{e}^{-t\Phi_{21}} & 0 \\
		\frac{\delta_{1}^2 }{\tilde\delta_{v1}}r_{1,r}^{*}(k) \mathrm{e}^{t\Phi_{21}} & 1-r_{1,r}(k)r^*_{1,r}(k) & 0 \\
		0 & 0 & 1
	\end{pmatrix}.
	$$
	The same procedure yields
	$$
	v_{4}^{(1)}=v^{(1)}_{4,upper}\ v_{4,r}^{(1)}\ v^{(1)}_{4,lower},
	$$
	where
	$$
	v^{(1)}_{4,upper}=\begin{pmatrix}
		1 & 0 & 0 \\
		\frac{\tilde\delta_{v_4}}{\delta_{4-}^2}\rho_{2,a} \mathrm{e}^{t\Phi_{21}} & 1 & 0 \\
		0 & 0 & 1
	\end{pmatrix},\quad
	v^{(1)}_{4,lower}=\begin{pmatrix}
		1 & -\frac{\delta_{4+}^2}{\tilde\delta_{v_4}}\rho_{2,a}^{*} \mathrm{e}^{-t\Phi_{21}} & 0 \\
		0 & 1 & 0 \\
		0 & 0 & 1
	\end{pmatrix},
	$$
	and
	$$
	v_{4,r}^{(1)}=\begin{pmatrix}
		1 & -\frac{\delta_{4+}^2}{\tilde\delta_{v_4}}\rho^*_{2,r} \mathrm{e}^{-t\Phi_{21}} & 0 \\
		\frac{\tilde\delta_{v_4}}{\delta_{4-}^2}\rho_{2,r} \mathrm{e}^{t\Phi_{21}} & 1-\frac{\delta_{4+}^2}{\delta_{4-}^2}\rho_{2,r}\rho^*_{2,r} & 0 \\
		0 & 0 & 1
	\end{pmatrix}.
	$$
	On the other hand, one has
	$$
	v_{10}^{(1)}=v^{(1)}_{10,upper}\ v_{10,r}^{(1)}\ v^{(1)}_{10,lower},
	$$
	with
	$$
	v^{(1)}_{10,upper}=\begin{pmatrix}
		1 & -\frac{\delta_{4}^2}{\tilde\delta_{v_4}}r_{2,a}^* \mathrm{e}^{-t \Phi_{21}} & 0 \\
		0 & 1 & 0 \\
		0 & 0 & 1
	\end{pmatrix},\quad
	v^{(1)}_{10,lower}=\begin{pmatrix}
		1 & 0 & 0 \\
		\frac{\tilde\delta_{v_4}}{\delta_{4}^2}r_{2,a} \mathrm{e}^{t \Phi_{21}} & 1 & 0 \\
		0 & 0 & 1
	\end{pmatrix},
	$$
	and
	$$
	v_{10,r}^{(1)}=\begin{pmatrix}
		1-r_{2,r}(k)r^*_{2,r}(k) & -\frac{\delta_{4}^2}{\tilde\delta_{v_4}}r_{2,r}^{*}(k) \mathrm{e}^{-t\Phi_{21}} & 0 \\
		\frac{\tilde\delta_{v_4}}{\delta_{4}^2}r_{2,r}(k) \mathrm{e}^{t\Phi_{21}} & 1 & 0 \\
		0 & 0 & 1
	\end{pmatrix}.
	$$
    \par
	Let $\Sigma^{(2)}$ be depicted in Figure \ref{Sigma2} and transform the RH problem $M^{(1)}(x,t;k)\rightarrow M^{(2)}(x,t;k)$ by
	$$
	M^{(2)}(x,t;k)=M^{(1)}(x,t;k)G^{(1)}(x,t;k),\ k\in\C\setminus\Sigma^{(2)},
	$$
	where $G^{(1)}(x,t;k):=G^{(1)}_n(x,t;k)$ for $n=1,2,\cdots,6$. To be specific, $G^{(1)}_1(x,t;k)$ is defined near $k_0$ by
	$$
	G^{(1)}_1(x, t; k):= \begin{cases}\left(v_{1,upper}^{(1) }\right)^{-1}, & k\ \text{on the} - \text{side of } \Sigma_{1}^{(2)},\\ \left(v_{7,upper}^{(1) }\right)^{-1}, & k\ \text{on the} + \text{side of } \Sigma_{2}^{(2)}, \\
		v_{7,lower}^{(1) }, & k\ \text{on the} - \text{side of } \Sigma_{3}^{(2)}, \\ v_{1,lower}^{(1) }, & k\ \text{on the} + \text{side of } \Sigma_{4}^{(2)}, \\
	\end{cases}
	$$
	and
	$G^{(1)}_4(x,t;k)$ is defined near $-k_0$ by
	$$
	G^{(1)}_4(x, t; k):= \begin{cases}
		v_{10,upper}^{(1) }, & k\ \text{on the} - \text{side of } \Sigma_{10}^{(2)}, \\
		v_{4,upper}^{(1) }, & k\ \text{on the} + \text{side of } \Sigma_{7}^{(2)}, \\ \left(v_{4,lower}^{(1) }\right)^{-1}, & k\ \text{on the} - \text{side of } \Sigma_{8}^{(2)}, \\ \left(v_{10,lower}^{(1) }\right)^{-1}, & k\ \text{on the} + \text{side of } \Sigma_{9}^{(2)}. \\
	\end{cases}
	$$
	The matrix-valued functions $G_n^{(1)}(x,t;k)$ for $n=2,3,5,6$ near \(\pm\omega^j k_0\) for \(j=1,2\) can be derived by the symmetries in (\ref{symmetry}), so we omit them for conciseness.
	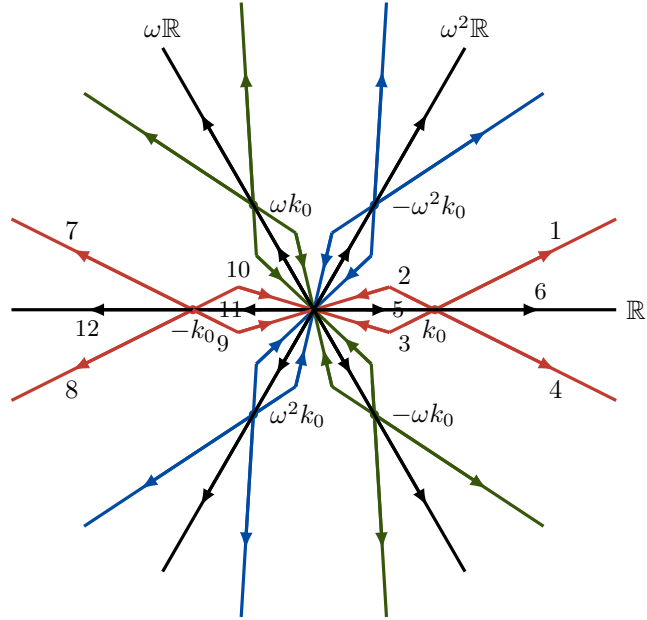
\begin{figure}[h]
		\centering
		\begin{tikzpicture}[>=latex]
			\draw[very thick] (-4,0) to (4,0) node[black,right]{$\mathbb{R}$};
			\draw[very thick] (-2,-1.732*2) to (2,1.732*2)  node[black,above]{$\omega^2\mathbb{R}$};
			\draw[very thick] (2,-1.732*2) to (-2,1.732*2)    node[black,above]{$\omega\mathbb{R}$};
			\filldraw[mred] (1.6,0) node[black,below=0.1mm]{$k_{0}$} circle (1.5pt);
			\filldraw[mred] (-1.6,0) node[black,below=0.1mm]{$-k_{0}$} circle (1.5pt);
			\filldraw[mblue] (.8,1.732*0.8) node[black,right=1mm]{$-\omega^{2}k_{0}$} circle (1.5pt);
			\filldraw[mblue] (-.8,-1.732*0.8) node[black,right=1mm]{$\omega^{2}k_{0}$} circle (1.5pt);
			\filldraw[mgreen] (.8,-1.732*0.8) node[black,right=1mm]{$-\omega k_{0}$} circle (1.5pt);
			\filldraw[mgreen] (-.8,1.732*0.8) node[black,right=1mm]{$\omega k_{0}$} circle (1.5pt);
			
			\draw[->,very thick,rotate=60,mblue] (1.6,0)  to (3.2,0.8) ;
			\draw[-,very thick,rotate=60,mblue] (1.6,0)  to (4,1.2);
			\draw[->,very thick,rotate=60,mblue] (1.6,0)  to (3.2,-0.8) ;
			\draw[->,very thick,rotate=60,mblue] (-1.6,0)  to (-3.2,0.8) ;
			\draw[-,very thick,rotate=60,mblue] (-1.6,0)  to (-4,1.2);
			\draw[->,very thick,rotate=60,mblue] (-1.6,0)  to (-3.2,-0.8);
			\draw[-,very thick,rotate=60,mblue] (-1.6,0)  to (-4,-1.2);
			\draw[-,very thick,rotate=60,mblue] (1.6,0)  to (4,-1.2);
			\draw[-,very thick,rotate=60,mblue] (-1.6,0)  to (-1.0,0.3);
			\draw[-,very thick,rotate=60,mblue] (1.6,0)  to (1.0,0.3);
			\draw[-,very thick,rotate=60,mblue] (1.0,-0.3)  to (0,0);
			\draw[->,very thick,rotate=60,mblue] (1.0,-0.3)  to (0.5,-0.15);
			\draw[-,very thick,rotate=60,mblue] (1.0,0.3)  to (0,0);
			\draw[->,very thick,rotate=60,mblue] (1.0,0.3)  to (0.5,0.15);
			\draw[-,very thick,rotate=60,mblue] (0.7,0)  to (0.9,0);
			\draw[-,very thick,rotate=60,mblue] (-1.6,0)  to (-1.2,-0.2) ;
			\draw[-,very thick,rotate=60,mblue] (-1.6,0)  to (-1.0,-0.3);
			\draw[-,very thick,rotate=60,mblue] (-1.0,-0.3)  to (0,0);
			\draw[->,very thick,rotate=60,mblue] (-1.0,-0.3)  to (-0.5,-0.15);
			\draw[-,very thick,rotate=60,mblue] (1.6,0)  to (1.0,-0.3);
			\draw[->,very thick,rotate=60,mblue] (-1.0,0.3)  to (-0.5,0.15);
			\draw[-,very thick,rotate=60,mblue] (-1.0,0.3)  to (0,0);
			
			\draw[->,very thick,rotate=120,mgreen] (1.6,0)  to (3.2,0.8) ;
			\draw[-,very thick,rotate=120,mgreen] (1.6,0)  to (4,1.2);
			\draw[->,very thick,rotate=120,mgreen] (1.6,0)  to (3.2,-0.8) ;
			\draw[->,very thick,rotate=120,mgreen] (-1.6,0)  to (-3.2,0.8) ;
			\draw[-,very thick,rotate=120,mgreen] (-1.6,0)  to (-4,1.2);
			\draw[->,very thick,rotate=120,mgreen] (-1.6,0)  to (-3.2,-0.8);
			\draw[-,very thick,rotate=120,mgreen] (-1.6,0)  to (-4,-1.2);
			\draw[-,very thick,rotate=120,mgreen] (1.6,0)  to (4,-1.2);
			\draw[-,very thick,rotate=120,mgreen] (-1.6,0)  to (-1.0,0.3);
			\draw[-,very thick,rotate=120,mgreen] (1.6,0)  to (1.0,0.3);
			\draw[-,very thick,rotate=120,mgreen] (1.0,-0.3)  to (0,0);
			\draw[->,very thick,rotate=120,mgreen] (1.0,-0.3)  to (0.5,-0.15);
			\draw[-,very thick,rotate=120,mgreen] (1.0,0.3)  to (0,0);
			\draw[->,very thick,rotate=120,mgreen] (1.0,0.3)  to (0.5,0.15);
			\draw[-,very thick,rotate=120,mgreen] (0.7,0)  to (0.9,0);
			\draw[-,very thick,rotate=120,mgreen] (-1.6,0)  to (-1.2,-0.2) ;
			\draw[-,very thick,rotate=120,mgreen] (-1.6,0)  to (-1.0,-0.3);
			\draw[-,very thick,rotate=120,mgreen] (-1.0,-0.3)  to (0,0);
			\draw[->,very thick,rotate=120,mgreen] (-1.0,-0.3)  to (-0.5,-0.15);
			\draw[-,very thick,rotate=120,mgreen] (1.6,0)  to (1.0,-0.3);
			\draw[->,very thick,rotate=120,mgreen] (-1.0,0.3)  to (-0.5,0.15);
			\draw[-,very thick,rotate=120,mgreen] (-1.0,0.3)  to (0,0);
			
			\draw[->,very thick,mred] (1.6,0)  to (3.2,0.8) node[above,black] {$\small 1$};
			\draw[-,very thick,mred] (1.6,0)  to (4,1.2);
			\draw[->,very thick,mred] (1.6,0)  to (3.2,-0.8) node[below,black] {$\small 4$};
			\draw[->,very thick,mred] (-1.6,0)  to (-3.2,0.8) node[above,black] {$\small 7$};
			\draw[-,very thick,mred] (-1.6,0)  to (-4,1.2);
			\draw[->,very thick,mred] (-1.6,0)  to (-3.2,-0.8) node[below,black] {$\small 8$};
			\draw[-,very thick,mred] (-1.6,0)  to (-4,-1.2);
			\draw[-,very thick,mred] (-1.6,0)  to (-1.0,0.3)node[above,black] {\small $10$};
			\draw[->,very thick,mred] (-1.0,0.3)  to (-0.5,0.15);
			\draw[-,very thick,mred] (-1.0,0.3)  to (0,0);
			\draw[-,very thick,mred] (1.6,0)  to (4,-1.2);
			\draw[-,very thick,mred] (1.6,0)  to (1.2,0.2) node[above,black] {$\small 2$};
			\draw[-,very thick,mred] (1.6,0)  to (1.0,0.3);
			\draw[-,very thick,mred] (1.0,0.3)  to (0,0);
			\draw[->,very thick,mred] (1.0,0.3)  to (0.5,0.15);
			\draw[-,very thick,mred] (1.6,0)  to (1.2,-0.2) node[below,black] {\small $3$};
			\draw[-,very thick,mred] (1.6,0)  to (1.0,-0.3);
			\draw[-,very thick,mred] (1.0,-0.3)  to (0,0);
			\draw[->,very thick,mred] (1.0,-0.3)  to (0.5,-0.15);
			\draw[-,very thick] (0.7,0)  to (0.9,0)node[right] {\small $5$};
			\draw[<->,very thick] (-1,0)  to (1,0);
			\draw[<->,very thick] (-3,0) node[below] {\small $12$} to (3,0) node[above] {\small $6$};
			\draw[-,very thick,mred] (-1.6,0)  to (-1.2,-0.2) node[below,black] {\small $9$};
			\draw[-,very thick,mred] (-1.6,0)  to (-1.0,-0.3);
			\draw[-,very thick,mred] (-1.0,-0.3)  to (0,0);
			\draw[->,very thick,mred] (-1.0,-0.3)  to (-0.5,-0.15);
			\draw[-,very thick] (-0.7,0)  to (-1.4,0)node[right] {\small $11$};
			
			\draw[<->,very thick] (-.5,-1.732*.5)  to (.5,1.732*.5);
			
			%
			
			\draw[<->,very thick] (-1.5,-1.732*1.5)  to (1.5,1.732*1.5) ;
			\draw[<->,very thick] (-.5,1.732*.5)  to (.5,-1.732*.5) ;
			\draw[<->,very thick] (-1.5,1.732*1.5)  to (1.5,-1.732*1.5);
			
		\end{tikzpicture}
		\caption{{\protect\small
				The jump contour $\Sigma^{(2)}$ and saddle points $\pm\omega^jk_0$ for $j=0,1,2$.}}
		\label{Sigma2}
	\end{figure}
	\begin{lemma}
		The functions $G^{(1)}(x,t;k)$ and $(G^{(1)}(x,t;k))^{-1}$ are uniformly bounded for $k\in\C\setminus\Sigma^{(2)}$, and $G^{(1)}(x,t;k)=I+\mathcal{O}(\frac{1}{k})$ as $k\to\infty$.
	\end{lemma}
	\begin{proof}
		We focus only on \(G_1^{(1)}(x,t;k)\), the treatment of $(G^{(1)}(x,t;k))^{-1}$ is analogous. Indeed, it suffices to show that  $\frac{\delta_{1+}^2 }{\tilde\delta_{v_1}} \rho^*_{1,a} \mathrm{e}^{t \Phi_{21}}$ and $\frac{\tilde\delta_{v1} }{\delta_{1}^2}r_{1,a} \mathrm{e}^{-t \Phi_{21}}$ are bounded on their corresponding regions. Recall that \(\delta_j(k)\) is bounded in \(\mathbb{C} \setminus \Sigma^{(1)}\), and \(\rho_{1,a}\) and \(r_{j,a}\) satisfy the Lemma \ref{Lemma-analytic-extension}. Hence, it follows that \( \left| \frac{\delta_{1+}^2}{\tilde{\delta}_{v_1}} \rho^*_{1,a} \mathrm{e}^{t \Phi_{21}} \right| \leq \frac{C}{1+|k|^N} \mathrm{e}^{-\frac{3t}{4}|\Re {\Phi}_{21}(\xi;k)|} \) and \( \left| \frac{\tilde{\delta}_{v_1}}{\delta_{1}^2} r_{1,a} \mathrm{e}^{-t \Phi_{21}} \right| \)
		is uniformly bounded due to the compactness in its domain.
	\end{proof}
	
	\begin{lemma}
		For $\frac{1}{M}\le\xi \le M$ or $\zeta\in[0,\zeta_{\rm{max}}]$, the jump matrix $v^{(2)}$ converges uniformly to identity matrix
 $I$ as $t\to\infty$ or $x\to\infty$ and $\partial_xv^{(2)}$ uniformly converges to the zero matrix except for the points near the saddle points, i.e., $\left\{\pm k_0,\pm \omega k_0,\pm \omega^2k_0\right\}$. In particular, the jump matrix $v^{(2)}$ on $\Sigma_{5,6}$ has the following estimates:
		$$
		\begin{aligned}
			\|(1+|\cdot|)\partial_x^l(v^{(2)}-I)\|_{(L^1\cap L^{\infty})(\Sigma_{5,6}^{(2)})}&\le Ct^{-N},\ \text{for}\ M^{-1}\le\xi\le M,\\
			\|(1+|\cdot|)\partial_x^l(v^{(2)}-I)\|_{(L^1\cap L^{\infty})(\Sigma_{5,6}^{(2)})}&\le Cx^{-N},\ \text{for}\ \zeta\in[0,\zeta_{\rm{max}}].
		\end{aligned}
		$$
		Moreover, reminding the  symmetries of the jump matrices, the similar estimates on the other cuts of $\Sigma_j^{(2)}$  can be gotten immediately. 
	\end{lemma}
	\begin{proof}
		We focus on the jump matrices on $\Sigma_{1,2,\cdots,6}^{(2)}$ and since for $k\in\Sigma_{1,2,3,4}^{(2)}$ the exponential part $\Re{(t\Phi_{21})}$ or $\Re{(x\tilde{\Phi}_{21})}$ is strictly less than zero for $k\in\Sigma_{1,3}^{(2)}$ and strictly larger than zero for $k\in\Sigma_{2,4}^{(2)}$, except for the points near the saddle point $k_0$. Using the Lemma \ref{Lemma-analytic-extension} on the properties of $r_{1,a},\rho_{1,a}$ and the boundedness of the functions $\delta_j(k)$, it is concluded that $v^{(2)}_{1,2,3,4}$ ($\partial_x v_{1,2,3,4}^{(2)}$) converges to $I$ (resp. to the zero matrix) as $t\to\infty$ or $x\to\infty$.
		For \(\frac{1}{M}\le  \xi \le M\), one has
		\[
		(v_5^{(2)} - I)_{12} = -\frac{\tilde{\delta}_{v1}}{\delta_{1}^2} r_{1,r}(k) \mathrm{e}^{-t \Phi_{21}}, \quad
		(v_6^{(2)} - I)_{12} = -\frac{\tilde{\delta}_{v_1}}{\delta_{1-}^2} \rho_{1,r} \mathrm{e}^{-t \Phi_{21}}.
		\]
		Moreover, the properties of \(\delta_j(k)\) for \(j=1,2,\cdots,6\) in Lemma \ref{properties of delta}, along with those of \(r_{j,r}\) and \(\rho_{j,r}\) for \(j=1,2\), imply that
		\[
		|(v_5^{(2)} - I)_{12}| \le Ct^{-N}, \quad |(v_6^{(2)} - I)_{12}| \le Ct^{-N}.
		\]
        \par
		The analysis for \(\zeta \in [0, \zeta_{\rm{max}}]\) is similar and this completes the proof of this lemma.
	\end{proof}

	\subsection{The third deformation }\label{local parmetrices}
	In order to factorize the RH problem for $M^{(2)}(x, t; k)$ into a model problem, focus on the contours $\Sigma_A$ and $\Sigma_B$ of the forms
	$$
	\Sigma_A=\Sigma_{\{1,2,3,4\}}^{(2)}\cap B_{\epsilon}(k_0),\ \Sigma_B=\Sigma_{\{7,8,9,10\}}^{(2)}\cap B_{\epsilon}(-k_0),
	$$
	with the disk $B_{\epsilon}(\pm k_0):=\{k\in\C||k\mp k_0|<\epsilon\}$. Observing that the exponential parts in the jump matrices on the contours $\Sigma_A$ and $\Sigma_B$ are $\pm t\Phi_{21}(\xi;k)$ or $\pm x\tilde{\Phi}_{21}(\zeta;k)$, expand $t\Phi_{21}(\xi;k)$ at $k_0$ into
	$$
	\begin{aligned}
		t\Phi_{21}(k)&=t[(\omega^2-\omega)k\xi+(\omega-\omega^2)9k^5]
		=9t(\omega-\omega^2)(k^5-5kk_0^4)\\
		&=9\sqrt{3}it[(k-k_0)^5+5k_0(k-k_0)^4+10k_0^2(k-k_0)^3+10k_0^3(k-k_0)^2-4k_0^5],
	\end{aligned}
	$$
	and set $t=\frac{x}{45k_0^4}$ to expand $x\tilde{\Phi}_{21}(\zeta;k)$ into 
	$$
	\begin{aligned}
		x\tilde{\Phi}_{21}(k)&=x[(\omega^2-\omega)k+(\omega-\omega^2)9k^5\zeta]
		=\frac{x}{5k_0^4}(\omega-\omega^2)(k^5-5kk_0^4)\\
		&=\frac{\sqrt{3}i x}{5k_0^4}[(k-k_0)^5+5k_0(k-k_0)^4+10k_0^2(k-k_0)^3+10k_0^3(k-k_0)^2-4k_0^5].
	\end{aligned}
	$$
	
	Suppose  $z_1={3^{\frac{5}{4}}2\sqrt{5t}k_0^{\frac{3}{2}}}(k-k_0)={3^{\frac{1}{4}}2\sqrt{x}k_0^{-\frac{1}{2}}}(k-k_0)$, then rewrite $t\Phi_{21}(\xi;k)$ and $x\tilde{\Phi}_{21}(\zeta;k)$ as
	$$
	\begin{aligned}
		t\Phi_{21}(k)
		&=9\sqrt{3}ita^3[a^2z_1^5+5ak_0z_1^4+10k_0^2z_1^3]+\frac{iz_1^2}{2}+t\Phi_{21}(k_0)\\
		&:=t\Phi_{21}^{0}(k_0;z_1)+\frac{iz_1^2}{2}+t\Phi_{21}(k_0),
	\end{aligned}
	$$
	and
	$$
	\begin{aligned}
		x\tilde{\Phi}_{21}(k)
		&=\frac{\sqrt{3}i x}{5k_0^4}a^3[a^2z_1^5+5ak_0z_1^4+10k_0^2z_1^3]+\frac{iz_1^2}{2}+x\tilde{\Phi}_{21}(k_0)\\
		&:=x\tilde{\Phi}_{21}^{0}(k_0;z_1)+\frac{iz_1^2}{2}+x\tilde{\Phi}_{21}(k_0),
	\end{aligned}
	$$
	where $a=\frac{1}{{3^{\frac{5}{4}}2\sqrt{5t}k_0^{\frac{3}{2}}}}=\frac{k_0^{\frac{1}{2}}}{3^{\frac{1}{4}}2\sqrt{x}}$.
	\par
	The other parts of the jump matrices on contour $\Sigma_A$ involve the function $\delta_1(k)$, i.e.,
	$$
	\delta_1( k)=\mathrm{e}^{-i \nu_1 \log_0\left(k-k_0\right)} \mathrm{e}^{-\chi_1( k)},\quad k\in\C\setminus[k_0,\infty),
	$$
	where
	$$
	\nu_1=-\frac{1}{2 \pi} \ln \left(1-\left|r_1\left(k_0\right)\right|^2\right),
	$$
	and
	$$
	\chi_1( k)=\frac{1}{2 \pi i} \int_{k_0}^{\infty} \log_0(k-s) \mathrm{d}\ln \left(1-\left|r_1(s)\right|^2\right) .
	$$
	Again, rewrite the following fraction as
	$$
	\begin{aligned}
		\frac{\delta_{1+}^{2}(k)}{\delta_{\tilde v_1}(k)}	
		&=\mathrm{e}^{-2i \nu_1 \log_0\left(z\right)}\frac{a^{-2i\nu_1}\mathrm{e}^{-2\chi_1(k_0)}}{\tilde\delta_{ v_1}(k_0)}
		\frac{\mathrm{e}^{2\chi_1(k_0)-2\chi_1(k)}\tilde\delta_{v_1}(k_0)}{\tilde\delta_{ v_1}(k)}\\
		&:=\mathrm{e}^{-2i \nu_1 \log_0\left(z\right)}\delta_A^0\delta_A^1,
	\end{aligned}
	$$
	where $\delta_A^0=\frac{a^{-2i\nu}\mathrm{e}^{-2\chi_1(k_0)}}{\delta_{\tilde v_1}(k_0)}$ and $\delta_A^1=\frac{\mathrm{e}^{2\chi_1(k_0)-2\chi_1(k)}\delta_{\tilde v_1}(k_0)}{\delta_{\tilde v_1}(k)}$.
    \par
	On the other hand, on the contour $\Sigma_B$, expand $t\Phi_{21}(\xi;k)$ and $x\tilde{\Phi}_{21}(\zeta;k)$ at $-k_0$ as
	$$
	\begin{aligned}
		&t\Phi_{21}(k)
		=9\sqrt{3}it[(k+k_0)^5-5k_0(k+k_0)^4+10k_0^2(k+k_0)^3-10k_0^3(k+k_0)^2+4k_0^5],\\
		&x\tilde{\Phi}_{21}(k)
		=\frac{\sqrt{3}i x}{5k_0^4}[(k+k_0)^5-5k_0(k+k_0)^4+10k_0^2(k+k_0)^3-10k_0^3(k+k_0)^2+4k_0^5].
	\end{aligned}
	$$
    \par
	Suppose $z_2={3^{\frac{5}{4}}2\sqrt{5t}k_0^{\frac{3}{2}}}(k+k_0)={3^{\frac{1}{4}}2\sqrt{x}k_0^{-\frac{1}{2}}}(k+k_0)$, and rewrite $t\Phi_{21}$ and $x\tilde{\Phi}_{21}$ as
	$$
	\begin{aligned}
		t\Phi_{21}(k)
		&=9\sqrt{3}ita^3[a^2z^5_2-5ak_0z^4_2+10k_0^2z^3_2]-\frac{iz^2_2}{2}+t\Phi_{21}(-k_0)\\
		&=t\Phi_{21}^{0}(-k_0;z_2)-\frac{iz_2^2}{2}+t\Phi_{21}(-k_0),
	\end{aligned}
	$$
	and
	$$
	\begin{aligned}
		x\tilde{\Phi}_{21}(k)
		&=\frac{\sqrt{3}i x}{5k_0^4}a^3[a^2z^5_2-5ak_0z^4_2+10k_0^2z^3_2]-\frac{iz^2_2}{2}+x\tilde{\Phi}_{21}(-k_0)\\
		&=x\tilde{\Phi}_{21}^{0}(-k_0;z_2)-\frac{iz_2^2}{2}+x\tilde{\Phi}_{21}(-k_0).
	\end{aligned}
	$$
    \par
	Moreover, recall the function $\delta_4$ on the contour $\Sigma_B$ as
	$$
	\delta_4( k)=\mathrm{e}^{-i \nu_4 \log_{\pi}\left(k+k_0\right)} \mathrm{e}^{-\chi_4( k)},\quad k\in\C\setminus(-\infty,-k_0],
	$$
	with
	$$
	\nu_4=-\frac{1}{2 \pi} \ln \left(1-\left|r_2\left(-k_0\right)\right|^2\right),
	$$
	and
	$$
	\chi_4( k)=\frac{1}{2 \pi i} \int_{-k_0}^{-\infty} \log_{\pi}(k-s) \mathrm{d}\ln \left(1-\left|r_2(s)\right|^2\right).
	$$
    \par
	In addition, one has
	$$
	\begin{aligned}
		\frac{\tilde\delta_{v_4}}{\delta_{4}^2}&=\mathrm{e}^{2i \nu_4 \log_{\pi}\left(z_2\right)} \frac{\tilde\delta_{v_4}(-k_0)}{a^{-2i\nu_4}\mathrm{e}^{-2\chi_4( -k_0)}}\frac{\tilde\delta_{ v_4}(k)}{\mathrm{e}^{2\chi_4(-k_0)-2\chi_4(k)}\tilde\delta_{ v_4}(-k_0)}\\
		&:=\mathrm{e}^{2i \nu_4 \log_{\pi}\left(z_2\right)}\left(\delta_B^0\right)^{-1}\left(\delta_B^1\right)^{-1},
	\end{aligned}
	$$
	where $\delta_B^0=\frac{a^{-2i\nu_4}\mathrm{e}^{-2\chi_4(-k_0)}}{\delta_{\tilde v_4}(-k_0)}$ and $\delta_B^1=\frac{\mathrm{e}^{2\chi_4(-k_0)-2\chi_4(k)}\delta_{ v_4}(-k_0)}{\delta_{ \tilde v_4}(k)}$.

	Now, define the matrix-valued function $H(\pm k_0,t)$ and deform the RH problem for $M^{(2)}(x,t;k)$ by the transformation
	$$
	M^{(3,\epsilon)}(x,t;k)=M^{(2)}(x,t;k)H(\pm k_0,t),\quad k\in B_{\epsilon}(\pm k_0),
	$$
	where
	$$
	H(k_0,t)=\left(\begin{array}{ccc}
		\left(\delta_A^{0}\right)^{-\frac{1}{2}} \mathrm{e}^{-\frac{t}{2} \Phi_{21}\left( k_0\right)} & 0 & 0 \\
		0 & \left(\delta_A^{0}\right)^{\frac{1}{2}} \mathrm{e}^{\frac{t}{2} \Phi_{21}\left( k_0\right)} & 0 \\
		0 & 0 & 1
	\end{array}\right),
	$$
	and
	$$
	H(-k_0,t)=\left(\begin{array}{ccc}
		\left(\delta_B^{0}\right)^{\frac{1}{2}} \mathrm{e}^{-\frac{t}{2} \Phi_{21}\left( -k_0\right)} & 0 & 0 \\
		0 & \left(\delta_B^{0}\right)^{-\frac{1}{2}} \mathrm{e}^{\frac{t}{2} \Phi_{21}\left( -k_0\right)} & 0 \\
		0 & 0 & 1
	\end{array}\right).
	$$
	For the case \(\zeta \in [0, \zeta_{\rm{max}}]\), introduce the matrix
    \(\tilde{H}(\pm k_0, x) = H(\pm k_0, t)\). Thus we adopt the convention notation \(H(\pm k_0, t)\) to denote the transformation for both \(M^{-1} \le \xi \le M\) and \(\zeta \in [0, \zeta_{\rm{max}}]\). In order to keep the symbol with the Appendix \ref{appendix}, we let $z$ denote $z_1$ and $z_2$. Consequently, the jump matrices on the contours $\Sigma_A$ and $\Sigma_B$ are
	$$
	\begin{aligned}
		v^{(3,\epsilon)}_1&=\begin{pmatrix}
			1 & 0 & 0 \\
			\mathrm{e}^{-2i \nu_1 \log_0\left(z\right)}\delta_A^1 \rho^*_{1,a} \mathrm{e}^{t \Phi_{21}^0(k_0;z)+\frac{iz^2}{2}} & 1 & 0 \\
			0 & 0 & 1
		\end{pmatrix},\\
		v^{(3,\epsilon)}_2&=\begin{pmatrix}
			1 & \mathrm{e}^{2i \nu_1 \log_0\left(z\right)}(\delta_A^1)^{-1} r_{1,a} \mathrm{e}^{-t \Phi_{21}^0(k_0;z)-\frac{iz^2}{2}} & 0 \\
			0 & 1 & 0 \\
			0 & 0 & 1
		\end{pmatrix},\\
		v^{(3,\epsilon)}_3&=\begin{pmatrix}
			1 & 0 & 0 \\
			-\mathrm{e}^{-2i \nu_1 \log_0\left(z\right)}\delta_A^1r^*_{1,a} \mathrm{e}^{t \Phi_{21}^0(k_0;z)+\frac{iz^2}{2}} & 1 & 0 \\
			0 & 0 & 1
		\end{pmatrix},\\
		v^{(3,\epsilon)}_4&=\begin{pmatrix}
			1 & -\mathrm{e}^{2i \nu_1 \log_0\left(z\right)}(\delta_A^1)^{-1}\rho_{1,a} \mathrm{e}^{-t \Phi_{21}^0(k_0;z)-\frac{iz^2}{2}} & 0 \\
			0 & 1 & 0 \\
			0 & 0 & 1
		\end{pmatrix}.\\
	\end{aligned}
	$$
	Moreover, one also has
	$$
	\begin{aligned}
		v^{(3,\epsilon)}_7&=\begin{pmatrix}
			1 & 0 & 0 \\
			\mathrm{e}^{2i \nu_4 \log_{\pi}\left(z\right)}\left(\delta_B^1\right)^{-1}\rho_{2,a} \mathrm{e}^{t\Phi_{21}^0(-k_0;z_2)-\frac{iz^2}{2}} & 1 & 0 \\
			0 & 0 & 1
		\end{pmatrix},\\
		v^{(3,\epsilon)}_8&=\begin{pmatrix}
			1 & -\mathrm{e}^{-2i \nu_4 \log_{\pi}\left(z\right)}\delta_B^1\rho_{2,a}^{*} \mathrm{e}^{-t\Phi_{21}^0(-k_0;z_2)+\frac{iz^2}{2}} & 0 \\
			0 & 1 & 0 \\
			0 & 0 & 1
		\end{pmatrix},\\
		v^{(3,\epsilon)}_9&=\begin{pmatrix}
			1 & 0 & 0 \\
			-\mathrm{e}^{2i \nu_4 \log_{\pi}\left(z\right)}\left(\delta_B^1\right)^{-1} r_{2,a} \mathrm{e}^{t\Phi_{21}^0(-k_0;z_2)-\frac{iz^2}{2}} & 1 & 0 \\
			0 & 0 & 1
		\end{pmatrix},\\
		v^{(3,\epsilon)}_{10}&=\begin{pmatrix}
			1 & \mathrm{e}^{-2i \nu_4 \log_{\pi}\left(z\right)}\delta_B^1r_{2,a}^* \mathrm{e}^{-t\Phi_{21}^0(-k_0;z_2)+\frac{iz^2}{2}} & 0 \\
			0 & 1 & 0 \\
			0 & 0 & 1
		\end{pmatrix}.\\
	\end{aligned}
	$$
	
	When $z$ is fixed, it is observed that $r_{j,a} \to r_j(k_0)$, $\rho_{j,a} \to \frac{r_j(k_0)}{1-|r_j(k_0)|^2}$, $\delta_A^1\to 1,$ $\delta_B^1\to 1$ and $\mathrm{e}^{\pm t\Phi_{21}^0(\pm k_0;z)}\to1$ (or $\mathrm{e}^{\pm x\tilde{\Phi}_{21}^0(\pm k_0;z)}\to1$) as $t\to\infty$ for $M^{-1}\le\xi\le M$ (resp. $x\to\infty$ for $0\le\zeta\le\zeta_{\rm{max}}$), so that the jump matrix $v^{3,\epsilon} \to v^X_{A,B}$ as $t\to\infty \ \text{or}\ x\to\infty$, in which $v^{X}_{A,B}$ are the jump matrices of the model problems for $M^{X}_{A}$ and $M^{X}_{B}$ in the Appendix \ref{appendix}. 
	
	\begin{lemma}
		The matrix-valued function $H(\pm k_0,t)$ is uniformly bounded in sense of
		$$
		\sup_{t\ge 1}|\partial_x^lH(\pm k_0,t)|\le C, \quad M^{-1}\le\xi\le M,
		$$
		and
		$$
		\sup_{x\ge 1}|\partial_x^l\tilde{H}(\pm k_0,x)|\le C, \quad 0\le\zeta\le\zeta_{\rm{max}},
		$$
		for $l=0,1$. Moreover, for $(x,t)$ belong to the Sectors $\rm I$ and $\rm II$, the functions $\delta_{A,B}^0,\delta_{A,B}^1$ satisfy
		$
		|\delta_A^0|=\mathrm{e}^{2\pi\nu},\  |\delta_B^0|=1\
		$, and one has
		$$
		|\delta_{A}^1(k)-1|\le C|k-k_0|(1+|\ln|k- k_0||),\quad
		|\delta_{B}^1(k)-1|\le C|k+k_0|(1+|\ln|k+ k_0||).
		$$
		Especially, for
		$M^{-1}\le\xi\le M$ and $t\ge 1$, it follows that
		$$
		|\partial_x\delta_A^0|\le \frac{C\ln t}{t},\  |\partial_x\delta_B^0|\le \frac{C\ln t}{t},\
		|\partial_x\delta_{A}^1(k)|\le\frac{C}{t}|\ln|k-k_0||,\
		|\partial_x\delta_{B}^1(k)|\le\frac{C}{t}|\ln|k+k_0||.
		$$
		For $0\le\zeta\le \zeta_{\rm{max}}$ and $x\ge1$, it follows that
		$$
		|\partial_x\delta_A^0|\le \frac{Ck_0^4\ln x}{x},\  |\partial_x\delta_B^0|\le \frac{Ck_0^4\ln x}{x},\
		|\partial_x\delta_{A}^1(k)|\le\frac{Ck_0^4}{x}|\ln|k-k_0||,\
		|\partial_x\delta_{B}^1(k)|\le\frac{Ck_0^4}{x}|\ln|k+k_0||.
		$$
		
	\end{lemma}
	\begin{proof}
		Recalling that $\delta_A^0=\frac{a^{-2i\nu_1}\mathrm{e}^{-2\chi_1(k_0)}}{\tilde \delta_{v_1}(k_0)}$, direct calculation shows that
		$$
		|a^{-2i\nu_1}|=|({3^{\frac{5}{4}}2\sqrt{5t}k_0^{\frac{3}{2}}})^{2i\nu_1}|=|\mathrm{e}^{2i\nu_1 ln(a)}|=1,
		$$
		since the coefficients $\nu_1$ and $a$ are real, and
		$$
		|\tilde \delta_{v_1}(k_0)|=\left|\frac{\delta_3(k_0)\delta^2_4(k_0)\delta_5(k_0)}{\delta_6(k_0)\delta_2(k_0)}\right|=\left|\frac{\delta_1(\omega^2k_0)\delta^2_4(k_0)\delta_1(\omega k_0)}{\delta_4(\omega k_0)\delta_4(\omega^2 k_0)}\right|=1,
		$$
		where the fact that $\delta_{1,4}(k)={(\overline{\delta_{1,4}( \bar{k})})^{-1}}$ and the symmetries between $\delta_{1}(k)$ (resp. $\delta_{4}(k)$) and $\delta_{3,5}(k)$ (resp. $\delta_{2,6}(k)$) have been used.
		\par
		Furthermore, the real part of $\chi_1$ is written as
		$$
		\Re\chi_1(k_0)=\frac{1}{2 \pi} \int_{k_0}^{\infty} \pi \mathrm{d}\ln \left(1-\left|r_1(s)\right|^2\right)=-\frac{1}{2} \ln \left(1-\left|r_1\left(k_0\right)\right|^2\right)=\pi \nu_1,
		$$
		since the branch cut from $0$ to $2\pi$ is chosen.
		
		Thus we have
		$$
|\delta_A^0|=\left|\frac{a^{-2i\nu_1}\mathrm{e}^{-2\chi_1(k_0)}}{\delta_{\tilde v_1}(k_0)}\right|=\mathrm{e}^{-2\pi\nu_1}.
		$$
		Similarly, it is observed that
		$$
		\Re\chi_4(-k_0)=\frac{1}{2 \pi} \int_{-k_0}^{-\infty} 0 \mathrm{d}\ln \left(1-\left|r_2(s)\right|^2\right)=0,
		$$
		and
		$$
|\delta_B^0|=\left|\frac{a^{-2i\nu_4}\mathrm{e}^{-2\chi_4(-k_0)}}{\tilde\delta_{ v_4}(-k_0)}\right|=1.
		$$
		
		Moreover, the formulas indicate that
		$$
		\begin{aligned}
			\left|\partial_x \delta_A^0(\zeta, t)\right| & =\left|\delta_A^0(\zeta, t) \partial_x \ln \delta_A^0(\zeta, t)\right|=\mathrm{e}^{-2 \pi \nu_1}\left|\partial_x \ln \delta_A^0(\zeta, t)\right| \\
			& \leq C\left(\left|\ln t \partial_x \nu_1\right|+\left|\partial_x \chi_1\left(k_0\right)\right|+\left|\partial_x \ln \tilde\delta_{v_1}\left( k_0\right) \right|\right).
		\end{aligned}
		$$
        \par
		Since $k_0=\sqrt[4]{\frac{x}{45t}}$, it can be gotten that $\partial_x=\frac{1}{4k_0^3t}\partial_{k_0}$, thus it follows
		$$
		|\partial_x\nu_1| \le C \frac{1}{t} \frac{[\partial_k|r_1(k)|^2]|_{k=k_0}}{1-|r_1(k_0)|^2}\le C\frac{1}{t},\\
		\left|\partial_x \chi_1\left(k_0\right)\right|\le \frac{C}{t},\ \left|\partial_x \ln \tilde\delta_{v_1}\left( k_0\right) \right|\le \frac{C}{t} \left|\partial_{k_0} \ln \tilde\delta_{v_1}\left( k_0\right) \right|,
		$$
		since the function $\tilde\delta_{v_1}(k)$ is analytic near $k_0$.
		
		Recalling that $\delta_A^1=\frac{\mathrm{e}^{2\chi_1(k_0)-2\chi_1(k)}\delta_{\tilde v_1}(k_0)}{\delta_{\tilde v_1}(k)}$, we have
		$$
		|\mathrm{e}^{2\chi_1(k_0)-2\chi_1(k)}-1|\le C|\chi_1(k_0)-\chi_1(k)|\le C|k-k_0|(1+|\ln|k-k_0||),
		$$
		and direct calculation shows that
		$$
		\partial_x \delta_A^1(k)=\delta_A^1( k) \partial_x \log \delta_A^1( k).
		$$
        \par
		Using the fact that the function $\tilde \delta_{v_1}(k)$ is analytic near $k_0$ again and combining all the estimates above, it can be obtained that
		$$
		\left|\partial_x \delta_{A}^1(\zeta; k)\right| \leq C\left(\left|\partial_x\left(\chi_1( k)-\chi_1\left(k_0\right)\right)\right|+\frac{1}{t}\left|\partial_{k_0} \log \tilde\delta_{v_1}\right|\right)\le \frac{C |\ln|k-k_0||}{t}.
		$$
		Notice that the above estimates still hold for the case \(\zeta \in [0, \zeta_{\rm{max}}]\). Under the equality \(t = \frac{x}{45k_0^4}\), the estimate for \(\tilde{H}(\pm k_0, x)\) can be given similarly.
	\end{proof}
	In conclusion, for $k\in\Sigma_{A,B}$, we have $$M^{(2)}(x,t;k)=M^{(3,\epsilon)}(x,t;k)H(\pm k_0,t)^{-1} \to M^{X_{A,B}}(y;z)H(\pm k_0,t)^{-1}$$ as $t\to\infty$ or $x\to\infty$. But on the boundary of $\partial B_{\epsilon}(\pm k_0)$, the RH problem for $M^{X}_{A,B}H(\pm k_0,t)^{-1}$ does not converge to the identity matrix $I$ as $t\to\infty$, which suggests that a new RH problem should be introduced. To do so, define
	$$
	M^{(\pm k_0)}(x,t;k)=H(\pm k_0,t)M^{X_{A,B}}(y;z)H(\pm k_0,t)^{-1},\quad k\in B_{\epsilon}(\pm k_0),
	$$
	then the following lemma holds.
	\begin{lemma}
		The function $M^{(\pm k_0)}(x,t;k)$ is analytic for $k\in B_{\epsilon}(\pm k_0)\setminus\Sigma_{A,B}$ and satisfies the jump condition $M^{(\pm k_0)}_+=M^{(\pm k_0)}_-V^{(\pm k_0)}$ on the contours $\Sigma_{A,B}$, respectively. Moreover, for $t$ large enough and $M^{-1}\le\xi\le M$, the following estimates hold
		$$
		\|\partial_x^l(v^{(2)}-V^{(\pm k_0)})\|_{L^1(\Sigma_{A,B})}\le C\frac{\ln t}{t},\quad \|\partial_x^l(v^{(2)}-V^{(\pm k_0)})\|_{L^{\infty}(\Sigma_{A,B})}\le C\frac{\ln t}{t^{\frac{1}{2}}}.
		$$
		Furthermore, one has
		$$
		\left\|\partial_x^l\left(M^{(\pm k_0)}(x, t; \cdot)^{-1}-I\right)\right\|_{L^{\infty}\left(\partial B_{\epsilon}(\pm k_0)\right)}  =\mathcal{O}\left(t^{-1 / 2}\right), \\
		$$
		$$
		\frac{1}{2 \pi i} \int_{\partial B_{\left(\pm k_0,\epsilon\right)}}\left(M^{(\pm k_0)}(x, t; k)^{-1}-I\right) d k  =-\frac{H(\pm k_0, t) \left(M^{X_{A,B}}(y)\right)^{(1)} H(\pm k_0, t)^{-1}}{a(t)}+\mathcal{O}\left(t^{-1}\right).
		$$
		On the other hand, for $\zeta\in[0,\zeta_{\rm{max}}]$ and $x\ge1$, it follows that
		$$
		\|\partial_x^l(v^{(2)}-V^{(\pm k_0)})\|_{L^1(\Sigma_{A,B})}\le \frac{\mathrm{C}_N(\zeta)\ln x}{x},\quad \|\partial_x^l(v^{(2)}-V^{(\pm k_0)})\|_{L^{\infty}(\Sigma_{A,B})}\le \frac{\mathrm{C}_N(\zeta)\ln x}{x^{\frac{1}{2}}},
		$$
		and
		$$
		\left\|\partial_x^l\left(M^{(\pm k_0)}(x, t; \cdot)^{-1}-I\right)\right\|_{L^{\infty}\left(\partial B_{\epsilon}(\pm k_0)\right)}  =\mathcal{O}\left(\mathrm{C}_N(\zeta)x^{-1 / 2}\right), \\
		$$
		$$
		\begin{aligned}
			\frac{1}{2 \pi i} \int_{\partial B_{\left(\pm k_0,\epsilon\right)}}\left(M^{(\pm k_0)}(x, t; k)^{-1}-I\right) d k =-\frac{H(\pm k_0, t) \left(M^{X_{A,B}}(y)\right)_{1} H(\pm k_0, t)^{-1}}{a(x)}\\+\mathcal{O}\left(\mathrm{C}_N(\zeta)x^{-1}\right),
		\end{aligned}
		$$
		where $\mathrm{C}_N(\zeta)\ge0$ is a smooth function which vanishes in any order derivative at $\zeta=0$.
	\end{lemma}
	\begin{proof}
		Recall that
		$$
		M^{(k_0)}(x,t;k)=H(k_0,t)M^{X_A}(y;z)H(k_0,t)^{-1},\quad k\in B_{\epsilon}(k_0),
		$$
		where
		$$
		V^{(k_0)}(x,t;k)=H(k_0,t)v^{X_A}(y;z)H(k_0,t)^{-1},
		$$
		and
		$$
		v^{(2)}(x,t;k)=H(k_0,t)v^{(3,\epsilon)}(x,t;k)H(k_0,t)^{-1},
		$$
		thus we get that
		$$
		v^{(2)}-V^{(k_0)}=H(k_0,t)\left(v^{(3,\epsilon)}-v^{X_A}\right)H(k_0,t).
		$$
		Since $H(k_0,t)^{\pm1}$ is bounded and it is sufficient to show that
		$$
		\begin{aligned}
			& \left\|\partial_x^l\left[v^{(3,\epsilon)}(x, t; \cdot)-v^{X_A}(x, t; z(k_0, \cdot))\right]\right\|_{L^1\left(\mathcal{X}_j^\epsilon\right)} \leq C t^{-1} \ln t\ \text{or}\ \mathrm{C}_N(\zeta) x^{-1} \ln x, \\
			& \left\|\partial_x^l\left[v^{(3,\epsilon)}(x, t; \cdot)-v^{X_A}(x, t; z(k_0, \cdot))\right]\right\|_{L^{\infty}\left(\mathcal{X}_j^\epsilon\right)} \leq C t^{-1 / 2} \ln t\ \text{or}\ \mathrm{C}_N(\zeta) x^{-1/2} \ln x.
		\end{aligned}
		$$
		Indeed, the Lemma \ref{Lemma-analytic-extension} shows that for $t$ large enough and $M^{-1}\le\xi\le M$, it follows that
		$$
		\begin{aligned}
			&\left|\mathrm{e}^{-2i \nu_1 \log_0\left(z\right)}\delta_A^1 \rho^*_{1,a} \mathrm{e}^{t \Phi_{21}^0(k_0;z)+\frac{iz^2}{2}} -{\frac{\bar{y}}{1-|y|^2}} z^{-2 i \nu_1(y)} \mathrm{e}^{\frac{i z^2}{2}}\right|\\
			&=|\mathrm{e}^{-2i \nu_1 \log_0\left(z\right)}|\left|\delta_A^1 \rho^*_{1,a} \mathrm{e}^{t \Phi_{21}^0(k_0;z)}-{\frac{\bar{y}}{1-|y|^2}}\right||\mathrm{e}^{\frac{iz^2}{2}}|\\
			&\le C\left|(\delta_A^1-1)\rho^*_{1,a} \mathrm{e}^{t \Phi_{21}^0(k_0;z)}+(\mathrm{e}^{t \Phi_{21}^0(k_0;z)}-1) \rho^*_{1,a}+(\rho^*_{1,a}(k)-\rho^*_{1,a}(k_0))\right||\mathrm{e}^{\frac{iz^2}{2}}|\\
			&\le C |k-k_0|(1+|\ln|k-k_0||)\mathrm{e}^{-ct|k-k_0|^2},
		\end{aligned}
		$$
and for $\zeta\in[0,\zeta_{\rm{max}}]$ and $x\ge1$, it can also be gotten that
$$
\begin{aligned}
			&\left|\mathrm{e}^{-2i \nu_1 \log_0\left(z_1\right)}\delta_A^1 \rho^*_{1,a} \mathrm{e}^{t \Phi_{21}^0(k_0;z)+\frac{iz^2}{2}} -{\frac{\bar{y}}{1-|y|^2}} z^{-2 i \nu_1(y)} \mathrm{e}^{\frac{i z^2}{2}}\right|\\
			&\le \mathrm{C}_N(\zeta)|k-k_0|(1+|\ln|k-k_0||)\mathrm{e}^{-cx|k-k_0|^2},
\end{aligned}
$$
which imply that for $t$ large enough and $M^{-1}\le\xi\le M$, one has
		$$
		\begin{aligned}
			\left\|\left(v^{(3,\epsilon)}-v^{X_A}\right)_{21}\right\|_{L^1\left(\Sigma_A\right)} &\leq C \int_0^{\infty} s(1+|\ln s|) \mathrm{e}^{-c t s^2} d s \leq C t^{-1} \ln t,\\
		\left\|\left(v^{(3,\epsilon)}-v^{X_A}\right)_{21}\right\|_{L^{\infty}\left(\Sigma_A\right)} &\leq C \sup _{s \geq 0} s(1+|\ln s|) \mathrm{e}^{-c t s^2} \leq C t^{-1 / 2} \ln t,	
		\end{aligned}
		$$
		and for $\zeta\in[0,\zeta_{\rm{max}}]$ and $x\ge1$, one has
		$$
		\begin{aligned}
        \left\|\left(v^{(3,\epsilon)}-v^{X_A}\right)_{21}\right\|_{L^1\left(\Sigma_A\right)}&\leq \mathrm{C}_N(\zeta) \int_0^{\infty} s(1+|\ln s|) \mathrm{e}^{-c x s^2} d s \leq \mathrm{C}_N(\zeta) x^{-1} \ln x,\\
			\left\|\left(v^{(3,\epsilon)}-v^{X_A}\right)_{21}\right\|_{L^{\infty}\left(\Sigma_A\right)}&\leq \mathrm{C}_N(\zeta) \sup _{s \geq 0} s(1+|\ln s|) \mathrm{e}^{-c x s^2} \leq \mathrm{C}_N(\zeta) x^{-1 / 2} \ln x.
		\end{aligned}
		$$
		\par
        Furthermore, it is derived that
		$$
		\begin{aligned}
			&\partial_x\left(v^{(3,\epsilon)}-v^{X_A}\right)_{21}\\&=\partial_x(\mathrm{e}^{-2i\nu_1\log_0(z)})\left((\delta_A^1-1)\rho^*_{1,a} \mathrm{e}^{t \Phi_{21}^0(k_0;z)}+(\mathrm{e}^{t \Phi_{21}^0(k_0;z)}-1) \rho^*_{1,a}+(\rho^*_{1,a}(k)-\rho^*_{1,a}(k_0)\right)\mathrm{e}^{\frac{iz^2}{2}}\\
			&+\mathrm{e}^{-2i\nu_1\log_0(z)}\partial_x\left((\delta_A^1-1)\rho^*_{1,a} \mathrm{e}^{t \Phi_{21}^0(k_0;z)}+(\mathrm{e}^{t \Phi_{21}^0(k_0;z)}-1) \rho^*_{1,a}+(\rho^*_{1,a}(k)-\rho^*_{1,a}(k_0)\right)\mathrm{e}^{\frac{iz^2}{2}}\\
			&+\mathrm{e}^{-2i\nu_1\log_0(z)}\left((\delta_A^1-1)\rho^*_{1,a} \mathrm{e}^{t \Phi_{21}^0(k_0;z)}+(\mathrm{e}^{t \Phi_{21}^0(k_0;z)}-1) \rho^*_{1,a}+(\rho^*_{1,a}(k)-\rho^*_{1,a}(k_0)\right)\partial_x\mathrm{e}^{\frac{iz^2}{2}}\\
			&:={\rm I} +{\rm II}+{\rm III}.
		\end{aligned}
		$$
        \par
		For the first part ${\rm I}$, the fact that $|\partial_x\mathrm{e}^{-2i\nu_1\log_0(z)}|\le \frac{C}{t(k-k_0)}(\le\frac{\mathrm{C}_N(\zeta)}{x(k-k_0)})$ indicates that
		$$
		\begin{aligned}
		\|{\rm I}\|_{L^1(\Sigma_A)}&\le C t^{-1} \int_0^{\infty}(1+\ln s) \mathrm{e}^{-c t s^2} d s \leq C t^{-3 / 2} \ln t,\quad {\rm for}~M^{-1}\le\xi\le M,\\
			\|{\rm I}\|_{L^1(\Sigma_A)}&\le \mathrm{C}_N(\zeta) x^{-1} \int_0^{\infty}(1+\ln s) \mathrm{e}^{-c x s^2} d s \leq \mathrm{C}_N(\zeta) x^{-3 / 2} \ln x,\quad {\rm for}~ \zeta\in[0,\zeta_{\rm{max}}],\\
			\|{\rm I}\|_{L^{\infty}(\Sigma_A)} &\leq C t^{-1} \sup _{u \geq 0}(1+\ln s) \mathrm{e}^{-c t s^2} \leq C t^{-1} \ln t,\quad {\rm for}~M^{-1}\le\xi\le M,\\
			\|{\rm I}\|_{L^{\infty}(\Sigma_A)} &\leq \mathrm{C}_N(\zeta) x^{-1} \sup _{u \geq 0}(1+\ln s) \mathrm{e}^{-c x s^2} \leq \mathrm{C}_N(\zeta) x^{-1} \ln x, \quad {\rm for}~ \zeta\in[0,\zeta_{\rm{max}}].
		\end{aligned}
		$$
        \par
		For the parts ${\rm II}$ and ${\rm III}$, the same estimates can also be obtained correspondingly.
		\par
		Since
		$$
		z_1={3^{\frac{5}{4}}2\sqrt{5t}k_0^{\frac{3}{2}}}(k-k_0)={3^{\frac{1}{4}}2\sqrt{x}k_0^{-\frac{1}{2}}}(k-k_0),
		$$
		for the $k\in\partial B_{\epsilon}(k_0)$, it is obvious that $z_1\to\infty$ as $t\to\infty$ and $z_1\to\infty$ as $x\to\infty$. Combining this with the WKB expansion of $M^{X_A}$, it is found that
		$$
		\begin{aligned}
			M^{X_A}(y;z)&=I+\frac{M^{X_A}_1(y)}{{3^{\frac{5}{4}}2\sqrt{5t}k_0^{\frac{3}{2}}}(k-k_0)}+\mathcal{O}\left(\frac{1}{t}\right),\quad {\rm as}~t\to\infty,\\
			M^{X_A}(y;z)&=I+\frac{M^{X_A}_1(y)}{{3^{\frac{1}{4}}2\sqrt{x}k_0^{-\frac{1}{2}}}(k-k_0)}+\mathcal{O}\left(\frac{ \mathrm{C}_N(\zeta)}{x}\right),\quad {\rm as}~x\to\infty,
		\end{aligned}
		$$
		and further 
		$$
		\begin{aligned}
			\left(M^{(k_0)}\right)^{-1}-I&=-\frac{H(k_0, t) M_1^{X_A}(y) H(k_0, t)^{-1}}{{3^{\frac{5}{4}}2\sqrt{5t}k_0^{\frac{3}{2}}}(k-k_0)}+\mathcal{O}\left(t^{-1}\right), \quad t \rightarrow \infty\\
			\left(M^{(k_0)}\right)^{-1}-I&=-\frac{H(k_0, t) M_1^{X_A}(y) H(k_0, t)^{-1}}{{3^{\frac{1}{4}}2\sqrt{x}k_0^{-\frac{1}{2}}}(k-k_0)}+\mathcal{O}\left( \mathrm{C}_N(\zeta)x^{-1}\right), \quad x \rightarrow \infty.
		\end{aligned}
		$$
	\end{proof}
	\subsection{Local parametrix near the saddle points}
	By means of the symmetry properties of the RH problems, deform the RH problem for $M^{(\pm k_0)}$ in the way
	$$
	\tilde M^{(\pm k_0)}(x,t;k)=\mathcal{A}M^{(\pm k_0)}(x,t;\omega k)\mathcal{A}^{-1}.
	$$
	Denote $\tilde B^{(\pm k_0)}_{\epsilon}=B_{\epsilon}(\pm k_0)\cup B_{\epsilon}(\pm \omega k_0)\cup B_{\epsilon}(\pm \omega^2 k_0)$ and introduce a new RH problem with solution $\tilde M(x,t;k)$ as follows:
	$$
	\tilde M(x,t;k):=\begin{cases}\begin{aligned}
			&M^{(2)}\left(\tilde M^{(k_0)}\right)^{-1}, &&k\in \tilde B^{(k_0)}_{\epsilon},\\
			&M^{(2)}\left(\tilde M^{(-k_0)}\right)^{-1},&&k\in \tilde B^{(-k_0)}_{\epsilon},\\
			&M^{(2)}, && \text{otherelse}.
		\end{aligned}
	\end{cases}
	$$
	Moreover, the jump contour is denoted as $\tilde{\Sigma}:=\Sigma^{(2)}\cup\partial \tilde B^{(k_0)}_{\epsilon}\cup \partial \tilde B^{(-k_0)}_{\epsilon}$ (see Figure \ref{Sigmatilde})
	and the jump matrices are defined by
	$$
	\tilde{V}:=\begin{cases}\begin{aligned}
			&v^{(2)}, && k \in \tilde{\Sigma}\setminus \overline {\left(\tilde B^{(\pm k_0)}_{\epsilon}\right)},\\
			&(\tilde M^{(k_0)})^{-1}, && k \in \partial \tilde B^{(k_0)}_{\epsilon},\\
			&(\tilde M^{(-k_0)})^{-1}, && k \in \partial \tilde B^{(-k_0)}_{\epsilon},\\
			&\tilde M^{(k_0)}_{-}v^{(2)}(\tilde M^{(k_0)}_+)^{-1}, && k \in \tilde B^{(k_0)}_{\epsilon}\cap \tilde{\Sigma},\\
			&\tilde M^{(-k_0)}_{-}v^{(2)}(\tilde M^{(-k_0)}_+)^{-1}, && k \in \tilde B^{(-k_0)}_{\epsilon}\cap \tilde{\Sigma}.
			
	\end{aligned}\end{cases}
	$$
    \par
	Thus we have constructed a new RH problem for $\tilde M(x,t;k)$ that satisfies $\tilde M_+(x,t;k)=\tilde M_-(x,t;k)\tilde V$ for $k\in\tilde \Sigma$ and is analytic in $\C\setminus \tilde\Sigma$.
	
	Suppose $\tilde\Sigma_{A,B}:=\Sigma_{A,B}\cup\omega\Sigma_{A,B}\cup\omega^2\Sigma_{A,B}$ and denote
	$$
\Sigma':=\tilde{\Sigma}\setminus\left(\Sigma\cup\tilde\Sigma_{A,B}\cup\partial \tilde{B}^{(\pm k_0)}_{\epsilon}\right).
	$$

	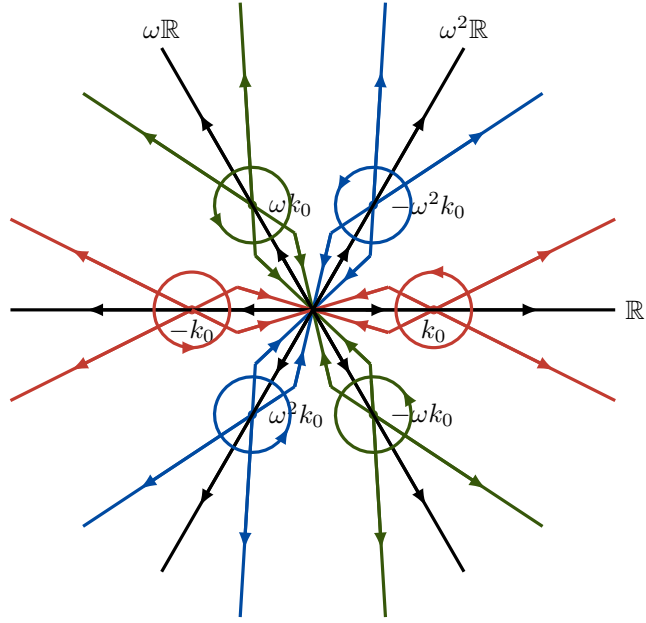
\begin{figure}[h]
		\centering
		\begin{tikzpicture}[>=latex]
			\draw[very thick] (-4,0) to (4,0) node[black,right]{$\mathbb{R}$};
			\draw[very thick] (-2,-1.732*2) to (2,1.732*2)  node[black,above]{$\omega^2\mathbb{R}$};
			\draw[very thick] (2,-1.732*2) to (-2,1.732*2)    node[black,above]{$\omega\mathbb{R}$};
			\filldraw[mred] (1.6,0) node[black,below=0.1mm]{$k_{0}$} circle (1.5pt);
			\filldraw[mred] (-1.6,0) node[black,below=0.1mm]{$-k_{0}$} circle (1.5pt);
			\filldraw[mblue] (.8,1.732*0.8) node[black,right=1mm]{$-\omega^{2}k_{0}$} circle (1.5pt);
			\filldraw[mblue] (-.8,-1.732*0.8) node[black,right=1mm]{$\omega^{2}k_{0}$} circle (1.5pt);
			\filldraw[mgreen] (.8,-1.732*0.8) node[black,right=1mm]{$-\omega k_{0}$} circle (1.5pt);
			\filldraw[mgreen] (-.8,1.732*0.8) node[black,right=1mm]{$\omega k_{0}$} circle (1.5pt);
			
			\draw[->,very thick,rotate=60,mblue] (1.6,0)  to (3.2,0.8) ;
			\draw[-,very thick,rotate=60,mblue] (1.6,0)  to (4,1.2);
			\draw[->,very thick,rotate=60,mblue] (1.6,0)  to (3.2,-0.8) ;
			\draw[->,very thick,rotate=60,mblue] (-1.6,0)  to (-3.2,0.8) ;
			\draw[-,very thick,rotate=60,mblue] (-1.6,0)  to (-4,1.2);
			\draw[->,very thick,rotate=60,mblue] (-1.6,0)  to (-3.2,-0.8);
			\draw[-,very thick,rotate=60,mblue] (-1.6,0)  to (-4,-1.2);
			\draw[-,very thick,rotate=60,mblue] (1.6,0)  to (4,-1.2);
			\draw[-,very thick,rotate=60,mblue] (-1.6,0)  to (-1.0,0.3);
			\draw[-,very thick,rotate=60,mblue] (1.6,0)  to (1.0,0.3);
			\draw[-,very thick,rotate=60,mblue] (1.0,-0.3)  to (0,0);
			\draw[->,very thick,rotate=60,mblue] (1.0,-0.3)  to (0.5,-0.15);
			\draw[-,very thick,rotate=60,mblue] (1.0,0.3)  to (0,0);
			\draw[->,very thick,rotate=60,mblue] (1.0,0.3)  to (0.5,0.15);
			\draw[-,very thick,rotate=60,mblue] (0.7,0)  to (0.9,0);
			\draw[-,very thick,rotate=60,mblue] (-1.6,0)  to (-1.2,-0.2) ;
			\draw[-,very thick,rotate=60,mblue] (-1.6,0)  to (-1.0,-0.3);
			\draw[-,very thick,rotate=60,mblue] (-1.0,-0.3)  to (0,0);
			\draw[->,very thick,rotate=60,mblue] (-1.0,-0.3)  to (-0.5,-0.15);
			\draw[-,very thick,rotate=60,mblue] (1.6,0)  to (1.0,-0.3);
			\draw[->,very thick,rotate=60,mblue] (-1.0,0.3)  to (-0.5,0.15);
			\draw[-,very thick,rotate=60,mblue] (-1.0,0.3)  to (0,0);
			
			\draw[->,very thick,rotate=120,mgreen] (1.6,0)  to (3.2,0.8) ;
			\draw[-,very thick,rotate=120,mgreen] (1.6,0)  to (4,1.2);
			\draw[->,very thick,rotate=120,mgreen] (1.6,0)  to (3.2,-0.8) ;
			\draw[->,very thick,rotate=120,mgreen] (-1.6,0)  to (-3.2,0.8) ;
			\draw[-,very thick,rotate=120,mgreen] (-1.6,0)  to (-4,1.2);
			\draw[->,very thick,rotate=120,mgreen] (-1.6,0)  to (-3.2,-0.8);
			\draw[-,very thick,rotate=120,mgreen] (-1.6,0)  to (-4,-1.2);
			\draw[-,very thick,rotate=120,mgreen] (1.6,0)  to (4,-1.2);
			\draw[-,very thick,rotate=120,mgreen] (-1.6,0)  to (-1.0,0.3);
			\draw[-,very thick,rotate=120,mgreen] (1.6,0)  to (1.0,0.3);
			\draw[-,very thick,rotate=120,mgreen] (1.0,-0.3)  to (0,0);
			\draw[->,very thick,rotate=120,mgreen] (1.0,-0.3)  to (0.5,-0.15);
			\draw[-,very thick,rotate=120,mgreen] (1.0,0.3)  to (0,0);
			\draw[->,very thick,rotate=120,mgreen] (1.0,0.3)  to (0.5,0.15);
			\draw[-,very thick,rotate=120,mgreen] (0.7,0)  to (0.9,0);
			\draw[-,very thick,rotate=120,mgreen] (-1.6,0)  to (-1.2,-0.2) ;
			\draw[-,very thick,rotate=120,mgreen] (-1.6,0)  to (-1.0,-0.3);
			\draw[-,very thick,rotate=120,mgreen] (-1.0,-0.3)  to (0,0);
			\draw[->,very thick,rotate=120,mgreen] (-1.0,-0.3)  to (-0.5,-0.15);
			\draw[-,very thick,rotate=120,mgreen] (1.6,0)  to (1.0,-0.3);
			\draw[->,very thick,rotate=120,mgreen] (-1.0,0.3)  to (-0.5,0.15);
			\draw[-,very thick,rotate=120,mgreen] (-1.0,0.3)  to (0,0);
			
			\draw[->,very thick,mred] (1.6,0)  to (3.2,0.8) ;
			\draw[-,very thick,mred] (1.6,0)  to (4,1.2);
			\draw[->,very thick,mred] (1.6,0)  to (3.2,-0.8) ;
			\draw[->,very thick,mred] (-1.6,0)  to (-3.2,0.8) ;
			\draw[-,very thick,mred] (-1.6,0)  to (-4,1.2);
			\draw[->,very thick,mred] (-1.6,0)  to (-3.2,-0.8) ;
			\draw[-,very thick,mred] (-1.6,0)  to (-4,-1.2);
			\draw[-,very thick,mred] (-1.6,0)  to (-1.0,0.3);
			\draw[->,very thick,mred] (-1.0,0.3)  to (-0.5,0.15);
			\draw[-,very thick,mred] (-1.0,0.3)  to (0,0);
			\draw[-,very thick,mred] (1.6,0)  to (4,-1.2);
			\draw[-,very thick,mred] (1.6,0)  to (1.2,0.2) ;
			\draw[-,very thick,mred] (1.6,0)  to (1.0,0.3);
			\draw[-,very thick,mred] (1.0,0.3)  to (0,0);
			\draw[->,very thick,mred] (1.0,0.3)  to (0.5,0.15);
			\draw[-,very thick,mred] (1.6,0)  to (1.2,-0.2) ;
			\draw[-,very thick,mred] (1.6,0)  to (1.0,-0.3);
			\draw[-,very thick,mred] (1.0,-0.3)  to (0,0);
			\draw[->,very thick,mred] (1.0,-0.3)  to (0.5,-0.15);
			\draw[-,very thick] (0.7,0)  to (0.9,0);
			\draw[<->,very thick] (-1,0)  to (1,0);
			\draw[<->,very thick] (-3,0)  to (3,0);
			\draw[-,very thick,mred] (-1.6,0)  to (-1.2,-0.2) ;
			\draw[-,very thick,mred] (-1.6,0)  to (-1.0,-0.3);
			\draw[-,very thick,mred] (-1.0,-0.3)  to (0,0);
			\draw[->,very thick,mred] (-1.0,-0.3)  to (-0.5,-0.15);
			\draw[-,very thick] (-0.7,0)  to (-1.4,0);
			
			\draw[<->,very thick] (-.5,-1.732*.5)  to (.5,1.732*.5);
			
			\draw[<->,very thick] (-1.5,-1.732*1.5)  to (1.5,1.732*1.5) ;
			\draw[<->,very thick] (-.5,1.732*.5)  to (.5,-1.732*.5) ;
			\draw[<->,very thick] (-1.5,1.732*1.5)  to (1.5,-1.732*1.5);
			\draw[very thick,mred] (1.6,0) circle [radius=0.5cm];
			\draw[->,very thick,mblue,rotate=60] (1.61,0.5) to (1.50,0.5);
			\draw[very thick,mblue,rotate=60] (1.6,0) circle [radius=0.5cm];
			\draw[->,very thick,mred] (1.61,0.5) to (1.50,0.5);
			\draw[->,very thick,mgreen,rotate=120] (1.61,0.5) to (1.50,0.5);
			\draw[very thick,mgreen,rotate=120] (1.6,0) circle [radius=0.5cm];
			\draw[->,very thick,mred,rotate=180] (1.61,0.5) to (1.50,0.5);
			\draw[very thick,mred,rotate=180] (1.6,0) circle [radius=0.5cm];
			\draw[->,very thick,mgreen,rotate=-60] (1.61,0.5) to (1.50,0.5);
			\draw[very thick,mgreen,rotate=-60] (1.6,0) circle [radius=0.5cm];
			\draw[->,very thick,mblue,rotate=-120] (1.61,0.5) to (1.50,0.5);
			\draw[very thick,mblue,rotate=-120] (1.6,0) circle [radius=0.5cm];
		\end{tikzpicture}
		\caption{{\protect\small
				The jump contour $\tilde{\Sigma}:=\Sigma^{(2)}\cup\partial \tilde{B}_{\epsilon}^{(\pm k_0)}$ with circles oriented anticlockwise.}}
		\label{Sigmatilde}
	\end{figure}

	\begin{lemma}
		Let $W=\tilde V-I$. The following estimates hold uniformly for $t$ large enough and $M^{-1}\le\xi\le M$
		$$
		\begin{aligned}
			& \left\|(1+|\cdot|) \partial_x^l {W}\right\|_{\left(L^1 \cap L^{\infty}\right)(\Sigma)} \leq \frac{C}{k_0^3 t}, \\
			& \left\|(1+|\cdot|) \partial_x^l W\right\|_{\left(L^1 \cap L^{\infty}\right)\left(\Sigma^{\prime}\right)} \leq C \mathrm{e}^{-c t}, \\
			& \left\|\partial_x^l W\right\|_{\left(L^1 \cap L^{\infty}\right)(\partial \tilde{B}^{(\pm k_0)})} \leq C t^{-1 / 2}, \\
			& \left\|\partial_x^l W\right\|_{L^1\left(\tilde{\Sigma}_{A,B}\right)} \leq C t^{-1} \ln t, \\
			& \left\|\partial_x^l W\right\|_{L^{\infty}\left(\tilde{\Sigma}_{A,B}\right)} \leq C t^{-1 / 2} \ln t,
		\end{aligned}
		$$
and for $\zeta\in[0,\zeta_{\rm{max}}]$ and $x$ large enough, similar estimates also hold
        $$
		\begin{aligned}
			& \left\|(1+|\cdot|) \partial_x^l {W}\right\|_{\left(L^1 \cap L^{\infty}\right)(\Sigma)} \leq \frac{\mathrm{C}_N(\zeta)}{ x}, \\
			& \left\|(1+|\cdot|) \partial_x^l W\right\|_{\left(L^1 \cap L^{\infty}\right)\left(\Sigma^{\prime}\right)} \leq Cx^{-N}, \\
			& \left\|\partial_x^l W\right\|_{\left(L^1 \cap L^{\infty}\right)(\partial \tilde{B}^{(\pm k_0)})} \leq \mathrm{C}_N(\zeta) x^{-1 / 2}, \\
			& \left\|\partial_x^l W\right\|_{L^1\left(\tilde{\Sigma}_{A,B}\right)} \leq \mathrm{C}_N(\zeta) x^{-1} \ln x, \\
			& \left\|\partial_x^l W\right\|_{L^{\infty}\left(\tilde{\Sigma}_{A,B}\right)} \leq \mathrm{C}_N(\zeta) x^{-1/2} \ln x.
		\end{aligned}
		$$
	\end{lemma}

	\begin{proof} We first prove the case that $t$ is large enough and $M^{-1}\le\xi\le M$.
    \par
		For the first inequality, notice that the jump matrix on $\Sigma$ involves the terms $r_{j,r}$ and $\rho_{j,r},~j=1,2$, and the function $(\tilde M^{(\pm k_0)})^{\pm1}$ is bounded, then we have
		$$
		\|(1+|\cdot|)\partial_x^l(v^{(2)}-I)\|_{(L^1\cap L^{\infty})(\Sigma_{5,6}^{(2)})}\le Ct^{-1}.
		$$
		So that on the cuts $\Sigma^{(2)}_{5,6}\cap B_{\epsilon}(k_0)$, it follows that
		$$
		W=\tilde V-I=M^{(k_0)}_-v^{(2)}\left(M^{(k_0)}_+\right)^{-1}-I=M^{(k_0)}_-\left(v^{(2)}-I\right)\left(M^{(k_0)}_+\right)^{-1}.
		$$
		Since the jump of the RH problem for $\tilde M^{(\pm k_0)}$ is on contours $\tilde\Sigma_{A,B}$, the function $\tilde M^{(k_0)}$ is analytic on $\Sigma^{(2)}_{5,6}\cap B_{\epsilon}(k_0)$ and is bounded. Then we have
		$$
		\left\|(1+|\cdot|) \partial_x^l {W}\right\|_{\left(L^1 \cap L^{\infty}\right)(\Sigma)} \leq \frac{C}{k_0^3 t}. \\
		$$
        \par
		For the second inequality, notice the contour $\Sigma'=\Sigma^{(2)}\setminus\overline{\tilde{B}^{(\pm k_0)}_\epsilon}$. We would like to focus on the contour $\Sigma^{(2)}\setminus B_{\epsilon}(k_0)$ and the matrix $W$ involving the entry $(v^{(2)}_1)_{21}=\frac{\delta_{1+}^2 }{\tilde\delta_{v_1}} \rho^*_{1,a} \mathrm{e}^{t \Phi_{21}}\neq 0$.
		Because the functions $\partial_x^l\delta_j~(j=1,2,\cdots,6)$ are bounded, the estimate of $\rho^*_{1,a}$ is
		$$
		\left|\partial_x \rho_{1, a}^*(x, t; k)\right| \leq \frac{C\mathrm{e}^{t\Re \Phi_{21}( k)}}{1+|k|}.
		$$
		Moreover, it is seen that $\Re\Phi_{21}<-c$ for $|k-k_0|>\epsilon$, so that the following inequality holds
		$$
		\left\|(1+|\cdot|) \partial_x^l W\right\|_{\left(L^1 \cap L^{\infty}\right)\left(\Sigma^{\prime}\right)} \leq C \mathrm{e}^{-c t}. \\
		$$
        \par
		The third inequality is reached by direct outcome of the above lemmas.
        \par
		For the last inequality, noticing that
		$$
		W=\tilde M^{(k_0)}_{-}(v^{(2)}-V^{(k_0)})(\tilde M^{(k_0)}_+)^{-1},\quad k\in\tilde \Sigma_{A},
		$$
 it is found that the function $M^{(k_0)}$ is bounded uniformly for $M^{-1}\le\xi\le M$. For \(\zeta \in [0, \zeta_{\rm{max}}]\), the proof of the inequalities in this lemma follows a similar approach.
	\end{proof}

	Now, introduce the Cauchy operator
	$$
	\left({C} f\right)(z)=\int_{\tilde\Sigma} \frac{f(\zeta)}{\zeta-z} \frac{\mathrm{d} \zeta}{2 \pi i}, \quad z \in \C\setminus \tilde \Sigma.
	$$
	If $(1+|z|)^{\frac{1}{3}}f(z)\in L^3(\tilde\Sigma)$, then $(Cf)(z)$ is analytic from $\C\setminus\tilde\Sigma$ to $\C$ with property that for any component $D$ in $\C\setminus \tilde \Sigma$, there are curves $\{\mathrm{C}_n \}_{n=1}^\infty$ which surround each compact subset of $D$ satisfying
	$$
	\sup_{n\ge1}\int_{\mathrm{C}_n}(1+|z|)|f(z)|^3|dz|<\infty.
	$$
	Moreover, $C_{\pm}f$ exist a.e. for $z\in\tilde\Sigma$ and $(1+|z|)^{\frac{1}{3}}C_{\pm}f(z)\in L^3(\tilde\Sigma)$.
	
	On one hand, the $C_{\pm}$ are bounded operators from weighted space $L^3(\tilde\Sigma)$ to itself (thus denote it as $\dot L^{3}(\tilde\Sigma)$), which satisfy $C_+-C_-=I$.
	
	On the other hand, recall the estimates for $l=0,1$
	$$
	\begin{cases}\begin{aligned}
			&\|(1+|\cdot|)\partial_x^lW\|_{L^1(\tilde\Sigma)}\le C t^{-\frac{1}{2}},\\
			&\|(1+|\cdot|)\partial_x^lW\|_{L^{\infty}(\tilde\Sigma)}\le C t^{-\frac{1}{2}}\ln t.\\
	\end{aligned}\end{cases}
	$$
    \par
	Then the Riesz interpolation inequality yields that
	$$
	\|(1+|\cdot|)\partial_x^lW\|_{L^p(\tilde\Sigma)}\le C t^{-\frac{1}{2}}(\ln t)^{\frac{1}{p}},
	$$
	so that $W$ belongs to the weighted space $L^3(\tilde\Sigma)$ and $L^{\infty}(\tilde\Sigma)$.

	Define the map $C_W:\dot L^3(\tilde\Sigma)+ L^{\infty}(\tilde\Sigma)\to \dot L^3(\tilde\Sigma)$ by
	$$
	C_{W} f=C_{+}\left(fW_{-}\right)+C_{-}\left(fW_{+}\right),
	$$
    then the following lemma holds.

	\begin{lemma}\label{invertible}
		For $t$ large enough and $M^{-1}<\xi<M$, the operator $I-C_W$ is invertible and $(I-C_W)^{-1}$is a bounded linear operator from $\dot L^3(\tilde\Sigma)$ to itself.
	\end{lemma}
	
	\begin{proof}
		Since $C_{\pm}$ are bounded operators from weighted space $L^3(\tilde\Sigma)$ to itself, then for any $f\in\dot L^3(\tilde\Sigma)$, we have
		$$
		\begin{aligned}	C_Wf&={C}_{+}\left(fW_{-}\right)+{C}_{-}\left(fW_{+}\right)\\
			&\le \left(\|C_+\|_{\dot L^3(\tilde\Sigma)\to\dot L^3(\tilde\Sigma)}+\|C_-\|_{\dot L^3(\tilde\Sigma)\to\dot L^3(\tilde\Sigma)}\right)\|W\|_{L^\infty(\tilde\Sigma)}\|f\|_{\dot L^3(\tilde\Sigma)}.\\
		\end{aligned}
		$$
		Then $\|C_W\|_{\dot L^3(\tilde\Sigma)\to\dot L^3(\tilde\Sigma)}\le \left(\|C_+\|_{\dot L^3(\tilde\Sigma)\to\dot L^3(\tilde\Sigma)}+\|C_-\|_{\dot L^3(\tilde\Sigma)\to\dot L^3(\tilde\Sigma)}\right)\|W\|_{L^\infty(\tilde\Sigma)}$, and by the estimate above, it follows that for $l=0,1$
		$$
		\|(1+|\cdot|)\partial_x^lW\|_{L^{\infty}(\tilde\Sigma)}\le C t^{-\frac{1}{2}}(\ln t),\quad t\to \infty.\\
		$$
		Thus $\|W\|_{L^\infty(\tilde\Sigma)}<\frac{1}{\left(\|C_+\|_{\dot L^3(\tilde\Sigma)\to\dot L^3(\tilde\Sigma)}+\|C_-\|_{\dot L^3(\tilde\Sigma)\to\dot L^3(\tilde\Sigma)}\right)}$ holds, then the operator $I-C_W$ is invertible.
	\end{proof}
	\begin{remark}
		For $\zeta\in[0,\zeta_{\rm{max}}]$ and $x$ large enough, the Lemma \ref{invertible} still holds and the proof is similar, just replacing 
        $\|(1+|\cdot|)\partial_x^lW\|_{L^{\infty}(\tilde\Sigma)}\le C t^{-\frac{1}{2}}(\ln t)$ with $\|(1+|\cdot|)\partial_x^lW\|_{L^{\infty}(\tilde\Sigma)}\le \mathrm{C}_N(\zeta) x^{-\frac{1}{2}}(\ln x)$.
	\end{remark}
	Let $\mu\in I+\dot L^3(\tilde \Sigma)$ satisfy the integral equation $\mu=I+C_W\mu,$ then one has $\mu=I+(I-C_W)^{-1}C_WI$.

	\begin{lemma}
		For $t$ large enough and $M^{-1}<\xi<M$ or $\zeta\in[0,\zeta_{\rm{max}}]$ and $x$ large enough, the RH problem for the function $\tilde M(x,t;k)$ has a unique solution of the form
		$$
		\tilde M(x,t;k)=I+C(\mu W)=I+\int_{\tilde\Sigma} \frac{\mu(x,t;\zeta)W(x,t;\zeta)}{\zeta-k} \frac{\mathrm{d} \zeta}{2 \pi i}, \quad k \in \C\setminus \tilde \Sigma.
		$$
	\end{lemma}

	\begin{lemma}
		For $t$ large enough, $M^{-1}<\xi<M$ and for $1\le p\le\infty$, it is found that
		$$
		\|\partial_x^l(\mu-I)\|_{L^p(\tilde \Sigma)}\le  {\frac{C(\ln t)^{\frac{1}{p}}}{t^{\frac{1}{2}}}},\quad l=0,1.
		$$
        Moreover, for $x$ large enough and $\zeta\in[0,\zeta_{\rm{max}}]$, it follows that
        $$
		\|\partial_x^l(\mu-I)\|_{L^p(\tilde \Sigma)}\le  {\frac{\mathrm{C}_N(\zeta)(\ln x)^{\frac{1}{p}}}{x^{\frac{1}{2}}}},\quad l=0,1.
		$$
        
	\end{lemma}
	
	\begin{proof}
		Denote $\|C_{\pm}\|_p:=\left(\|C_+\|_{L^p(\tilde\Sigma)\to L^p(\tilde\Sigma)}+\|C_-\|_{ L^p(\tilde\Sigma)\to L^p(\tilde\Sigma)}\right)$ and assume $t$ large enough to satisfy $\|W\|_{L^{\infty}(\tilde\Sigma)}<\|C_{\pm}\|_p^{-1}$. When $l=0$, we have
		$$
		\begin{aligned}
			\|\mu-I\|_{L^p(\tilde\Sigma)}&\le \sum_{j=1}^\infty\|C_{\pm}\|^{j}_{p}\|W\|^{j-1}_{L^\infty(\tilde \Sigma)}\|W\|_{L^p(\tilde\Sigma)}=\frac{\|C_{\pm}\|_p \|W\|_{L^p(\tilde\Sigma)}}{1-\|C_{\pm}\|_p\|W\|_{L^{\infty}(\tilde\Sigma)}}.
		\end{aligned}
		$$
		So combining the estimate of $\|W\|_{L^p(\tilde\Sigma)}$, the estimate for $l=0$ holds immediately. 
		
		When $l=1$, it can be gotten that
		$
		\partial_x(\mu-I)=\partial_x\sum_{j=1}^\infty(C_W)^jI.
		$
		Since the series on the right hand side is uniformly bounded and the order of sum and derivative can be changed, then we have
		$$
		\begin{aligned}
			\|\partial_x(\mu-I)\|_{L^p(\tilde\Sigma)}&\leq \sum_{j=2}^{\infty}(j-1)\left\|{{C}}_W\right\|_{L^p(\tilde{\Sigma})\to L^p(\tilde{\Sigma})}^{j-2}\left\|\partial_x C_W\right\|_{L^p(\tilde{\Sigma})\to L^p(\tilde{\Sigma})}\left\|C_W I\right\|_{L^p(\tilde{\Sigma})}\\
&+\sum_{j=1}^{\infty}\left\|C_W\right\|_{L^p(\tilde{\Sigma})\to L^p(\tilde{\Sigma})}^{j-1}\left\|\partial_x {C_W} I\right\|_{L^p(\tilde{\Sigma})}\\
			&\le C\frac{\|\partial_xW\|_{L^\infty(\tilde\Sigma)}\|W\|_{L^p(\tilde\Sigma)}+ \|\partial_xW\|_{L^p(\tilde\Sigma)}}{1-\|C_{\pm}\|_p\|W\|_{L^\infty(\tilde\Sigma)}}.
		\end{aligned}
		$$
	\end{proof}

	Now, the following non-tangential limit holds for $k\to\infty$
	$$
	Q(x,t):=\lim_{k\to\infty}k(\tilde M(x,t;k)-I)=-\frac{1}{2\pi i}\int_{\tilde\Sigma}\mu (x,t;k)W(x,t;k)\mathrm{d}k.
	$$

	\begin{lemma}
		For $M^{-1}\le\xi\le M$ and $t\to\infty$, the asymptotics for $Q(x,t)$ is formulated as
		$$
		Q(x,t)=-\frac{1}{2\pi i}\int_{\partial\tilde B^{(k_0)}_\epsilon\cup\partial\tilde B^{(-k_0)}_\epsilon}W(x,t;k)\mathrm{d}k+\mathcal{O}\left(\frac{\ln t}{t}\right).
		$$
		Furthermore, for $\zeta\in[0,\zeta_{\rm{max}}]$ and $x\to\infty$, it becomes 
        $$
		Q(x,t)=-\frac{1}{2\pi i}\int_{\partial\tilde B^{(k_0)}_\epsilon\cup\partial\tilde B^{(-k_0)}_\epsilon}W(x,t;k)\mathrm{d}k+\mathcal{O}\left(x^{-N}+\frac{C_{N}(\zeta)\ln x}{x}\right).
		$$
        
	\end{lemma}

	\begin{proof}
		Decompose $Q(x,t)$ as  
		$$
		Q(x,t)=-\frac{1}{2\pi i}\int_{\partial\tilde B^{(k_0)}_\epsilon\cup\partial\tilde B^{(-k_0)}_\epsilon}W(x,t;k)\mathrm{d}k+Q_1(x,t)+Q_2(x,t),
		$$
		where
		$$
		Q_1(x,t):=-\frac{1}{2\pi i}\int_{\tilde\Sigma}(\mu (x,t;k)-I)W(x,t;k)\mathrm{d}k,\
		Q_2(x,t):=-\frac{1}{2\pi i}\int_{\tilde\Sigma\setminus{(\partial\tilde B^{(k_0)}_\epsilon}\cup \partial\tilde B^{(k_0)}_\epsilon)}W(x,t;k)\mathrm{d}k.
		$$
		For the function $Q_1(x,t)$, the H\"{o}lder inequality indicates that
		$$
		|Q_1(x,t)|\le C \|\mu (x,t;\cdot)-I\|_{L^p(\tilde\Sigma)}\|W(x,t;\cdot)\|_{L^q(\tilde\Sigma)}\le\frac{C\ln t}{t},\quad {\rm for}~M^{-1}\le\xi\le M,~t\to\infty,
		$$
        $$
		|Q_1(x,t)|\le C \|\mu (x,t;\cdot)-I\|_{L^p(\tilde\Sigma)}\|W(x,t;\cdot)\|_{L^q(\tilde\Sigma)}\le\frac{\mathrm{C}_N(\zeta)\ln x}{x},\quad {\rm for}~\zeta\in[0,\zeta_{\rm{max}}],~x\to\infty,
		$$
		where $\frac{1}{p}+\frac{1}{q}=1$. For the function $Q_2(x,t)$, the estimates below hold
		$$
		|Q_2(x,t)|\le C\|W(x,t;\cdot)\|_{L^1(\tilde \Sigma \setminus(\partial\tilde B^{(k_0)}_\epsilon\cup\partial\tilde B^{(-k_0)}_\epsilon))}\le\frac{C\ln t}{t},\quad {\rm for}~M^{-1}\le\xi\le M,~t\to\infty,
		$$
        $$
		|Q_2(x,t)|\le C\|W(x,t;\cdot)\|_{L^1(\tilde \Sigma \setminus(\partial\tilde B^{(k_0)}_\epsilon\cup\partial\tilde B^{(-k_0)}_\epsilon))}\le x^{-N},\quad {\rm for}~\zeta\in[0,\zeta_{\rm{max}}],~x\to\infty.
		$$
        \par
		Now, suppose
		$$
		R(x,t;\pm k_0):=-\frac{1}{2\pi i}\int_{\partial B_{\epsilon}(\pm k_0)}W(x,t;k)\mathrm{d}k=-\frac{1}{2\pi i}\int_{\partial B_{\epsilon}(\pm k_0)}((M^{(k_0)})^{-1}-I) \mathrm{d}k,
		$$
		and then it yields that
		$$
		R(x,t; k_0)=\frac{H( k_0, t) M_1^{X_{A}}(y(k_0) H( k_0, t)^{-1}}{{3^{\frac{5}{4}}2\sqrt{5t}k_0^{\frac{3}{2}}}}+\mathcal{O}\left(t^{-1}\right), \quad {\rm as}~t\to\infty,
		$$
        $$
		R(x,t; k_0)=\frac{H( k_0, t) M_1^{X_A}(y(k_0) H( k_0, t)^{-1}}{{3^{\frac{5}{4}}2\sqrt{5t}k_0^{\frac{3}{2}}}}+\mathcal{O}\left(\frac{ \mathrm{C}_N(\zeta)}{x}\right),\quad {\rm as}~x\to\infty,
		$$
        and for the case $R(x,t;-k_0)$, just replacing $M_1^{X_A}(y)$ into $M_1^{X_B}(y)$ and $H(k_0,t)$ int0 $H(-k_0,t)$.
        \par
		Reminding the symmetry of the function $\tilde M(x,t;k)$, one has
		$$
		\tilde M(x,t;k)=\mathcal{A}\tilde M(x,t;\omega k)\mathcal{A}^{-1}, \quad k\in\C\setminus\tilde \Sigma.
		$$
		Notice that $\mu$ and $W$ also satisfy this symmetry, then it can be found that
		$$
		\begin{aligned}
			-\frac{1}{2\pi i}\int_{\partial\tilde B^{(k_0)}_\epsilon\cup\partial\tilde B^{(-k_0)}_\epsilon}W(x,t;k)\mathrm{d}k&=-\frac{1}{2\pi i}
			\int_{\cup_{j=0}^{2}\partial B_{\epsilon}(\pm \omega^jk_0)}
			W(x,t;k)\mathrm{d}k\\
			&=R(x,t;\pm k_0)+\omega \mathcal{A}^{-1}R(x,t;\pm k_0)\mathcal{A}+\omega^2 \mathcal{A}^{-2}R(x,t;\pm k_0)\mathcal{A}^2,
		\end{aligned}
		$$	
        which immediately gives the asymptotic formulas in this lemma.
	\end{proof}
    
	\par
	\noindent{\bf Asymptotic behaviors of the SK and mSK equations in Sectors ${\rm I}$ and ${\rm II}$}. The reconstruction formula of the SK equation (\ref{SK}) is given in (\ref{recover-formula SK}), i.e.,
	$$
	u(x,t)=-\frac{1}{2}\partial_x\left(\lim_{k\to\infty}k(N_3(x,t;k)-1) \right),
	$$
	where $N(x,t;k)=\left(N_1,N_2,N_3\right)=(\omega,\omega^2,1)M(x,t;k)$. Recall that for $k\in\C\setminus{(\tilde B^{(k_0)}_{\epsilon}\cup\tilde B^{(-k_0)}_{\epsilon})}$, the function $M(x,t;k)$ is related with the function $\tilde M(x,t;k)$ by
	$$
	M=\tilde MG^{-1}\Delta^{-1}.
	$$
	Then it follows that
	$$
	\begin{aligned}
		u(x,t)&=-\frac{1}{2}\partial_x\left(\lim_{k\to\infty}k[((\omega,\omega^2,1)\tilde MG^{-1}\Delta^{-1})_3-1] \right)\\
		&=-\frac{1}{2}\partial_x\left(\lim_{k\to\infty}k[((\omega,\omega^2,1)\tilde M)_3-1] \right)+\mathcal{O}\left(\frac{\ln t}{t}\right),\ \text{as}\  t\to \infty,
	\end{aligned}
	$$
    and 
       $$
		u(x,t)=-\frac{1}{2}\partial_x\left(\lim_{k\to\infty}k[((\omega,\omega^2,1)\tilde M)_3-1] \right)+\mathcal{O}\left(\frac{\mathrm{C}_N(\zeta)\ln x}{x}+x^{-N}\right), \ {\rm as}~ x\to \infty.
	$$
    \par
	For the second equality, since the $G^{-1}\Delta^{-1}$ tends to $I$ as $k\to\infty$ and their derivatives are dominated by $\frac{\ln t}{t}$ or $\frac{\mathrm{C}_N(\zeta)\ln x}{x}+x^{-N}$ for $t\to \infty$ and $x\to \infty$, respectively, it is concluded that for $t\to \infty$
	\begin{align*}
		u(x,t)=&-\frac{1}{2}\partial_x\left(\begin{pmatrix}
			\omega & \omega^2 & 1
		\end{pmatrix} \frac{\sum_{j=0}^2 \omega^j \mathcal{A}^{-j} H(k_0, t) M_1^{X_A}(y(k_0)) H(k_0, t)^{-1} \mathcal{A}^j}{{3^{\frac{5}{4}}2\sqrt{5t}k_0^{\frac{3}{2}}}}\right)\\
		&-\frac{1}{2}\partial_x\left(\begin{pmatrix}
			\omega & \omega^2 & 1
		\end{pmatrix}
		\frac{\sum_{j=0}^2 \omega^j \mathcal{A}^{-j} H(-k_0, t) M_1^{X_B}(y(-k_0)) H(-k_0, t)^{-1} \mathcal{A}^j}{{3^{\frac{5}{4}}2\sqrt{5t}k_0^{\frac{3}{2}}}}\right)+\mathcal{O}\left(\frac{\ln t}{t}\right)\\
		=&-\frac{1}{3^{\frac{5}{4}}2\sqrt{5t}k_0^{\frac{3}{2}}}\left(\partial_x\Re\left(\omega^2\beta^A_{21}\delta^0_A\mathrm{e}^{t\Phi_{21}(k_0)}\right)+\partial_x\Re\left(\omega\beta^B_{12}\delta^0_B\mathrm{e}^{-t\Phi_{21}(-k_0)}\right)\right)+\mathcal{O}\left(\frac{\ln t}{t}\right),\\
	\end{align*}
    while for $x\to \infty$
    $$
    u(x,t)=-\frac{\left(\partial_x\Re\left(\omega^2\beta^A_{21}\delta^0_A\mathrm{e}^{t\Phi_{21}(k_0)}\right)+\partial_x\Re\left(\omega\beta^B_{12}\delta^0_B\mathrm{e}^{-t\Phi_{21}(-k_0)}\right)\right)}{3^{\frac{5}{4}}2\sqrt{5t}k_0^{\frac{3}{2}}}
    \\+\mathcal{O}\left(\frac{\mathrm{C}_N(\zeta)\ln x}{x}+x^{-N}\right).
    $$
	\begin{theorem}
		Suppose $u(x,t)$ is the solution of the SK equation (\ref{SK}) with initial data $u(x,0)=u_0(x)$ in Schwartz space and the Assumption \ref{solitonless} and \ref{r less 1} hold, then in the generic case, for $1/M\le \xi \le M$ with $x>0$, the solution $u(x,t)$ in Sector ${\rm II}$ of Theorem \ref{main-theorem-SK} has the following asymptotics as $t\to\infty$
		\begin{equation}\label{longtime-solution}
			\begin{split}
				\begin{aligned}
					u(x,t)=&-\frac{1}{3^{\frac{3}{4}}2\sqrt{5t}k_0^{\frac{1}{2}}}\left[\sqrt{\nu_1}\sin \left(\frac{19\pi }{12}-\left(\arg y_1+\arg\Gamma(i\nu_1)\right)-(36\sqrt{3}tk_0^5)+\nu_1\ln({{3^{\frac{7}{2}}20tk_0^{5}}})+s_1\right)\right.\\
					&\left.+\sqrt{\nu_4}\sin \left(\frac{11\pi }{12}-\left(\arg y_4+\arg\Gamma(i\nu_4)\right)-(36\sqrt{3}tk_0^5)+\nu_4\ln({{3^{\frac{7}{2}}20tk_0^{5}}})+s_2\right)\right]+\mathcal{O}\left(\frac{\ln t}{t}\right),\\
				\end{aligned}
			\end{split}
		\end{equation}
		with $s_1=\nu_4\ln(4)+\frac{1}{\pi }\int_{-k_0}^{-\infty}\log_{\pi}\frac{|s-\omega k_0|}{|s- k_0|}d\ln(1-|r_2(s)|^2)+\frac{1}{\pi }\int_{k_0}^\infty\log_0\frac{|s-k_0|}{|s-\omega k_0|}d\ln(1-|r_1(s)|^2)$, and $s_2=\nu_1\ln(4)+\frac{1}{\pi }\int_{k_0}^{\infty}\log_{0}\frac{|s+\omega k_0|}{|s+ k_0|}d\ln(1-|r_1(s)|^2)+\frac{1}{\pi }\int_{-k_0}^{-\infty}\log_{\pi}\frac{|s+ k_0|}{|s+\omega k_0|}d\ln(1-|r_2(s)|^2)$.
        \par
        Moreover, for $\zeta\in[0,\zeta_{\rm{max}}]$ and $x\to\infty$, the leading-order term of $u(x,t)$ in Sector ${\rm I}$ of Theorem \ref{main-theorem-SK} is the same as that in (\ref{longtime-solution}), but the error term should be replaced with $\mathcal{O}\left(\frac{\mathrm{C}_N(\zeta)\ln x}{x}+x^{-N}\right)$.
	\end{theorem}
	
	\begin{proof}
		To be specific, denote
		$$
		y_1=r_1(k_0),\ \nu_1=-\frac{1}{2 \pi} \ln \left(1-\left|r_1\left(k_0\right)\right|^2\right),\ \nu_4=-\frac{1}{2 \pi} \ln \left(1-\left|r_2\left(-k_0\right)\right|^2\right),
		$$
		$$
		 y_4=r_2(-k_0),\ \beta^A_{21}=\sqrt{\nu_1}\mathrm{e}^{ i\left(\frac{\pi}{4}-\arg y-\arg\Gamma(i\nu_1)\right)},\quad \mathrm{e}^{t\Phi_{21}(k_0)}=\mathrm{e}^{-i(36\sqrt{3}tk_0^5)}.
		$$
		
		Recall
		$
		\delta^0_A=\frac{a^{-2i\nu_1}\mathrm{e}^{-2\chi_1(k_0)}}{\tilde\delta_{v_1}(k_0)}
		$
		and notice that
		$$
		a^{-2i\nu_1}=\exp\left(i\nu_1\ln({{3^{\frac{5}{2}}20tk_0^{3}}})\right),\quad
		\mathrm{e}^{-2\chi_1(k_0)}=\exp\left(-\frac{1}{\pi i}\int_{k_0}^\infty\log_0|s-k_0|d\ln(1-|r_1(s)|^2)\right).
		$$
		On the other hand, it can be calculated that
		\begin{align*}	&\delta_3(k_0)\delta_5(k_0)=\delta_1(\omega^2k_0)\delta_1(\omega k_0)=\exp\left(-i\nu_1\ln(3k_0^2)-\frac{1}{\pi i}\int_{k_0}^\infty\log_0|s-\omega k_0|d\ln(1-|r_1(s)|^2)\right),\\
			&\delta_2(k_0)\delta_6(k_0)=\delta_4(\omega^2k_0)\delta_4(\omega k_0)=\exp\left(-i\nu_4\ln(k_0^2)-\frac{1}{\pi i}\int_{-k_0}^{-\infty}\log_{\pi}|s-\omega k_0|d\ln(1-|r_2(s)|^2)\right),\\
			&\delta_4^2(k_0)=\exp\left(-i\nu_4\ln(4k_0^2)-\frac{1}{\pi i}\int_{-k_0}^{-\infty}\log_{\pi}|s- k_0|d\ln(1-|r_2(s)|^2)\right).
		\end{align*}
		The computation near the saddle point \(-k_0\) is similar and quite tedious. Finally, by incorporating the aforementioned computations into the formulas and performing the complex calculations, the desired results can be obtained.
	\end{proof}
	Regarding to the asymptotic solution $w(x,t)$ of the mSK equation (\ref{msk-equation}) in Sectors $\rm I$ and $\rm II$, following the similar way of Deift-Zhou steepest-descent analysis, together with the reconstruction formula (\ref{recover-formula mSK}), it is immediate that
	$$
	w(x,t)=3\lim_{k\to\infty}k(\tilde mG^{-1}\Delta^{-1})_{13} =3\lim_{k\to\infty}k\tilde m_{13}+\mathcal{O}\left(\frac{\ln t}{t}\right),\ \text{for}\  t\to \infty,
	$$
    and
    $$
	w(x,t)=3\lim_{k\to\infty}k(\tilde mG^{-1}\Delta^{-1})_{13} =3\lim_{k\to\infty}k\tilde m_{13}+  \mathcal{O}\left(\frac{\mathrm{C}_N(\zeta)\ln x}{x}+x^{-N}\right), \ \text{for}\  x\to \infty.
	$$
    \par
    Consequently, the long-time asymptotics of the solution to the mSK equation (\ref{msk-equation}) is formulated below
	\begin{equation}\label{longtime-solution-msk}
		\begin{split}
			\begin{aligned}
				w(x,t)=&\frac{-1}{3^{\frac{1}{4}}2\sqrt{5t}k_0^{\frac{3}{2}}}\left[\sqrt{\tilde{\nu}_1}\cos \left(\frac{19\pi }{12}-\left(\arg \tilde{y}_1+\arg\Gamma(i\tilde{\nu}_1)\right)-(36\sqrt{3}tk_0^5)+\tilde{\nu}_1\ln({{3^{\frac{7}{2}}20tk_0^{5}}})+\tilde{s}_1\right)\right.\\	
				&+\left.\sqrt{\tilde{\nu}_4}\cos \left(\frac{11\pi }{12}-\left(\arg \tilde{y}_4+\arg\Gamma(i\tilde{\nu}_4)\right)-(36\sqrt{3}tk_0^5)+\tilde{\nu}_4\ln({{3^{\frac{7}{2}}20tk_0^{5}}})+\tilde{s}_2\right)\right]\\
                &+\mathcal{O}\left(\frac{\ln t}{t}\right).\\
			\end{aligned}
		\end{split}
	\end{equation}
 Moreover, the leading-order term of the large space solution $w(x,t)$ is the same as that in (\ref{longtime-solution-msk}), but the error term should be replaced with $\mathcal{O}\left(\frac{\mathrm{C}_N(\zeta)\ln x}{x}+x^{-N}\right)$.
    \par
	In fact, the RH problem for $M(x,t;k)$ and $m(x,t;k)$ can be factorized into the same model problem for $M^{X_{A,B}}$. However, the reconstruction formulas and the reflection coefficients are different in general.

\section{\bf The Painlev\'{e} region}\label{painleve region}
\par	
It is observed from Figures \ref{msk-comparisons} and \ref{SK-comparisons} that the long-time asymptotic solutions (\ref{longtime-solution})-(\ref{longtime-solution-msk}) in Sector ${\rm II}$ of the SK equation (\ref{SK}) and mSK equation (\ref{msk-equation}) are invalid near $x=0$. As the case of KdV equation \cite{Deift-Zhou-1994} and mKdV equation \cite{Deift-Zhou-1993} and motivated by the self-similar transformation from the mSK equation (\ref{msk-equation}) to the Painlev\'{e} transcendent equation (\ref{4th-Painleve}), it is conjectured that a region that can be described by Painlev\'{e} type equations may appear in the region around $x=0$. 
Since the reflection coefficients of the SK equation (\ref{SK}) satisfy $|r_j(0)|=1$ for $j=1,2$, the analysis of this region is very complicated. Fortunately, there is a Miura transformation between the SK equation  (\ref{SK}) and the mSK equation (\ref{msk-equation}), and 
in generic case, the absolute value of the reflection coefficients for the mSK equation (\ref{msk-equation}) are strictly less than 1. Thus in this section, we focus on the long-time asymptotic analysis of the mSK equation (\ref{msk-equation}) in Painlev\'{e} region. 
\par
For convenience, rewrite the mSK equation (\ref{msk-equation}) by taking $t\to-t$, which becomes
\begin{equation}\label{rewrite-msk}
	w_t=w_{x x x x x}-\left(5 w_x w_{x x}+5 w w_x^2+5 w^2 w_{x x}-w^5\right)_x.
\end{equation}
The jump matrices of the RH problem associated with the mSK equation (\ref{rewrite-msk}) are similar to that in (\ref{vn-jump-matrix}) except for the phase functions $\theta_{ij}~(1\le j<i\le 3)$. More precisely, the phase functions corresponding to (\ref{rewrite-msk}) are
$$
\theta_{i j}(x,t;k)=-t\left[\left(l_i-l_j\right) \xi+ \left(z_i-z_j\right) \right]:=t\Phi_{ij}(\xi;k)
$$
with $\xi:=-{x}/{t}$, $l_j(k)= \omega^j k$ and $z_j(k)=9 \omega^{2j} k^5$ for $j=1,2,3$. Notice that the transformation $t\to-t$ only changes the sign of the phase functions (see Figure \ref{regionU painleve} for the sign signature of $\Re \Phi_{21}$), thus we still adopt the same symbols of $\theta_{ij}, \xi$ and $\Phi_{ij}$ as that in Section ${\rm II}$. According to the self-similar transformation $w(x,t)=(5t)^{-\frac{1}{5}}p(s)$ with $s=\frac{x}{(5t)^{\frac{1}{5}}}$ from the mSK equation to the fourth-order analogues of Painlev\'{e} transcendent (\ref{4th-Painleve}), in what follows, we constrain ourself on the region $|\frac{x}{t^{1/5}}|<M$. Because of the transformation $t\to-t$, it is known that the case of $x<0$ for equation (\ref{rewrite-msk})  corresponds to the case of $x>0$ for equation (\ref{msk-equation}), and vice versa.  

\subsection{Painlev\'{e} region for $x<0$}

Firstly, consider the case of $x<0$, which implies that the critical points lie on the real line.  To be specific, the saddle points of phase function $\theta_{21}$ are $\pm k_0=\pm\sqrt[4]{\frac{-x}{45t}}$, and it is immediate that $k_0\sim t^{-\frac{1}{5}}$ as $t\to\infty$. Consequently, it is reasonable to adopt the new variable $\lambda:=k(5t)^{\frac{1}{5}}$ as the parameter of the model problem in the analysis of long-time asymptotics. Initially, decompose the reflection coefficients $\tilde{r}_{j}(k)~(j=1,2)$ into $\tilde{r}_{j,a}(k)+\tilde{r}_{j,r}(k)$ likewise the case in Sector $\rm II$, see Lemma \ref{Lemma-analytic-extension}.
	\begin{figure}[!h]
		\centering
		\begin{overpic}[width=.5\textwidth]{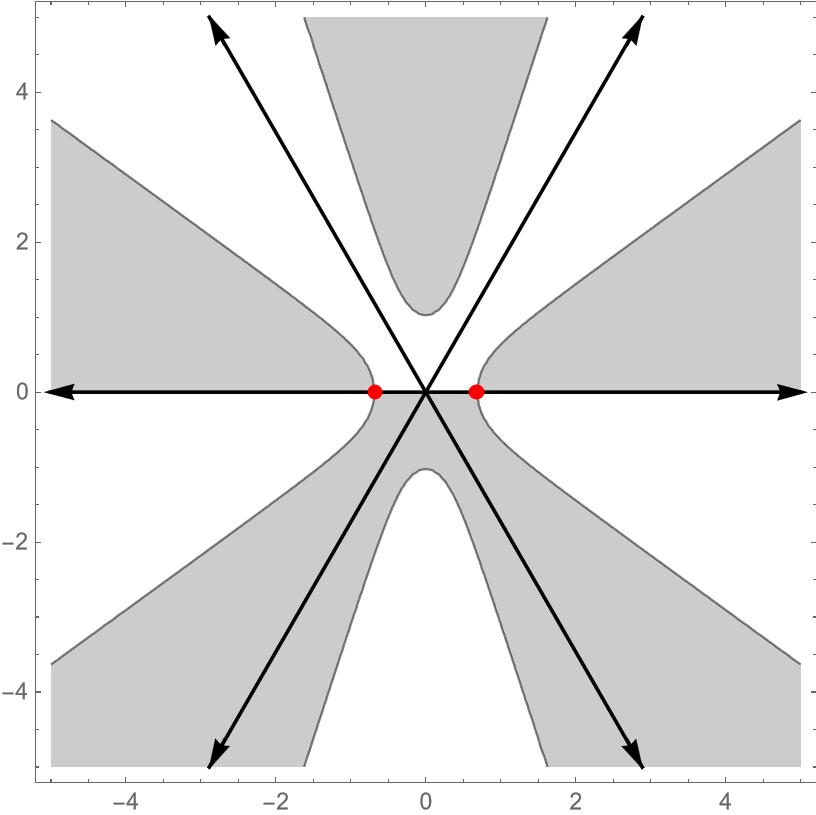}
			\put(70.6,58.9){\small $U_1:\Re \Phi_{21} > 0$}
			\put(10.5,58.9){\small $U_3:\Re \Phi_{21} > 0$}
			\put(70.6,40.3){\small $U_2:\Re \Phi_{21} < 0$}
			\put(10.5,41.3){\small $U_4:\Re \Phi_{21} < 0$}
			\put(57.5,48.3){\small $k_0$}
			\put(41.5,48.3){\small $-k_0$}
		\end{overpic}
		\caption{{\protect\small
				The open subsets $U_j$ for $j=1,2,3,4$ and the saddle points $\pm k_0$ (the red points).  The gray regions correspond to $\{k\in \C \mid \Re \Phi_{21} > 0\}$, while the white regions correspond to $\{k\in \C \mid \Re \Phi_{21} < 0\}$.}}
		\label{regionU painleve}
	\end{figure}
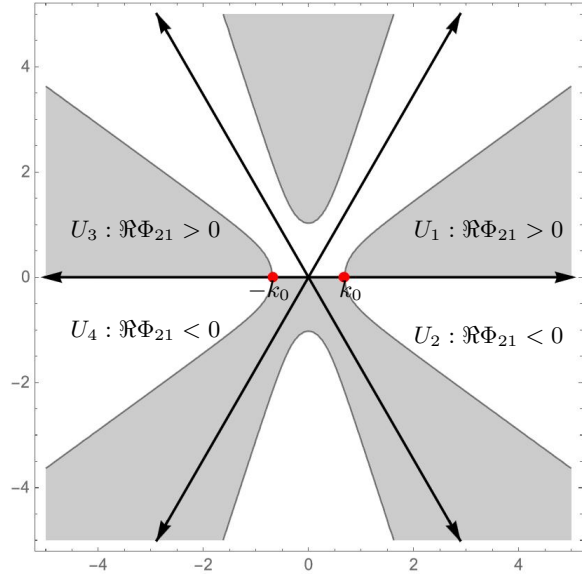
	\begin{lemma}\label{Lemma-analytic-extension painleve}
		For any integer $N\ge1$, letting $A>0$ be a constant, the functions $\tilde{r}_j(k)$ for $j=1,2$ have the following decompositions:
		$$
		\begin{array}{ll}
			\tilde{r}_1(x,t;k)=\tilde{r}_{1, a}(x, t; k)+\tilde{r}_{1, r}(x, t; k), & k \in\left(k_0, \infty\right), \\
			\tilde{r}_2(x,t;k)=\tilde{r}_{2, a}(x, t; k)+\tilde{r}_{2, r}(x, t; k), & k \in(-\infty,-k_0).
		\end{array}
		$$
		Furthermore, the decomposition functions $\tilde{r}_{j,a}(x,t;k)$ and $\tilde{r}_{j,r}(x,t;k)~(j=1,2)$ have the properties as follow:
		\begin{enumerate}
			\item For each $\xi\in[0,A]$ and $t\ge1$, $\tilde{r}_{1,a}(x,t;k)$ and $\tilde{r}_{2,a}(x,t;k)$ are well-defined and continuous for $\bar{U}_1$ and $\bar{U}_4$, respectively, and are analytic in the interior of their respective domains. The open subsets $U_j~(j=1,2,3,4)$ are depicted in Figure \ref{regionU painleve}.
			
			\item For each $\xi\in[0,A]$ and $t\ge1$, the functions $\tilde{r}_{j,a}(x,t;k)$ for $j=1,2$  satisfy the following estimates:
			$$
			\left|\tilde{r}_{j, a}(x, t; k)-\sum_{i=0}^N\frac{\tilde{r}_j^{(i)}(k_*)(k-k_*)^i}{i!}\right| \leq C|k-k_*|^{N+1}{\mathrm{e}^{t|\Re \Phi_{21}(\xi;k)|/4}},
			$$
			and
			$$
			\left|\tilde{r}_{j, a}(x, t; k)\right| \leq \frac{C}{1+|k|^{N+1}}{\mathrm{e}^{t|\Re \Phi_{21}(\xi;k)|/4}},
			$$
where the first inequality holds for \( j=1 \) when \( k_* = k_0 \) and \( k\in {U}_1 \), and for \( j=2 \) when \( k_* = -k_0 \) and \( k\in  {U}_4 \).
			\item For each $p$ satisfying \(1 \leq p \leq \infty\) and \(\xi\in[0,A]\), the \(L^p\)-norms of the functions \(\tilde{r}_{j,r}(x,t;k)\) \((j=1,2)\) are \(\mathcal{O}(t^{-N})\) on their respective domains as \(t \to \infty\).
		\end{enumerate}
	\end{lemma}
	\begin{proof}
		The proof parallels that in Lemma \ref{Lemma-analytic-extension}, thus it is omitted for brevity. For more information, refer to Ref. \cite{Lenells-mkdv}.
	\end{proof}
	Now, according to the decompositions of the functions $r_{j}(k)~(j=1,2)$ above, transform the RH problem for $m(x,t;k)$ of the mSK equation (\ref{rewrite-msk}) into the RH problem for $m^{(1)}(x,t;k)$ by
	$$
	m^{(1)}(x,t;k)=m(x,t;k)G(x,t;k),\ k\in\C\setminus\Sigma^{1},
	$$
	where the jump contour $\Sigma^{1}$ is depicted in Figure \ref{Sigma1 Painleve} and the transformation matrices are $G(x,t;k):=G_n(x,t;k)$ for $n=1,2,\cdots,6$. In particular, the matrix $G^{(1)}_1(x,t;k)$ is defined near $k_0$ by
	$$
	G_1(x, t; k):= \begin{cases}
		\small\left(\begin{array}{ccc}
			1 & \tilde{r}_{1,a}(k) \mathrm{e}^{-\theta_{21}(x,t;k)} & 0 \\
			0 & 1& 0 \\
			0 & 0 & 1
		\end{array}\right), & k\ \text{on the right}  \text{side of } \Sigma_{1}^{1},\\
		\small\left(\begin{array}{ccc}
			1 & 0 & 0 \\
			\tilde{r}_{1,a}^{*}(k) \mathrm{e}^{\theta_{21}(x,t;k)} & 1& 0 \\
			0 & 0 & 1
		\end{array}\right), & k\ \text{on the left} \text{side of } \Sigma_{3}^{1}, \\
		I, & \text{otherwise},
	\end{cases}
	$$
	and
	$G_4(x,t;k)$ is defined near $-k_0$ by
	$$
	G_4(x, t; k):= \begin{cases}
		\small\left(\begin{array}{ccc}
			1 & -\tilde{r}_{2,a}^{*}(k) \mathrm{e}^{-\theta_{21}(x,t;k)} & 0 \\
			0 & 1& 0 \\
			0 & 0 & 1
		\end{array}\right), & k\ \text{on the right}  \text{side of } \Sigma_{6}^{1}, \\
		\small\left(\begin{array}{ccc}
			1 & 0 & 0 \\
			-\tilde{r}_{2,a}(k) \mathrm{e}^{\theta_{21}(x,t;k)} & 1& 0 \\
			0 & 0 & 1
		\end{array}\right), & k\ \text{on the left}  \text{side of } \Sigma_{8}^{1}, \\
		I, & \text{otherwise}.
	\end{cases}
	$$
	 One can obtain the other functions $G_n(x,t;k)~(n=2,3,5,6)$ by the symmetry properties in (\ref{symmetry}).
	
	\begin{figure}[h]
		\centering
		\begin{tikzpicture}[>=latex]
			\draw[very thick] (-4,0) to (4,0) node[black,right]{$\mathbb{R}$};
			\draw[very thick] (-2,-1.732*2) to (2,1.732*2)  node[black,above]{$\omega^2\mathbb{R}$};
			\draw[very thick] (2,-1.732*2) to (-2,1.732*2)    node[black,above]{$\omega\mathbb{R}$};
			\filldraw[mred] (1.6,0) node[black,below=1mm]{$k_{0}$} circle (1.5pt);
			\filldraw[mred] (-1.6,0) node[black,below=1mm]{$-k_{0}$} circle (1.5pt);
			\filldraw[mblue] (.8,1.732*0.8) node[black,right=1mm]{$-\omega^{2}k_{0}$} circle (1.5pt);
			\filldraw[mblue] (-.8,-1.732*0.8) node[black,right=1mm]{$\omega^{2}k_{0}$} circle (1.5pt);
			\filldraw[mgreen] (.8,-1.732*0.8) node[black,right=1mm]{$-\omega k_{0}$} circle (1.5pt);
			\filldraw[mgreen] (-.8,1.732*0.8) node[black,right=1mm]{$\omega k_{0}$} circle (1.5pt);
			\draw[<->,very thick] (-1,0) node[below] {$5$} to (1,0) node[above] {$4$};
			\draw[<->,very thick] (-3,0) node[below] {$7$} to (3,0) node[above] {$2$};
			\draw[<->,very thick] (-.5,-1.732*.5) node[right] {} to (.5,1.732*.5)node[left] {};
			\draw[<->,very thick] (-1.5,-1.732*1.5) node[right] {} to (1.5,1.732*1.5) node[left] {};
			\draw[<->,very thick] (-.5,1.732*.5) node[left] {} to (.5,-1.732*.5) node[right] {};
			\draw[<->,very thick] (-1.5,1.732*1.5) node[left] {} to (1.5,-1.732*1.5) node[right] {};
			\draw[->,very thick,mred] (1.6,0)  to (3.2,-0.8) node[below,black] {$\small 1$};
			\draw[-,very thick,mred] (1.6,0)  to (4,-1.2);
			\draw[-,very thick,mred,rotate around x=180] (1.6,0)  to (4,-1.2);
			\draw[->,very thick,mred,rotate around x=180] (1.6,0)  to (3.2,-0.8) node[above,black] {$\small 3$};
			\draw[->,very thick,mred,rotate around y=180] (1.6,0)  to (3.2,-0.8) node[below,black] {$\small 8$};
			\draw[-,very thick,mred,rotate around y=180] (1.6,0)  to (4,-1.2);
			\draw[->,very thick,mred,rotate=180] (1.6,0)  to (3.2,-0.8) node[above,black] {$\small 6$};
			\draw[-,very thick,mred,rotate=180] (1.6,0)  to (4,-1.2);
			
			\draw[->,very thick,mblue,rotate=60] (1.6,0)  to (3.2,-0.8);
			\draw[-,very thick,mblue,rotate=60] (1.6,0)  to (4,-1.2);
			\draw[-,very thick,mblue,rotate around x=180,rotate=60] (1.6,0)  to (4,-1.2);
			\draw[->,very thick,mblue,rotate around x=180,rotate=60] (1.6,0)  to (3.2,-0.8);
			\draw[->,very thick,mblue,rotate around y=180,rotate=60] (1.6,0)  to (3.2,-0.8);
			\draw[-,very thick,mblue,rotate around y=180,rotate=60] (1.6,0)  to (4,-1.2);
			\draw[->,very thick,mblue,rotate=180,rotate=60] (1.6,0)  to (3.2,-0.8);
			\draw[-,very thick,mblue,rotate=180,rotate=60] (1.6,0)  to (4,-1.2);
			
			\draw[->,very thick,mgreen,rotate=120] (1.6,0)  to (3.2,-0.8);
			\draw[-,very thick,mgreen,rotate=120] (1.6,0)  to (4,-1.2);
			\draw[-,very thick,mgreen,rotate around x=180,rotate=120] (1.6,0)  to (4,-1.2);
			\draw[->,very thick,mgreen,rotate around x=180,rotate=120] (1.6,0)  to (3.2,-0.8);
			\draw[->,very thick,mgreen,rotate around y=180,rotate=120] (1.6,0)  to (3.2,-0.8);
			\draw[-,very thick,mgreen,rotate around y=180,rotate=120] (1.6,0)  to (4,-1.2);
			\draw[->,very thick,mgreen,rotate=180,rotate=120] (1.6,0)  to (3.2,-0.8);
			\draw[-,very thick,mgreen,rotate=180,rotate=120] (1.6,0)  to (4,-1.2);
		\end{tikzpicture}
		\caption{{\protect\small
				The jump contour $\Sigma^{1}$ and the saddle points $\pm\omega^jk_0~(j=0,1,2)$.}}
		\label{Sigma1 Painleve}
	\end{figure}
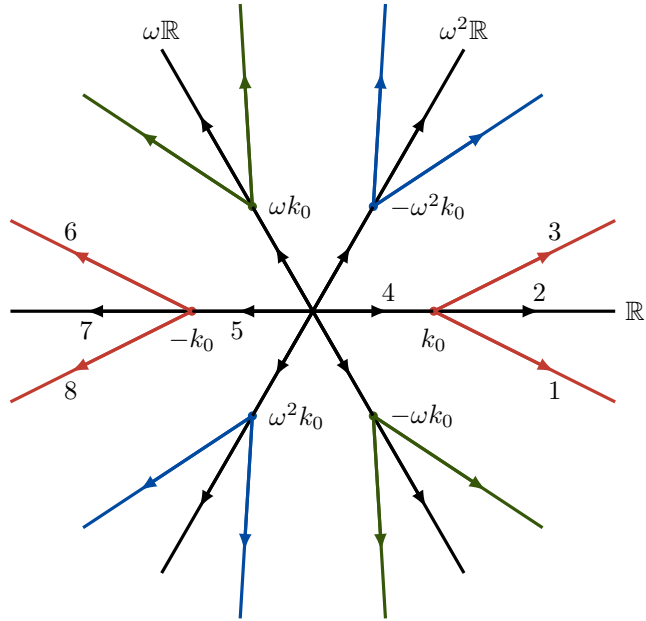
	For $\xi\in[0,A]$, the jump matrix $v^{(1)}$ converges uniformly to identity matrix $I$ as $t\to\infty$ except for the cuts near the saddle points, i.e., $\left\{\pm k_0,\pm \omega k_0,\pm \omega^2k_0\right\}$. Consequently, we only need to carry out the long-time asymptotic analysis near the saddle points. For $\xi\in[0,A]$, take the self-similar transformation
\begin{equation}\label{self-similar-transformation0}
    \lambda:={k}{(5t)^{1/5}},\quad y:=\frac{x}{(5t)^{\frac{1}{5}}},
\end{equation}
then the phase functions $\theta_{ij}(x,t;k)$ for $1\le j<i\le 3$ are transformed into
	$$
	\theta_{i j}(y;\lambda)=\left[\left(\omega^i-\omega^j\right)y\lambda-\frac{9\lambda^5}{5}\left(\omega^{2i}-\omega^{2j}\right) \right].
	$$
    \par
	As that in Section \ref{local parmetrices}, suppose $\Sigma_{\le}^{\epsilon}=\Sigma^{1}\cap B_{\epsilon}(0) \setminus \left(\cup_{j=0}^2(-\omega^j\infty,-\omega^jk_0)\cup(\omega^jk_0,\omega^j\infty)\right)$, where $B_{\epsilon}(0):=\{\lambda\in \C||\lambda|<\epsilon\}=\{k\in \C||k|<\epsilon(5t)^{1/5}\}$. Indeed, the local model problem with contour $\Sigma_{\le}^{\epsilon}$ is related to the Painlev\'e model problem for $m^{p}(y;\lambda)$ defined in Appendix \ref{appendix painleve}.

	\begin{lemma}\label{error for painleve}
		For $\xi\in[0,A]$, the function $m^{p}(x,t;k)=m^{p}(y;\lambda)$ is a bounded analytic function for $k\in B_{\epsilon}(0)\setminus\Sigma_{\le}^{\epsilon}$, such that for any $1\le p\le\infty$, one has
		\begin{equation}
			\|v^{(1)}-v^{p}\|_{L^p}\le t^{-\frac{1}{5}}.
		\end{equation}
		Furthermore, it can be obtained that
		$$
		m^{p}(x,t;k)^{-1}=I-\frac{m^{p}_{1}(y)}{kt^{\frac{1}{5}}}+\mathcal{O}\left(\frac{1}{t^{\frac{2}{5}}}\right).
		$$
	\end{lemma}
	\begin{proof}
		The analyticity of $m^p(y;\lambda)$ follows from its definition, while the boundedness is a consequence of the sign signature of the jump matrices. In particular, it can be derived that
		$$
		v^{(1)}-v^{p}=\begin{cases}
			\begin{aligned}
				&\small\left(\begin{array}{ccc}
					0 & 0 & 0 \\
					(\tilde{r}_{1,a}^{*}(k)-\tilde{r}_1^{*}(0)) \mathrm{e}^{\theta_{21}(x,t;k)} & 0& 0 \\
					0 & 0 & 0
				\end{array}\right), && k\in \Sigma_{1,\le}^{(\epsilon)},\\
				&\small\left(\begin{array}{ccc}
					0 & \tilde{r}_{1,r}(k) \mathrm{e}^{-\theta_{21}(x,t;k)} & 0 \\
					\tilde{r}_{1,r}^{*}(k) \mathrm{e}^{\theta_{21}(x,t;k)} & -\tilde{r}_{1,r}(k)\tilde{r}_{1,r}^*(k) & 0 \\
					0 & 0 & 0
				\end{array}\right), && k\in \Sigma_{2,\le}^{(\epsilon)},\\
				&\small\left(\begin{array}{ccc}
					0 & -(\tilde{r}_{1,a}(k)-\tilde{r}_1(0)) \mathrm{e}^{-\theta_{21}(x,t;k)} & 0 \\
					0 & 0& 0 \\
					0 & 0 & 0
				\end{array}\right), && k\in \Sigma_{3,\le}^{(\epsilon)},\\
				&\small \left(\begin{array}{ccc}
					0 & -(\tilde{r}_1(k)-\tilde{r}_1(0)) \mathrm{e}^{-\theta_{21}(x,t;k)} & 0 \\
					(\tilde{r}_1^*(k)-\tilde{r}_1^*(0)) \mathrm{e}^{\theta_{21}(x,t;k)} & \left|\tilde{r}_1(k)\right|^2-\left|\tilde{r}_1(0)\right|^2 & 0 \\
					0 & 0 & 0
				\end{array}\right), && k\in \Sigma_{4,\le}^{(\epsilon)}.\\
			\end{aligned}
		\end{cases}
		$$
		More precisely, we have
		$$
		\begin{aligned}
			\Re(-\theta_{21}(x,t;k))&=9\sqrt{3}t[(k-k_0)^5+5k_0(k-k_0)^4+10k_0^2(k-k_0)^3+10k_0^3(k-k_0)^2-4k_0^5]\\
			&\le -C t |k|^5,\quad k\in\Sigma^{\epsilon}_{3,\le},
		\end{aligned}
		$$
		where $C$ is a positive constant.
		On the other hand, it is seen that
		$$
		\tilde{r}_1(k_0)-\tilde{r}_1(0)=\tilde{r}^{(1)}(0)k_0+\mathcal{O}\left(k_0^2\right).
		$$
		Recalling the inequalities in Lemma \ref{Lemma-analytic-extension painleve}, for $\xi\in[0,A]$ and $k\in\Sigma^{(\epsilon)}_{3,\le}$, it follows that
		$$
		\begin{aligned}
			|v^{(1)}-v^{p}|&\le  |\tilde{r}_{1,a}(k)-\tilde{r}_{1}(k_0)|\mathrm{e}^{-t\Re \Phi_{21}}+|\tilde{r}_{1}(0)-\tilde{r}_{1}(k_0)|\mathrm{e}^{-t\Re \Phi_{21}}\\
			&\le C |k-k_0|\mathrm{e}^{-ct|k|^5}+C |k|\mathrm{e}^{-ct|k|^5}\le C|\lambda t^{-\frac{1}{5}}|\mathrm{e}^{-c|\lambda|^5}.
		\end{aligned}
		$$
		Consequently, for each $1\le p\le \infty$, it is immediate that
		$$
		\|v^{(1)}-v^{p}\|_{L^{\infty}(\Sigma^{(\epsilon)}_{3,\le})}\le C t^{-\frac{1}{5}},\
		\|v^{(1)}-v^{p}\|_{L^{1}(\Sigma^{(\epsilon)}_{3,\le})}\le C t^{-\frac{2}{5}}.
		$$
		Furthermore, by the Lemma \ref{Lemma-analytic-extension painleve}, for $k\in \Sigma_{2,\le}^{(\epsilon)}$,  we have
		$$
		|v^{(1)}-v^{p}|\le C t^{-1}.
		$$
	The estimates of the other jump matrices can also be gotten in a similar way.		Noting the expansion $m^{p}(y;\lambda)=I+\frac{m_{1}^p(y)}{\lambda}+\mathcal{O}\left(\frac{1}{\lambda^2}\right)$ in Appendix \ref{appendix painleve} and recalling  the self-similar variable $\lambda=kt^{\frac{1}{5}}$, it follows that
		$$
		m^{p}(x,t;k)^{-1}=I-\frac{m^{p}_{1}(y)}{kt^{\frac{1}{5}}}+\mathcal{O}\left(\frac{1}{t^{\frac{2}{5}}}\right).
		$$
	\end{proof}
	\begin{figure}[h]
		\centering
		\begin{tikzpicture}[>=latex]
			\draw[very thick] (-4,0) to (4,0) node[black,right]{$\mathbb{R}$};
			\draw[very thick] (-2,-1.732*2) to (2,1.732*2)  node[black,above]{$\omega^2\mathbb{R}$};
			\draw[very thick] (2,-1.732*2) to (-2,1.732*2)    node[black,above]{$\omega\mathbb{R}$};
			\filldraw[mred] (1.6,0) node[black,below=1mm]{$k_{0}$} circle (1.5pt);
			\filldraw[mred] (-1.6,0) node[black,below=1mm]{$-k_{0}$} circle (1.5pt);
			\filldraw[mblue] (.8,1.732*0.8) node[black,right=1mm]{$-\omega^{2}k_{0}$} circle (1.5pt);
			\filldraw[mblue] (-.8,-1.732*0.8) node[black,right=1mm]{$\omega^{2}k_{0}$} circle (1.5pt);
			\filldraw[mgreen] (.8,-1.732*0.8) node[black,right=1mm]{$-\omega k_{0}$} circle (1.5pt);
			\filldraw[mgreen] (-.8,1.732*0.8) node[black,right=1mm]{$\omega k_{0}$} circle (1.5pt);
			\draw[<->,very thick] (-1,0) node[below] {$5$} to (1,0) node[above] {$4$};
			\draw[<->,very thick] (-3,0) node[below] {$7$} to (3,0) node[above] {$2$};
			\draw[<->,very thick] (-.5,-1.732*.5) node[right] {} to (.5,1.732*.5)node[left] {};
			\draw[<->,very thick] (-1.5,-1.732*1.5) node[right] {} to (1.5,1.732*1.5) node[left] {};
			\draw[<->,very thick] (-.5,1.732*.5) node[left] {} to (.5,-1.732*.5) node[right] {};
			\draw[<->,very thick] (-1.5,1.732*1.5) node[left] {} to (1.5,-1.732*1.5) node[right] {};
			\draw[->,very thick,mred] (1.6,0)  to (3.2,-0.8) node[below,black] {$\small 1$};
			\draw[-,very thick,mred] (1.6,0)  to (4,-1.2);
			\draw[-,very thick,mred,rotate around x=180] (1.6,0)  to (4,-1.2);
			\draw[->,very thick,mred,rotate around x=180] (1.6,0)  to (3.2,-0.8) node[above,black] {$\small 3$};
			\draw[->,very thick,mred,rotate around y=180] (1.6,0)  to (3.2,-0.8) node[below,black] {$\small 8$};
			\draw[-,very thick,mred,rotate around y=180] (1.6,0)  to (4,-1.2);
			\draw[->,very thick,mred,rotate=180] (1.6,0)  to (3.2,-0.8) node[above,black] {$\small 6$};
			\draw[-,very thick,mred,rotate=180] (1.6,0)  to (4,-1.2);
			
			\draw[->,very thick,mblue,rotate=60] (1.6,0)  to (3.2,-0.8);
			\draw[-,very thick,mblue,rotate=60] (1.6,0)  to (4,-1.2);
			\draw[-,very thick,mblue,rotate around x=180,rotate=60] (1.6,0)  to (4,-1.2);
			\draw[->,very thick,mblue,rotate around x=180,rotate=60] (1.6,0)  to (3.2,-0.8);
			\draw[->,very thick,mblue,rotate around y=180,rotate=60] (1.6,0)  to (3.2,-0.8);
			\draw[-,very thick,mblue,rotate around y=180,rotate=60] (1.6,0)  to (4,-1.2);
			\draw[->,very thick,mblue,rotate=180,rotate=60] (1.6,0)  to (3.2,-0.8);
			\draw[-,very thick,mblue,rotate=180,rotate=60] (1.6,0)  to (4,-1.2);
			
			\draw[->,very thick,mgreen,rotate=120] (1.6,0)  to (3.2,-0.8);
			\draw[-,very thick,mgreen,rotate=120] (1.6,0)  to (4,-1.2);
			\draw[-,very thick,mgreen,rotate around x=180,rotate=120] (1.6,0)  to (4,-1.2);
			\draw[->,very thick,mgreen,rotate around x=180,rotate=120] (1.6,0)  to (3.2,-0.8);
			\draw[->,very thick,mgreen,rotate around y=180,rotate=120] (1.6,0)  to (3.2,-0.8);
			\draw[-,very thick,mgreen,rotate around y=180,rotate=120] (1.6,0)  to (4,-1.2);
			\draw[->,very thick,mgreen,rotate=180,rotate=120] (1.6,0)  to (3.2,-0.8);
			\draw[-,very thick,mgreen,rotate=180,rotate=120] (1.6,0)  to (4,-1.2);
			\draw[very thick,black,rotate=-120] (0,0) circle [radius=2.5cm];
			\draw[->,very thick] (0.2,2.5)  to (-0.1,2.5) node[above,black] {$\partial B_{\epsilon}(0)$};
		\end{tikzpicture}
		\caption{{\protect\small
				The jump contour $\hat\Sigma$ and the saddle points $\pm\omega^jk_0$ for $j=0,1,2$.}}
		\label{Sigma hat}
	\end{figure}
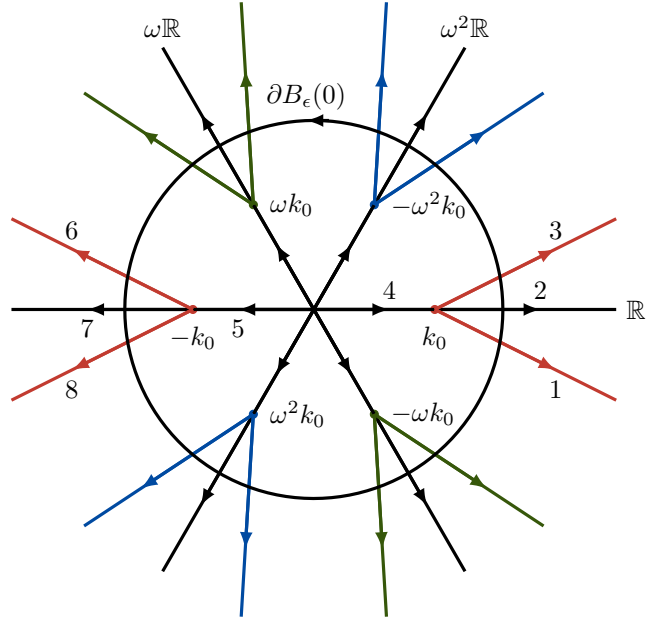
    \par
	Finally, denote the new contour $\hat \Sigma:=\Sigma^{1}\cup\partial B_{\epsilon}(0) $, which can be seen in Figure \ref{Sigma hat}, and define the function $\hat{m}(x,t;k)$ by
	\begin{equation}\label{hat m}
		\hat{m}(x,t;k)=\begin{cases}
			\begin{aligned}
				&m^{(1)}(x,t;k)(m^p)^{-1}(x,t;k),&&k\in B_{\epsilon}(0),\\
				&m^{(1)}(x,t;k),&&k\in \C\setminus B_{\epsilon}(0),\\
			\end{aligned}
		\end{cases}
	\end{equation}
	which is the solution of a RH problem with jump contour $\hat \Sigma$ and jump matrices 
	\begin{equation}\label{hat m jumps}
		\hat{v}(x,t;k)=\begin{cases}
			\begin{aligned}
				&m^{p}_{-}(x,t;k)v^{(1)}(m^{p}_{+})^{-1}(x,t;k),&&k\in B_{\epsilon}(0)\cap \hat \Sigma,\\
				&(m^p(x,t;k))^{-1},&&k\in \partial B_{\epsilon}(0),\\
				&v^{(1)},&&k\in\hat \Sigma\setminus \bar{B}_{\epsilon}(0).
			\end{aligned}
		\end{cases}
	\end{equation}
	\begin{lemma}
		Let $\hat w =\hat v -I$ and $p\ge 1$, then for $\xi\in[0,A]$, the following inequalities hold uniformly
		\begin{equation}
			\begin{aligned}
				&\|\hat w\|_{L^p(\partial B_{\epsilon}(0))}\le Ct^{-\frac{1}{5}},\ \|\hat w\|_{L^p(\hat \Sigma \setminus B_{\epsilon}(0))}\le Ct^{-N},\
				\|\hat w\|_{L^p( B_{\epsilon}(0)\cap \Sigma_{1})}\le Ct^{-\frac{1}{5}-\frac{1}{5p}}.
			\end{aligned}
		\end{equation}
	\end{lemma}
	\begin{proof}
		The first and last inequalities follow from Lemma \ref{error for painleve}, while the second inequality is the consequence of Lemma \ref{Lemma-analytic-extension painleve}.
	\end{proof}
	Since $\|\hat w\|_{L^{\infty}}\to0$ as $t\to\infty$, it follows that the operator $I-C_{\hat w}$ is invertible for $t$ large enough. Therefore, the solution $\hat m(x,t;k)$ of the RH problem  exists and is unique, that is
	\begin{equation}\label{m recover Cauchy}
		\hat m(x,t;k)=I+C(\hat\mu \hat w)=I+\int_{\hat\Sigma} \frac{\hat\mu(x,t;\zeta)\hat w(x,t;\zeta)}{\zeta-k} \frac{\mathrm{d} \zeta}{2 \pi i}, \quad k \in \C\setminus \hat \Sigma,
	\end{equation}
	with $\hat \mu=I+(I-C_{\hat w})^{-1}C_{\hat w}I$.
    
	\begin{lemma}\label{H-asymptotic-54}
		When $\xi\in[0,A]$ and $t\to\infty$, we have
		$$
		\lim_{k\to\infty}k(\hat{m}(x,t;k)-I)=-\frac{1}{2\pi i}\int_{\partial B_{\epsilon}(0)}\hat w(x,t;k)\mathrm{d}k+\mathcal{O}\left(t^{-\frac{2}{5}}\right).
		$$
	\end{lemma}

	\begin{proof}
		It follows from the equation (\ref{m recover Cauchy}) that
		$$
		\lim_{k\to\infty}k(\hat{m}(x,t;k)-I)=-\frac{1}{2\pi i}\int_{\hat \Sigma}\hat{\mu}(x,t;\zeta)\hat{w}(x,t;\zeta)\mathrm{d}\zeta.
		$$
		Decomposing the right integration into three parts, yields
		$$
		-\frac{1}{2\pi i}\int_{\partial B_\epsilon{(0)}}\hat{w}(x,t;k)\mathrm{d}k+Q_1(x,t)+Q_2(x,t),
		$$
		where
		$$
		Q_1(x,t):=-\frac{1}{2\pi i}\int_{\hat\Sigma}(\hat{\mu} (x,t;k)-I)\hat{w}(x,t;k)\mathrm{d}k,\
		Q_2(x,t):=-\frac{1}{2\pi i}\int_{\hat\Sigma\setminus{\partial B_\epsilon{(0)}}}\hat{w}(x,t;k)\mathrm{d}k.
		$$
		For the function $Q_1(x,t)$, the H\"{o}lder inequality indicates that
		$$
		|Q_1(x,t)|\le C \|\hat{\mu} (x,t;\cdot)-I\|_{L^p(\hat\Sigma)}\|\hat{w}(x,t;\cdot)\|_{L^q(\hat\Sigma)}\le\frac{1}{t^{\frac{2}{5}}},
		$$
		where $\frac{1}{p}+\frac{1}{q}=1$. For the function $Q_2(x,t)$, it is seen that
		$$
		|Q_2(x,t)|\le C\|\hat{w}(x,t;\cdot)\|_{L^1(\hat \Sigma \setminus\partial B_\epsilon{(0)})}\le t^{-\frac{2}{5}}.
		$$
	\end{proof}
    \par
	In summary, the lemmas above shows that, for $\xi\in[0,A]$ and $t\to\infty$, the long-time asymptotic behavior of the solution to the mSK equation (\ref{msk-equation}) is
	$$
	\begin{aligned}
		w(x,t)&=3\lim_{k\to\infty}k(m(x,t;k)-I)_{13}=-\frac{3}{2\pi i}\int_{\partial B_{\epsilon}(0)}\hat w(x,t;k)\mathrm{d}k+\mathcal{O}\left(t^{-\frac{2}{5}}\right)\\
		&=-\frac{3}{2\pi i}\int_{\partial B_{\epsilon}(0)}((m^p)^{-1}-I)\mathrm{d}k+\mathcal{O}\left(t^{-\frac{2}{5}}\right)=\frac{3(m^p_{1}(y))_{13}}{t^{\frac{1}{5}}}+\mathcal{O}\left(\frac{1}{t^{\frac{2}{5}}}\right),
	\end{aligned}
	$$
where $3(m^p_{1}(y))_{13}$ satisfies the fourth-order analogues of the Painlev\'{e} transcendent in (\ref{4th-Painleve}), see also Appendix \ref{appendix painleve}. Thus this completes the proof of Sector $\rm III$ for $x\ge0$ in Theorem \ref{msk Thm}.

\subsection{Painlev\'{e} region for $x>0$}
	Recall that the saddle point $k_0$ satisfies  $k_0^4=\frac{-x}{45t}$, which indicates that
   it no longer lies on the real line. However, the $k_0$ still behaviors like $t^{-\frac{1}{5}}$ as $t\to\infty$, thus the local self-similar parameters $\lambda$ and $y$ remain unchanged. Similar to the case of $x<0$, decompose the reflection coefficients $\tilde{r}_j(k)~(j=1,2)$ into two parts as shown in the lemma below.
	\begin{lemma} \label{decomposation-x-positive}
		For any integer $N\ge1$, let $A$ be a positive constant, then the functions $\tilde{r}_j(k)~(j=1,2)$ have the following decompositions:
		$$
		\begin{array}{ll}
			\tilde{r}_1(k)=\tilde{r}_{1, a}(x, t; k)+\tilde{r}_{1, r}(x, t; k), & k \in\left(0, \infty\right), \\
			\tilde{r}_2(k)=\tilde{r}_{2, a}(x, t; k)+\tilde{r}_{2, r}(x, t; k), & k \in(-\infty,0),
		\end{array}
		$$
		where the decomposition functions $\tilde{r}_{j,a}(x,t;k)$ and $\tilde{r}_{j,r}(x,t;k)~(j=1,2)$ have the properties as follow:
		\begin{enumerate}
			\item For each $\xi\in[-A,0]$ and $t\ge1$, $\tilde{r}_{1,a}(x,t;k)$ and $\tilde{r}_{2,a}(x,t;k)$ are well-defined and continuous on the regions $\bar{U}_1$ and $\bar{U}_4$, respectively, and are analytic in the interior of their respective domains. The open subsets $U_j$ for $j=1,2,3,4$ are depicted in Figure \ref{regionU painlevege0}.
			
			\item For $\xi\in[-A,0]$ and $t\ge1$, the functions $\tilde{r}_{j,a}(x,t;k)~(j=1,2)$ satisfy the following estimates:
			$$
			\left|\tilde{r}_{j, a}(x, t; k)-\sum_{i=0}^N\frac{\tilde{r}_j^{(i)}(0)k^i}{i!}\right| \leq C|k|^{N+1}{\mathrm{e}^{t|\Re \Phi_{21}(\xi;k)|/4}},
			$$
			and
			$$
			\left|\tilde{r}_{j, a}(x, t; k)\right| \leq \frac{C}{1+|k|^{N+1}}{\mathrm{e}^{t|\Re \Phi_{21}(\xi;k)|/4}}.
			$$
			\item For each \(1 \leq p \leq \infty\) and \(\xi\in[-A,0]\), the \(L^p\)-norms of \(\tilde{r}_{j,r}(x,t;k)\) \((j=1,2)\) on their respective domains are \(\mathcal{O}(t^{-N})\) as \(t \to \infty\).
		\end{enumerate}
	\end{lemma}
	\begin{proof}
		The proof of this lemma follows the techniques outlined in \cite{Charlier-Lenells-2021}. As these techniques are quite standard, we omit the details for the sake of brevity.
	\end{proof}
	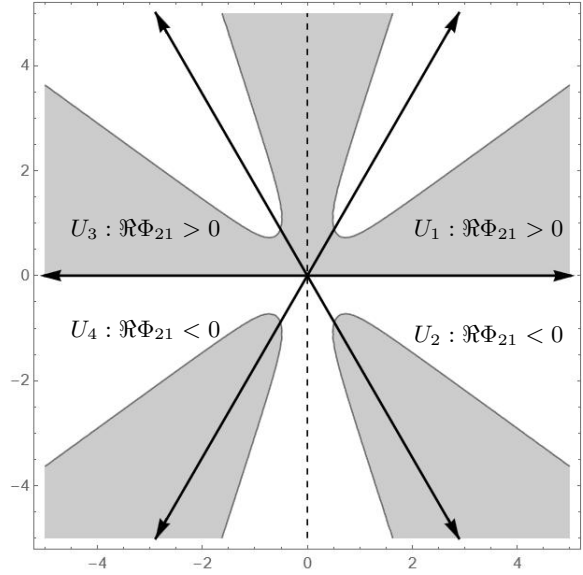
\begin{figure}[!h]
		\centering
		\begin{overpic}[width=.5\textwidth]{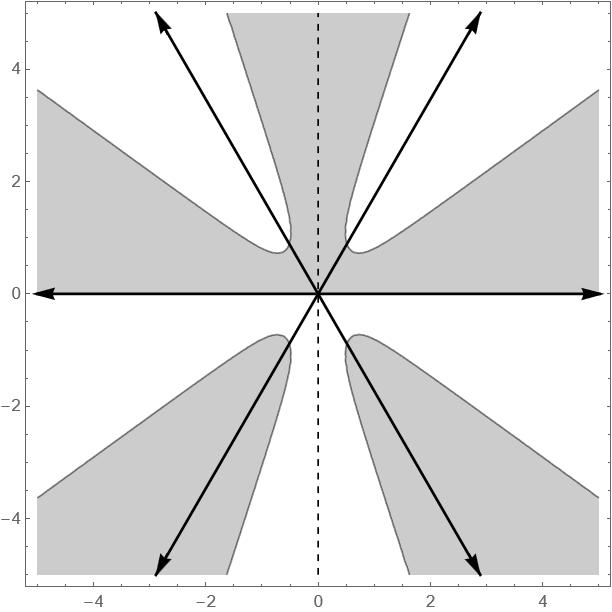}
			\put(70.6,58.9){\small $U_1:\Re \Phi_{21} > 0 $}
			\put(10.5,58.9){\small $U_3:\Re \Phi_{21} > 0$}
			\put(70.6,40.3){\small $U_2:\Re \Phi_{21} < 0$}
			\put(10.5,41.3){\small $U_4:\Re \Phi_{21} < 0$}
		\end{overpic}
		\caption{{\protect\small
				The open subsets $U_j$ for $j=1,2,3,4$, in which the gray regions correspond to $\{k\in \C \mid \Re \Phi_{21} > 0\}$, while the white regions correspond to $\{k\in \C \mid \Re \Phi_{21} < 0\}$}.}
		\label{regionU painlevege0}
	\end{figure}
	Now, employing the aforementioned decompositions of \( \tilde{r}_{j}(k)~(j=1,2) \), one can transform the RH problem for function \( m(x,t;k) \) into the RH problem for function \( m^{(1)}(x,t;k) \) by
	$m^{(1)}(x,t;k)=m(x,t;k)G(x,t;k)$ for $k\in\C\setminus\Sigma^{2}$, where $G(x,t;k):=G_n(x,t;k),$ $n=1,2,\cdots,6$. More precisely, $G_n(x,t;k)$ is similar to the case of $x<0$ above, the jump contour is $\Sigma^{2}$, see Figure \ref{Sigma1 painleve ge0}, and the jump matrix is defined as $\tilde{v}^{(1)}$.
	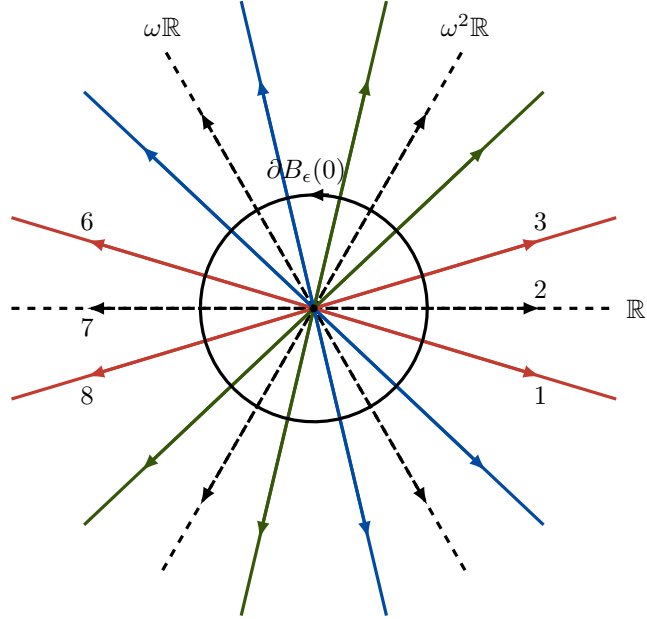
\begin{figure}[h]
	\centering
	\begin{tikzpicture}[>=latex]
		\draw[very thick,dashed] (-4,0) to (4,0) node[black,right]{$\mathbb{R}$};
		\draw[very thick,dashed] (-2,-1.732*2) to (2,1.732*2)  node[black,above]{$\omega^2\mathbb{R}$};
		\draw[very thick,dashed] (2,-1.732*2) to (-2,1.732*2)    node[black,above]{$\omega\mathbb{R}$};
		\draw[<->,very thick,mred] (-3,-0.9) node[below,black] {$8$} to (3,0.9) node[above,black] {$3$};
		\draw[<->,very thick,mred] (-3,0.9) node[above,black] {$6$} to (3,-0.9) node[below,black] {$1$};
		\draw[-,very thick,mred] (-4,1.2) to (4,-1.2);
		\draw[-,very thick,mred] (-4,-1.2) to (4,1.2);
		\draw[-,very thick,mgreen,rotate=60] (-4,1.2) to (4,-1.2);
		\draw[-,very thick,mgreen,rotate=60] (-4,-1.2) to (4,1.2);
		\draw[-,very thick,mblue,rotate=120] (-4,1.2) to (4,-1.2);
		\draw[-,very thick,mblue,rotate=120] (-4,-1.2) to (4,1.2);
		\draw[<->,very thick,mgreen,rotate=60] (-3,-0.9) to (3,0.9);
		\draw[<->,very thick,mgreen,rotate=60] (-3,0.9) to (3,-0.9);
		\draw[<->,very thick,mblue,rotate=120] (-3,-0.9) to (3,0.9);
		\draw[<->,very thick,mblue,rotate=120] (-3,0.9) to (3,-0.9);
		\draw[<->,very thick,dashed] (-3,0) node[below] {$7$} to (3,0) node[above] {$2$};
		\draw[<->,very thick,dashed] (-1.5,-1.732*1.5) node[right] {} to (1.5,1.732*1.5) node[left] {};
		\draw[<->,very thick,dashed] (-1.5,1.732*1.5) node[left] {} to (1.5,-1.732*1.5) node[right] {};
		
		\draw[very thick,black,rotate=-120] (0,0) circle [radius=1.5cm];
		\draw[->,very thick,dashed] (0.2,1.5)  to (-0.1,1.5) node[above,black] {$\partial B_{\epsilon}(0)$};
	\end{tikzpicture}
	\caption{{\protect\small
			The jump contour $\Sigma^{2}\cup\partial B_{\epsilon}(0)$.}}
	\label{Sigma1 painleve ge0}
\end{figure}

    For $\xi\in[-A,0]$, take the self-similar transformation (\ref{self-similar-transformation0}). The decomposition formulas in Lemma \ref{decomposation-x-positive} show that the jump matrices on $\Sigma^{2}$ tend to identity matrix $I$ as $t\to\infty$ except for the jumps near $k=0$, thus introduce the local model problem $M^{P}(y;\lambda)$ defined in Appendix \ref{appendix painleve}. 
	As in the case of $x<0$, suppose that $\Sigma_{\ge}^{\epsilon}=\Sigma^{2}\cap B_{\epsilon}(0) \setminus \Sigma$. Indeed, the local model problem on contour $\Sigma_{\ge}^{\epsilon}$ is related with the function $M^{P}(x,t;k)=M^{P}(y;\lambda)$ with jump contour in Figure \ref{MP}. The proof of the following lemmas are similar to the case of \( x < 0 \), so we only outline the context of these lemmas.
	
    \begin{lemma}\label{error for painleve xge0}
		For $\xi\in[-A,0]$, the function $M^{P}(x,t;k)$ is bounded and analytic for $k\in B_{\epsilon}(0)\setminus\Sigma_{\ge}^{\epsilon}$, such that for any $1\le p\le\infty$, one has
		\begin{equation}
			\|\tilde{v}^{(1)}-V^{P}\|\le t^{-\frac{1}{5}}.
		\end{equation}
		Furthermore, it is derived that
		$$
		M^{P}(x,t;k)^{-1}=I-\frac{M^{P}_{1}(y)}{kt^{\frac{1}{5}}}+\mathcal{O}\left(\frac{1}{t^{\frac{2}{5}}}\right).
		$$
	\end{lemma}
	Finally, let $ \check\Sigma:=\Sigma^{2}\cup\partial B_{\epsilon}(0) $ and define the function $\hat{M}(x,t;k)$ by
	\begin{equation}\label{hat m2}
		\hat{M}(x,t;k)=\begin{cases}
			\begin{aligned}
				&m^{(1)}(x,t;k)(M^P)^{-1}(x,t;k),&&k\in B_{\epsilon}(0),\\
				&m^{(1)}(x,t;k),&&k\in \C\setminus B_{\epsilon}(0),\\
			\end{aligned}
		\end{cases}
	\end{equation}
	with jump matrices
	\begin{equation}\label{hat m2 jumps}
		\hat{V}(x,t;k)=\begin{cases}
			\begin{aligned}
				&M^{P}_{-}(x,t;k)\tilde{v}^{(1)}(M^{P}_{+})^{-1}(x,t;k),&&k\in B_{\epsilon}(0)\cap \check\Sigma,\\
				&(M^P)^{-1}(x,t;k),&&k\in \partial B_{\epsilon}(0),\\
				&\tilde{v}^{(1)},&&k\in \check\Sigma \setminus \bar{B}_{\epsilon}(0).
			\end{aligned}
		\end{cases}
	\end{equation}
    
	\begin{lemma}
		Let $\hat W =\hat V -I$, then for $\xi\in[-A,0]$, the following inequalities hold uniformly
		\begin{equation}
			\begin{aligned}
				&\|\hat W\|_{L^p(\partial B_{\epsilon}(0))}\le Ct^{-\frac{1}{5}},\ \|\hat W\|_{L^p(\check\Sigma \setminus B_{\epsilon}(0))}\le Ct^{-N},\
				\|\hat W\|_{L^p( B_{\epsilon}(0)\cap \check\Sigma)}\le Ct^{-\frac{1}{5}-\frac{1}{5p}}.
			\end{aligned}
		\end{equation}
	\end{lemma}
Since $\|\hat W\|_{L^{\infty}}\to0$ as $t\to\infty$, it follows that the operator $I-C_{\hat W}$ is invertible for $t$ large enough. Therefore, the solution $\hat M(x,t;k)$ of the RH problem exists and is unique.
	
    \begin{lemma}\label{Painleve-58-Lemma}
		When $\xi\in[-A,0]$ and $t\to\infty$, it can also be obtained that
		$$
		\lim_{k\to\infty}k(m(x,t;k)-I)=-\frac{1}{2\pi i}\int_{\partial B_{\epsilon}(0)}\hat W(x,t;k)\mathrm{d}k+\mathcal{O}\left(t^{-\frac{2}{5}}\right).
		$$
	\end{lemma}
	In summary, for $\xi\in[0,A]$ and $t\to\infty$, the lemmas above result in
	$$
	\begin{aligned}
		w(x,t)&=3\lim_{k\to\infty}k(m(x,t;k)-I)_{13}=-\frac{3}{2\pi i}\int_{\partial B_{\epsilon}(0)}\hat W(x,t;k)\mathrm{d}k+\mathcal{O}\left(t^{-\frac{2}{5}}\right)\\
		&=-\frac{3}{2\pi i}\int_{\partial B_{\epsilon}(0)}((M^P)^{-1}-I)\mathrm{d}k+\mathcal{O}\left(t^{-\frac{2}{5}}\right)=\frac{3(M^P_{1}(y))_{13}}{t^{\frac{1}{5}}}+\mathcal{O}\left(\frac{1}{t^{\frac{2}{5}}}\right),
	\end{aligned}
	$$
	where $3(M^P_{1}(y))_{13}$ solves the fourth-order analogues of the Painlev\'{e} transcendent in (\ref{4th-Painleve}), see Appendix \ref{appendix painleve}. Thus this completes the proof of Sector $\rm III$ with $x<0$ in the Theorem \ref{msk Thm} for the mSK equation (\ref{msk-equation}).
	
	\begin{remark}
    The long-time asymptotics of the SK equation (\ref{SK}) in Sector ${\rm IV}$ can be formulated by 
		means of the Miura transformation \( u(x,t) = \frac{1}{6}(w_x(x,t) - w(x,t)^2) \). Although the Miura transformation is typically non-invertible, it still manages to reveal the asymptotic expression of the SK equation (\ref{SK}). This proves the asymptotic behavior of Sector ${\rm IV}$ in Theorem \ref{main-theorem-SK}.
	\end{remark}
	
	\begin{remark}	
		In fact, there are two transitional Painlev\'{e} regions in the long-time asymptotics of the mSK equation (\ref{msk-comparisons}): one between Sector $\rm II$ and Sector $\rm III$, and the other between Sector $\rm III$ and Sector $\rm VI$, which are also described by the fourth-order analogues of the Painlev\'{e} transcendent in (\ref{4th-Painleve}). This scenario is similar to the case that in the mKdV equation \cite{Deift-Zhou-1993}, in which the transition regions are expressed by the Painlev\'{e} $\rm II$ equation
        \( p_{II}''(s) - s p_{II}(s) - 2 p_{II}^3(s) = 0 \) that has a global real-valued solution. This solution aligns with the dispersive wave region as \( s \to \infty \) and behaves like the Airy function as \( s \to -\infty \).
	\end{remark}
	
	\section{\bf The Rapid Decay Region}\label{rapid region}
	When $x<0$ and $|x/t|\ge \frac{1}{M}$ for certain $M>1$, the potential function $u(x,t)$ rapidly decays as $t\to\infty$. In this case, the saddle points of the phase function $\theta_{21}$ also no longer lie on the real line. For practical purposes, take the transformation $t\to-t$ for the SK equation (\ref{SK}), and this sector corresponds to $|x/t|\ge \frac{1}{M}$ for $x>0$. Consequently, decompose $r_{j}(k)$ for $j=1,2$ into two parts as shown in the following lemma.
	\begin{lemma}\label{Rapid-Lemma}
		For any integer $N\ge1$, the functions $r_j(k)~(j=1,2)$ have the following decompositions:
		$$
		\begin{array}{ll}
			r_1(k)=r_{1, a}(x, t; k)+r_{1, r}(x, t; k), & k \in\left(0, \infty\right),\\
			r_2(k)=r_{2, a}(x, t; k)+r_{2, r}(x, t; k), & k \in\left( -\infty,0\right).
		\end{array}
		$$
		Furthermore, the decomposition functions have the properties of the forms:
		\begin{enumerate}
			\item For each $|\xi|:=|\frac{x}{t}|\ge \frac{1}{M}$ for $x>0$ and $t\ge1$, the functions $r_{1,a}(x,t;k)$ and $r_{2,a}(x,t;k)$ are well-defined and continuous on regions $\bar{U}_1$ and $\bar{U}_4$, respectively, and are analytic in the interior of their respective domains. The open subsets $U_j~(j=1,2,3,4)$ are similar to that in Figure \ref{regionU painlevege0}.
			
			\item For each $|\xi|:=|\frac{x}{t}|\ge \frac{1}{M}$ for $x>0$ and $t\ge1$, the functions $r_{j,a}(x,t;k)$ for $j=1,2$  satisfy the following estimates:
			$$
			\left|r_{j, a}(x, t; k)-\sum_{i=0}^N\frac{r_j^{(i)}(0)k^i}{i!}\right| \leq C|k|^{N+1}{\mathrm{e}^{t|\Re \Phi_{21}(\xi;k)|/4}}\  ,
			$$
			and
			$$
			\left|r_{j, a}(x, t; k)\right| \leq \frac{C}{1+|k|^{N+1}}{\mathrm{e}^{t|\Re \Phi_{21}(\xi;k)|/4}}.
			$$
            Especially, for $|x|\gg t$, we have
            $$
            \left|r_{j, a}(x, t; k)-\sum_{i=0}^N\frac{r_j^{(i)}(0)k^i}{i!}\right| \leq \mathrm{C}_N(\zeta)|k|^{N+1}{\mathrm{e}^{x|\Re {\Phi}_{21}(\zeta; k)|/4}}
            $$
            and
            $$
            \left|r_{j, a}(x, t; k)\right| \leq
            \frac{C}{1+|k|^{N+1}}{\mathrm{e}^{x|\Re \Phi_{21}(\zeta;k)|/4}}
            $$
            
			\item For each \(1 \leq p \leq \infty\) and  \(|\xi|\ge \frac{1}{M}\) with $x>0$, the \(L^p\)-norms of \(r_{j,r}(x, t; k)\) 
            \((j=1,2)\) on their respective domains are \(\mathcal{O}((|x|+t)^{-N-\frac{1}{2}})\) as \(t \to \infty\).
		\end{enumerate}
	\end{lemma}

    \par
    As the case of Painlev\'{e} region with $x>0$, one can open lense and transform the RH problem for the function $M(x, t; k)$ into the RH problem for $\check M(x, t; k)$ with contour $\Sigma^{2}$ in Figure \ref{Sigma1 painleve ge0} and the jump matrix $\check v$. It is obvious that the jump matrices on this contour tend to identity matrix $I$ as $t\to\infty$ except for the original point $k=0$.
	\begin{lemma}
		For $|\xi|\ge\frac{1}{M}$ and $x>0$, the jump matrices on contour $\Sigma^{2}$ tend to identity matrix $I$ rapidly as $t\to\infty$ except for the original point $k=0$. To be specific, for any $1\le p\le\infty$ and $N\ge1$, it follows that
		\begin{equation}
			\|\check v-I\|_{L^p}\le (|x|+t)^{-N}.
		\end{equation}
	\end{lemma}
	\begin{proof}
		By Lemma \ref{Rapid-Lemma}, for $k\neq0$, the jump matrices involving the terms $r_{j,a}~(j=1,2)$ decay exponentially, while the ones
        involving the terms $r_{j,r}~(j=1,2)$ is of order $\mathcal{O}\left((|x|+t)^{-N}\right)$.
	\end{proof}
    \par
	As a result, for any $1\le p\le \infty$, the solution $M(x,t;k)$ satisfies $\|M(x,t;k)-I\|_{L^p}=\mathcal{O}\left((|x|+t)^{-l}\right)$, for any positive integer $l$, so the reconstruction formula (\ref{recover-formula SK}) indicates that the solution $u(x,t)$ of the SK equation (\ref{SK}) decays rapidly in Sector $\rm V$ for $x<0$. Moreover, the analysis of the RH problem for function $m(x,t;k)$ associated with the mSK equation (\ref{msk-equation}) is analogous, thus recalling the reconstruction formula (\ref{recover-formula mSK}), the proof of the asymptotic expression in Sector $\rm IV$ for $x<0$ of the Theorem \ref{thm for msk} is completed. 
	

\appendix
    
	\section{The model problem 
    $M^{X_{A,B}}$}\label{appendix}
	\renewcommand{\thefigure}{A.\arabic{figure}}
\setcounter{figure}{0}
	Let $X_1=\{z\in\C:z=r\mathrm{e}^{\frac{\pi i}{4}},0\le r\le\infty\}$,~$X_2=\{z\in\C:z=r\mathrm{e}^{\frac{3\pi i}{4}},0\le r\le\infty\}$,  $X_3=\{z\in\C:z=r\mathrm{e}^{\frac{5\pi i}{4}},0\le r\le\infty\}$ and $X_4=\{z\in\C:z=r\mathrm{e}^{\frac{7\pi i}{4}},0\le r\le\infty\}$, depicted in Figure \ref{model-X}. Denote $X=\cup_{j=1}^4X_j$ and define the function $\nu(y)=-\frac{1}{2 \pi} \ln (1-\left|y\right|^2)$ from $B_{1}(0)$ to $(0,\infty)$.  In what follows, define the model problem for functions $M^{X_{A,B}}$ naturally.
	
	
	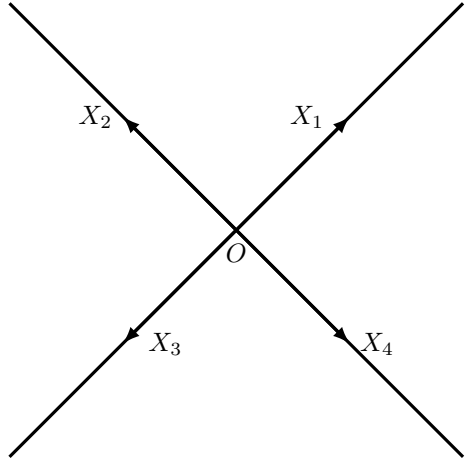
\begin{figure}[htp]
		\centering
		\begin{tikzpicture}[>=latex]
			\draw[<->,very thick] (-1.5,-1.5)node[right=2mm]{${X_3}$} to (1.5,1.5)node[left=2mm]{${X_1}$};
			\draw[very thick] (-3,-3) to (3,3);
			\draw[<->,very thick] (-1.5,1.5)node[left]{${X_2}$} to (1.5,-1.5)node[right]{${X_4}$};
			\draw[very thick] (-3,3) to (3,-3);
			\node at (0,-.3){$O$};
		\end{tikzpicture}
		\caption{{\protect\small
				The contour $X$ of the model problem for function $M^{X_{A,B}}$.}}
		\label{model-X}
	\end{figure}
	
	\begin{proposition}
		The $3 \times 3$ matrix-valued function $M^{X_A}$ satisfies the following properties:	
		
		(1). $M^{X_A}(y;\cdot\ ):\C\setminus X\to\C^{3\times3}$  is analytic for $z \in\C\setminus X$.
		
		(2). $M^{X_A}(y;z)$ is continuous for $z\in X\setminus\{0\}$ and satisfies the jump conditions:
		$$
\left(M^{X_A}(y;z)\right)_+=\left(M^{X_A}(y;z)\right)_-v^{X}_A(y;z),\quad z\in\C\setminus \{0\},
		$$
		where the jump matrix $v^{X_A}(y;z)$ is defined as
		$$
		\begin{aligned}
			& \left(\begin{array}{ccc}
				1 & 0 & 0 \\
				{\frac{\bar{y}}{1-|y|^2}} z^{-2 i \nu(y)} \mathrm{e}^{\frac{i z^2}{2}} & 1 & 0 \\
				0 & 0 & 1
			\end{array}\right) \quad \text { if } z \in X_1, \quad\left(\begin{array}{ccc}
				1 & y z^{2 i \nu(y)} \mathrm{e}^{-\frac{i z^2}{2}} & 0 \\
				0 & 1 & 0 \\
				0 & 0 & 1
			\end{array}\right) \text { if } z \in X_2, \\
			& \left(\begin{array}{ccc}
				1 & 0 & 0 \\
				-\bar y z^{-2i \nu(y) } \mathrm{e}^{\frac{iz^2}{2}} & 1 & 0 \\
				0 & 0 & 1
			\end{array}\right) \text { if } z \in X_3, \quad\left(\begin{array}{ccc}
				1 & -\frac{y}{1-|y|^2} z^{2i \nu(y) }  \mathrm{e}^{-\frac{iz^2}{2}} & 0 \\
				0 & 1 & 0 \\
				0 & 0 & 1
			\end{array}\right) \text { if } z \in X_4,
		\end{aligned}
		$$
		with $z^{2i\nu(y)}=\mathrm{e}^{2i\nu(y){\log_0(z)}}$ for choosing the branch cut along $\R_+$.
		
		(3). $M^{X_A}(y;z)\to I$  as $z\to\infty$.
		
		(4). $M^{X_A}(y;z)\to O(1)$ as $z\to 0$.
		
		For $|y|<1$, the solution $M^{X_A}(y;z)$ of the corresponding RH problem admits the following expansion:
		$$
		M^{X_A}(y;z)=I+\frac{M^{X_A}_1}{z}+\mathcal{O}\left(\frac{1}{z^2}\right),
		$$
		where
		$$
		M^{X_A}_1=\left(\begin{array}{ccc}
			0 & \beta_{12}^A & 0 \\
			\beta_{21}^A & 0 & 0 \\
			0 & 0 & 0
		\end{array}\right), \quad y \in B_1(0),
		$$
		and
		$$
		\beta_{12}^A=\frac{\sqrt{2 \pi} \mathrm{e}^{-\frac{\pi i}{4}} \mathrm{e}^{-\frac{5 \pi \nu}{2}}}{\bar{y} \Gamma(-i \nu)}, \quad \beta_{21}^A=\frac{\sqrt{2 \pi} \mathrm{e}^{\frac{\pi i}{4}} \mathrm{e}^{\frac{3 \pi \nu}{2}}}{y \Gamma(i \nu)}.
		$$
	\end{proposition}

	\begin{proposition}
		The $3 \times 3$ matrix-valued function $M^{X_B}$ satisfies the following properties:
		
		(1). $M^{X_B}(y;\cdot): \C\setminus X\to\C^{3\times3}$ is analytic for $z \in\C\setminus X$.
		
		(2). $M^{X_B}(y;z)$ is continuous for $z\in X\setminus\{0\}$ and satisfies the jump condition below:
		$$
\left(M^{X_B}(y;z)\right)_+=\left(M^{X_B}(y;z)\right)_-v^{X_B}(y;z),\quad z\in\C\setminus \{0\},
		$$
		where the jump matrix $v^{X_B}(y;z)$ is defined as 
		$$
		\begin{aligned}
			& \left(\begin{array}{ccc}
				1 & \bar{y} z^{-2 i \nu(y)} \mathrm{e}^{\frac{i z^2}{2}} & 0 \\
				0 & 1 & 0 \\
				0 & 0 & 1
			\end{array}\right) \quad \text { if } z \in X_1, \quad\left(\begin{array}{ccc}
				1 & 0 & 0 \\
				\frac{y}{1-|y|^2} z^{2 i \nu(y)} \mathrm{e}^{-\frac{i z^2}{2}} & 1 & 0 \\
				0 & 0 & 1
			\end{array}\right) \text { if } z \in X_2, \\
			& \left(\begin{array}{ccc}
				1 & -\frac{\bar y}{1-|y|^2} z^{-2i \nu(y) } \mathrm{e}^{\frac{iz^2}{2}} & 0 \\
				0 & 1 & 0 \\
				0 & 0 & 1
			\end{array}\right) \text { if } z \in X_3, \quad\left(\begin{array}{ccc}
				1 & 0 & 0 \\
				-y z^{2i \nu(y) }  \mathrm{e}^{-\frac{iz^2}{2}} & 1 & 0 \\
				0 & 0 & 1
			\end{array}\right) \text { if } z \in X_4,
		\end{aligned}
		$$
		with $z^{2i\nu(y)}=\mathrm{e}^{2i\nu(y)}\mathrm{e}^{\log_{\pi}(z)}$ for choosing the branch cut along $\R_-$.
		
		(3). $M^{X_B}(y;z)\to I$  as $z\to\infty$.
		
		(4). $M^{X_B}(y;z)\to O(1)$ as $z\to 0$.
		
		For $|y|<1$, the solution $M^{X_B}(y;z)$ of the corresponding RH problem admits the following expansion:
		$$
		M^{X_B}(y;z)=I+\frac{M^{X_B}_1}{z}+\mathcal{O}\left(\frac{1}{z^2}\right),
		$$
		where
		$$
		M^{X_B}_1=\left(\begin{array}{ccc}
			0 & \beta_{12}^B & 0 \\
			\beta_{21}^B & 0 & 0 \\
			0 & 0 & 0
		\end{array}\right), \quad y \in B_1(0),
		$$
		and
		$$
		\beta_{12}^B=\frac{\sqrt{2\pi}\mathrm{e}^{\frac{\pi i}{4}}\mathrm{e}^{-\frac{\pi\nu}{2}}}{ y\Gamma(i\nu)},\quad \beta_{21}^B=\frac{\sqrt{2\pi}\mathrm{e}^{-\frac{\pi i}{4}}\mathrm{e}^{-\frac{\pi\nu}{2}}}{\bar y\Gamma(-i\nu)}.
		$$
\end{proposition}

\section{ The Painlev\'e model problem}\label{appendix painleve}
\renewcommand{\thefigure}{B.\arabic{figure}}
\setcounter{figure}{0}
Define the RH problem for function $M^{P}(y;\lambda)$ with jump contour $\Sigma^{P}$ in Figure \ref{MP} and the jump matrices:
	\begin{equation}\label{vp}
		\begin{aligned}
			& v_1^P=\mathrm{e}^{ad(y\lambda\Lambda-\frac{9\lambda^5}{5}\Lambda^2)}
			\left(\begin{matrix}
				1 & \tilde{r}_1(0) & 0\\
				0 & 1 & 0\\
				0 & 0 & 1\\
			\end{matrix}
			\right), \quad
			&&v_{12}^P=\mathrm{e}^{ad(y\lambda\Lambda-\frac{9\lambda^5}{5}\Lambda^2)}
			\left(\begin{matrix}
				1 & 0 & 0\\
				\tilde{r}_1^*(0) & 1 & 0\\
				0 & 0 & 1\\
			\end{matrix}
			\right),\\	
			&v_2^P=\mathrm{e}^{ad(y\lambda\Lambda-\frac{9\lambda^5}{5}\Lambda^2)}
			\left(\begin{matrix}
				1 & 0 & 0\\
				0 & 1 & -\tilde{r}_2^*(0)\\
				0 & 0 & 1\\
			\end{matrix}
			\right), \quad
			&&v_{3}^P=\mathrm{e}^{ad(y\lambda\Lambda-\frac{9\lambda^5}{5}\Lambda^2)}
			\left(\begin{matrix}
				1 & 0 & 0\\
				0 & 1 & 0\\
				0 & -\tilde{r}_2(0) & 1\\
			\end{matrix}
			\right),\\	
			&v_4^P=\mathrm{e}^{ad(y\lambda\Lambda-\frac{9\lambda^5}{5}\Lambda^2)}
			\left(\begin{matrix}
				1 & 0 & \tilde{r}_1^*(0)\\
				0 & 1 & 0\\
				0 & 0 & 1\\
			\end{matrix}
			\right), \quad
			&&v_{5}^P=\mathrm{e}^{ad(y\lambda\Lambda-\frac{9\lambda^5}{5}\Lambda^2)}
			\left(\begin{matrix}
				1 & 0 & 0\\
				0 & 1 & 0\\
				\tilde{r}_1(0) & 0 & 1\\
			\end{matrix}
			\right),\\	
			&v_6^P=\mathrm{e}^{ad(y\lambda\Lambda-\frac{9\lambda^5}{5}\Lambda^2)}
			\left(\begin{matrix}
				1 & -\tilde{r}_2^*(0) & 0\\
				0 & 1 & 0\\
				0 & 0 & 1\\
			\end{matrix}
			\right), \quad
			&&v_{7}^P=\mathrm{e}^{ad(y\lambda\Lambda-\frac{9\lambda^5}{5}\Lambda^2)}
			\left(\begin{matrix}
				1 & 0 & 0\\
				-\tilde{r}_2(0) & 1 & 0\\
				0 & 0 & 1\\
			\end{matrix}
			\right),\\	
			&v_8^P=\mathrm{e}^{ad(y\lambda\Lambda-\frac{9\lambda^5}{5}\Lambda^2)}
			\left(\begin{matrix}
				1 & 0 & 0\\
				0 & 1 & 0\\
				0 & \tilde{r}_1^*(0) & 1\\
			\end{matrix}
			\right), \quad
			&&v_{9}^P=\mathrm{e}^{ad(y\lambda\Lambda-\frac{9\lambda^5}{5}\Lambda^2)}
			\left(\begin{matrix}
				1 & 0 & 0\\
				0 & 1 & \tilde{r}_1(0)\\
				0 & 0 & 1\\
			\end{matrix}
			\right),\\	
			&v_{10}^P=\mathrm{e}^{ad(y\lambda\Lambda-\frac{9\lambda^5}{5}\Lambda^2)}
			\left(\begin{matrix}
				1 & 0 & 0\\
				0 & 1 & 0\\
				-\tilde{r}_2^*(0) & 0 & 1\\
			\end{matrix}
			\right), \quad
			&&v_{11}^P=\mathrm{e}^{ad(y\lambda\Lambda-\frac{9\lambda^5}{5}\Lambda^2)}
			\left(\begin{matrix}
				1 & 0 & -\tilde{r}_2(0)\\
				0 & 1 & 0\\
				0 & 0 & 1\\
			\end{matrix}
			\right),\\
		\end{aligned}
	\end{equation}
	where $\Lambda:=\begin{pmatrix}
		\omega & 0 & 0\\
		0 &\omega^2 & 0\\
		0 & 0 &1
	\end{pmatrix}$ and $\mathrm{e}^{ad(A)}Y=\mathrm{e}^{A}Y\mathrm{e}^{-A}$. Notice that for $n=1,2,\cdots,8$, we have
    $$
    v_{n+4}^{P}=\mathcal{A}^{-1}v_{n}^{P}\mathcal{A},
	$$
    and for $n$ being integer odd with $ 1\le n\le12$, it follows that
	 $$
	v_{n-1}^{P}=\mathcal{B}\bar v_{n}^{P}\mathcal{B}.
      $$

	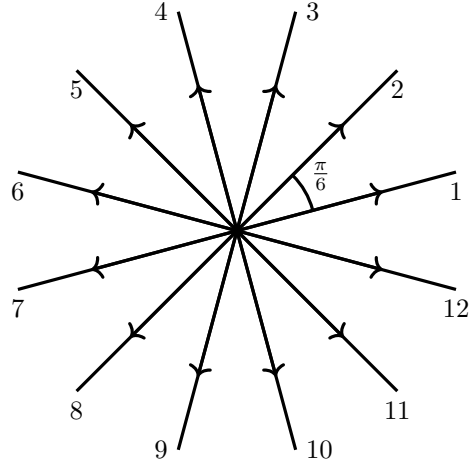
\begin{figure}[htp]
		\centering	
		\begin{tikzpicture}
			\draw[very thick, rotate around={15:(0,0)}] (-3,0)node[below] {$7$} -- (3,0) node[below] {$1$};
			\draw[very thick, rotate around={45:(0,0)}] (-3,0)node[below] {$8$} -- (3,0)node[below] {$2$};
			\draw[very thick, rotate around={75:(0,0)}] (-3,0)node[left] {$9$} -- (3,0)node[right] {$3$};
			\draw[very thick, rotate around={105:(0,0)}] (-3,0)node[right] {$10$} -- (3,0)node[left] {$4$};
			\draw[very thick, rotate around={135:(0,0)}] (-3,0)node[below] {$11$} -- (3,0)node[below] {$5$};
			\draw[very thick, rotate around={165:(0,0)}] (-3,0)node[below] {$12$} -- (3,0)node[below] {$6$};
			
			\draw[<->,very thick, rotate around={15:(0,0)}] (-2,0) -- (2,0) ;
			\draw[<->,very thick, rotate around={45:(0,0)}] (-2,0) -- (2,0);
			\draw[<->,very thick, rotate around={75:(0,0)}] (-2,0) -- (2,0);
			\draw[<->,very thick, rotate around={105:(0,0)}] (-2,0) -- (2,0);
			\draw[<->,very thick, rotate around={135:(0,0)}] (-2,0) -- (2,0);
			\draw[<->,very thick, rotate around={165:(0,0)}] (-2,0) -- (2,0);
			\draw[very thick] (1,0.28) arc[start angle=15, end angle=45, radius=1cm] node[right=1mm]{$\frac{\pi}{6}$};
		\end{tikzpicture}
		\caption{{\protect\small
				The jump contour $\Sigma^{P}$ of the RH problem for function $M^{P}(x,t;k)$.}}
		\label{MP}
	\end{figure}
    
	\begin{proposition}
		The $3\times3$ matrix-valued function $M^P(y;\lambda)$ has the following properties:	
		
		(1). $M^P(y;\cdot\ ): \C\setminus \Sigma^P\to\C^{3\times3}$ is analytic for $\lambda \in\C\setminus \Sigma^P$.
		
		(2). $M^P(y;\lambda)$ is continuous for $\lambda\in \Sigma^P\setminus\{0\}$ and satisfies the jump condition:
		$$
\left(M^P(y;\lambda)\right)_+=\left(M^P(y;\lambda)\right)_-V^{P}(y;\lambda),\quad \lambda\in\C\setminus \{0\},
		$$
		where the jump matrix $V^P=\{v^{P}_j(y;\lambda)\}_{j=1}^{12}$ is defined in (\ref{vp}).
		
		(3). $M^P(y;\lambda)\to I$  as $\lambda\to\infty$, and more precisely, we have
		$$
		M^P(y;\lambda)=I+\frac{M^P_{1}(y)}{\lambda}+\mathcal{O}\left(\frac{1}{\lambda^2}\right),
		$$
		where 3$(M^P_{1}(y))_{13}$ satisfies the fourth-order analogues of the Painlev\'{e} transcendent in (\ref{4th-Painleve}).
	\end{proposition} 
	
	\begin{proof}
		By using the jump condition, multiply the jump matrices recursively to arrive at
		$$
		v_1^Pv_2^P\cdots v_{12}^P=I,
		$$
		and then one has
		$$
		\tilde{r}_2(0)=\frac{\tilde{r}_2^*(0)^2-\tilde{r}_1^*(0)}{\tilde{r}_2^*(0)\tilde{r}_1^*(0)-1},\quad
		\tilde{r}_1(0)=\frac{\tilde{r}_1^*(0)^2-\tilde{r}_2^*(0)}{\tilde{r}_2^*(0)\tilde{r}_1^*(0)-1},
		$$
		which is coincided with the Assumption \ref{Painelve assumption}.
		It is immediate  that as $\lambda \to \infty$, the function $M^P(y;\lambda)$ has expansion $M^P(y;\lambda)=\sum_{j=0}^{\infty}\frac{M_j^P(y)}{\lambda^j}$ with $M_0^P=I$. 
        \par
        In particular, by the symmetry $M^P(y;\lambda)=\mathcal{A} M^P( y,\omega\lambda)\mathcal{A}^{-1}$, it follows that the coefficients $M_j^P(y)~(j=1,2,3,4)$ of the asymptotic expansion obey
		$$
		\begin{aligned}
			&M_{1}^P(y)=\omega^2\mathcal{A}M_{1}^P(y)\mathcal{A}^{-1},\ M_2^{P}(y)=\omega\mathcal{A}M_2^{P}(y)\mathcal{A}^{-1},\\
		&M_3^{P}(y)=\mathcal{A}M_3^{P}(y)\mathcal{A}^{-1},\ M_4^{P}(y)=\omega^2\mathcal{A}M_4^{P}(y)\mathcal{A}^{-1},
		\end{aligned}
		$$
		thus is is reasonable to assume that
		\begin{align*}
			&M_1^P(y)=\left(
			\begin{array}{ccc}
				a_1 & a_2 & a_3 \\
				a_3 \omega^2 & a_1 \omega^2 & a_2 \omega^2 \\
				a_2 \omega & a_3 \omega & a_1 \omega \\
			\end{array}
			\right),\
			M_2^P(y)=\left(
			\begin{array}{ccc}
				b_1 & b_2 & b_3 \\
				b_3 \omega & b_1 \omega & b_2 \omega \\
				b_2 \omega^2 & b_3 \omega^2 & b_1 \omega^2 \\
			\end{array}
			\right),\\
			&M_3^{P}(y)=\left(
			\begin{array}{ccc}
				c_1 & c_2 & c_3 \\
				c_3 & c_1 & c_2 \\
				c_2 & c_3 & c_1 \\
			\end{array}
			\right),\
			M_4^P(y)=\left(
			\begin{array}{ccc}
				d_1 & d_2 & d_3 \\
				d_3 \omega^2 & d_1 \omega^2 & d_2 \omega^2 \\
				d_2 \omega & d_3 \omega & d_1 \omega \\
			\end{array}
			\right).
		\end{align*}
        \par
		Furthermore, the conjugate symmetry $M^P(y;\lambda)=\mathcal{B}\overline{M^P(y,\bar\lambda)}\mathcal{B}$ indicates that
		
		\begin{align*}
		&a_1=\omega\bar a_1,~a_2=\omega \bar a_3,~a_3=\omega \bar a_2,~
			b_1=\omega^2\bar b_1,~b_2=\omega^2 \bar b_3,~b_3=\omega^2 \bar b_2,\\
		&c_1=\bar c_1,~c_2=\bar c_3,~c_3=\bar c_2,~
			d_1=\omega\bar d_1,~d_2=\omega \bar d_3,~d_3=\omega \bar d_2,
		\end{align*}
        which denotes that $\bar a_3=a_3$. This guarantees that the solution of the fourth-order analogues of the Painlev\'{e} transcendent in (\ref{4th-Painleve}) is real-valued, which is associated the real-valued solution of the mSK equation (\ref{msk-equation}).
        \par
		Define the $\mathcal{U}=\Psi_y\Psi^{-1}$ by
		$$
		\mathcal{U}=\Psi_y\Psi^{-1}=(M^P_y+\lambda M^P\Lambda)(M^P)^{-1},
		$$
		where $(M^P)^{-1}=I+\sum_{j=1}\frac{N_j^P}{\lambda^j}$, and since $\mathcal{U}$ is an entire function on $\lambda$, which means that
		$$
		\mathcal{U}(y;\lambda)=\mathcal{U}_0(y;\lambda)+\lambda\mathcal{U}_1(y;\lambda)\\
		=M_1^P(y)\Lambda+\Lambda N_1^P(y)+\lambda \Lambda.
		$$
		
		Define $\mathcal{A}=\Psi_{\lambda}\Psi^{-1}$ as
		$$
		\mathcal{A}=\Psi_{\lambda}\Psi^{-1}=(M^P_{\lambda}+M^P(y\Lambda-9\lambda^4\Lambda^2))(M^P)^{-1}.
		$$
		As the jump matrices are not concerned to $\lambda$ for $\Psi$, it implies that $\mathcal{A}$ is an entire function and can be expressed by
		$$
		\mathcal{A}=A_0+\lambda A_1+\lambda^2 A_2+\lambda^3 A_3+\lambda^4 A_4,
		$$
		where
		$$
		\begin{aligned}
			&A_0=y\Lambda-9(M_4^P\Lambda^2+M_3^P\Lambda^2N_1^P+M_2^P\Lambda^2N_2^P+M_1^P\Lambda^2N_3^P+\Lambda^2N_4^P),\\
			&A_1=-9(M_3^P\Lambda^2+M_2^P\Lambda^2N_1^P+M_1^P\Lambda^2N_2^P+\Lambda^2N_3^P),\\
			&A_2=-9(M_2^P\Lambda^2+M_1^P\Lambda^2N_1^P+\Lambda^2N_2^P),\\
			&A_3=-9(M_1^P\Lambda^2+\Lambda^2N_1^P),\\
			&A_4=-9\Lambda^2.
		\end{aligned}
		$$
		Furthermore, one has		
		\begin{align*}
		&N_1^P=-M_1^P,\ N_2^P=(M_1^P)^2-M_2^P,\ N_3^P=M_1^PM_2^P+M^P_2M^P_1-(M^P_1)^3-M^P_3,\\
		&N_4^P=(M^P_1)^4+M^P_1M^P_3+M^P_3M^P_1-(M^P_1)^2M^P_2-M^P_1M^P_2M^P_1-M^P_2(M^P_1)^2+(M^P_2)^2-M^P_4.
		\end{align*}
		
		Notice that
		$$
		\mathcal{U}=(M^P_y+\lambda M^P\Lambda)(M^P)^{-1},
		$$
		which implies that
		$$
		M_y^P=\mathcal{U}M^P-\lambda M^P\Lambda=\lambda[\Lambda,M^P]+[M_1^P,\Lambda]M^P.
		$$
		Moreover, it is seen that
		\begin{equation}
			\begin{aligned}
				(M_1^P)_y=[\Lambda,M_2^P]+[M_1^P,\Lambda]M_1^P,\\
				(M_2^P)_y=[\Lambda,M_3^P]+[M_1^P,\Lambda]M_2^P,\\
				(M_3^P)_y=[\Lambda,M_4^P]+[M_1^P,\Lambda]M_3^P,\\
			\end{aligned}
		\end{equation}
		which reduces that
		
		\begin{equation}\label{MP_y_equation}
			\begin{aligned}
				&a_1'(y)=3\omega^2a_3(y),\\
				&b_3(y)= \omega a_1(y) a_3(y)+\frac{a_3'(y)}{\omega-1}+\frac{\omega a_3(y)^2}{\omega+1},\\
				&b_2(y)=  a_1(y) a_3(y)-\frac{a_3'(y)}{\omega-1}-a_3(y)^2,\\
				&b_1'(y)=(1-\omega) \omega^2 a_3(y) (b_2(y)-\omega b_3(y)),\\
				&c_3(y)= \frac{\omega^2 \left((\omega-1) a_3(y) ((\omega+1) b_1(y)+b_2(y))-b_3'(y)\right)}{\omega^2-1},\\
				&c_2(y)= \omega a_3(y) (\omega b_1(y)-b_3(y))+\frac{b_2'(y)}{\omega-\omega^2},\\
				&c_1'(y)=(1-\omega) a_3(y) \left(c_2(y)-\omega^2 c_3(y)\right),\\
				&d_2(y)= \frac{a_3(y) \left(\omega^2 c_1(y)-c_3(y)\right)-\frac{c_2'(y)}{\omega-1}}{\omega},\\
				&d_3(y)= \frac{(\omega-1) \omega a_3(y) ((\omega+1) c_1(y)+\omega c_2(y))-c_3'(y)}{\omega \left(\omega^2-1\right)}.
			\end{aligned}
		\end{equation}
		\par
		On the other hand, it is obvious that
		$$
		\begin{cases}
			\begin{aligned}
				&\Psi_y=\mathcal{U}\Psi,\\
				&\Psi_{\lambda}=\mathcal{A}\Psi,
			\end{aligned}
		\end{cases}
		$$
		yields the comparable condition $\mathcal{A}_y-\mathcal{U}_{\lambda}+[\mathcal{A},\mathcal{U}]=0$, and it follows that
		
		\begin{equation}\label{MP_comparable}
			\begin{aligned}
				&\lambda^0:(A_0)_y-\mathcal{U}_1+[A_0,\mathcal{U}_0]=0,\\
				&\lambda^1:(A_1)_y+[A_1,\mathcal{U}_0]+[A_0,\mathcal{U}_1]=0,\\
				&\lambda^2:(A_2)_y+[A_2,\mathcal{U}_0]+[A_1,\mathcal{U}_1]=0,\\
				&\lambda^3:(A_3)_y+[A_3,\mathcal{U}_0]+[A_2,\mathcal{U}_1]=0,\\
				&\lambda^4:(A_4)_y+[A_4,\mathcal{U}_0]+[A_3,\mathcal{U}_1]=0,\\
				&\lambda^5:[A_4,\mathcal{U}_1]=0.
			\end{aligned}
		\end{equation}
        \par
	By using the ordinary differential equations in (\ref{MP_y_equation}), it is found that the equations in (\ref{MP_comparable}) can be reduced into
		\begin{equation}
			a_3^{(4)}(y)-45 a_3(y)^2 a_3''(y)+a_3(y) \left(y-45 a_3'(y)^2\right)-15 a_3'(y) a_3''(y)+81 a_3(y)^5=0,
		\end{equation}
		and it is immediate that $p(y):=\frac{a_3(y)}{3}$ satisfies 
        the fourth-order analogues of the Painlev\'{e} transcendent (\ref{4th-Painleve}) for $y=s$.
	\end{proof}
	\begin{proposition}
		The $3 \times 3$ matrix-valued function $m^p$ satisfies the following properties:	
		
		(1). $m^p(y;\lambda):\C\setminus \Sigma^p\to\C^{3\times3}$  is analytic for $\lambda \in\C\setminus \Sigma^p$, where the contour $\Sigma^p$ is depicted in Figure \ref{MPge0}.
		
		(2). $m^p(y;\lambda)$ is continuous for $\lambda\in\Sigma^p\setminus\{0\}$ and satisfies the jump condition below:
		$$
		\left(m^p(y;\lambda)\right)_+=\left(m^p(y;\lambda)\right)_-v^{p}(y;\lambda),\quad \lambda\in\C\setminus \{0\},
		$$
		where the jump matrix $v^p=v^{p}_j(y;\lambda)$. In particular, they are
		$$
		v_1^p(y;\lambda)=\mathrm{e}^{ad(y\lambda\Lambda-\frac{9\lambda^5}{5}\Lambda^2)}
		\left(\begin{matrix}
			1 & 0 & 0\\
			\tilde{r}_1^*(0) & 1 & 0\\
			0 & 0 & 1\\
		\end{matrix}
		\right),\  v_2^p(y;\lambda)=\mathrm{e}^{ad(y\lambda\Lambda-\frac{9\lambda^5}{5}\Lambda^2)}
		\left(\begin{matrix}
			1 & \tilde{r}_1(0) & 0\\
			0 & 1 & 0\\
			0 & 0 & 1\\
		\end{matrix}
		\right),
		$$
		and
		$$
		v_3^p(y;\lambda)=\mathrm{e}^{ad(y\lambda\Lambda-\frac{9\lambda^5}{5}\Lambda^2)}
		\left(\begin{matrix}
			1 & \tilde{r}_1(0) & 0\\
			\tilde{r}_1^*(0) & 1-|\tilde{r}_1(0)|^2 & 0\\
			0 & 0 & 1\\
		\end{matrix}
		\right).
		$$
		(3). $m^p(y;\lambda)\to I$  as $\lambda\to\infty$, in particular, it follows that
		$$
		m^p(y;\lambda)=I+\frac{m^p_{1}(y)}{\lambda}+\mathcal{O}\left(\frac{1}{\lambda^2}\right),
		$$
		where 3$(m^p_{1}(y))_{13}$ satisfies the fourth-order analogues of the Painlev\'{e} transcendent (\ref{4th-Painleve}) for $y=s$.
	\end{proposition}
	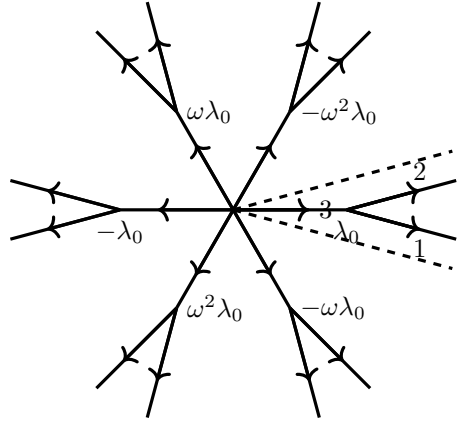
\begin{figure}[htp]
		\centering	
		\begin{tikzpicture}
			\draw[-,very thick, rotate around={15:(1.5,0)}] (1.5,0) -- (3,0) ;
			\draw[-,very thick, rotate around={-15:(1.5,0)}] (1.5,0) -- (3,0) ;
			\draw[->,very thick, rotate around={15:(1.5,0)}] (1.5,0) -- (2.5,0) ;
			\draw[->,very thick, rotate around={-15:(1.5,0)}] (1.5,0) -- (2.5,0) ;
			\draw[very thick, rotate around={-15:(1.5,0)}] (1.5,0) -- (2.5,0) node[below] {$1$} ;
			\draw[very thick, rotate around={15:(1.5,0)}] (1.5,0) -- (2.5,0) node[above] {$2$} ;
			\draw[very thick] (0,0)  to (1.5,0) node[below]{$\lambda_0$};
			\draw[very thick] (0,0)  to (1,0) node[right]{$3$};
			
			\draw[-,very thick,rotate=60, rotate around={15:(1.5,0)}] (1.5,0) -- (3,0) ;
			\draw[-,very thick,rotate=60, rotate around={-15:(1.5,0)}] (1.5,0) -- (3,0) ;
			\draw[->,very thick,rotate=60, rotate around={15:(1.5,0)}] (1.5,0) -- (2.5,0) ;
			\draw[->,very thick,rotate=60, rotate around={-15:(1.5,0)}] (1.5,0) -- (2.5,0) ;
			\draw[very thick,rotate=60] (0,0) to (1.5,0)node[right]{$-\omega^2\lambda_0$};
			
			\draw[-,very thick,rotate=120, rotate around={15:(1.5,0)}] (1.5,0) -- (3,0) ;
			\draw[-,very thick,rotate=120, rotate around={-15:(1.5,0)}] (1.5,0) -- (3,0) ;
			\draw[->,very thick,rotate=120, rotate around={15:(1.5,0)}] (1.5,0) -- (2.5,0) ;
			\draw[->,very thick,rotate=120, rotate around={-15:(1.5,0)}] (1.5,0) -- (2.5,0) ;
			\draw[very thick,rotate=120] (0,0) to (1.5,0)node[right]{$\omega\lambda_0$};
			
			\draw[-,very thick,rotate=180, rotate around={15:(1.5,0)}] (1.5,0) -- (3,0) ;
			\draw[-,very thick,rotate=180, rotate around={-15:(1.5,0)}] (1.5,0) -- (3,0) ;
			\draw[->,very thick,rotate=180, rotate around={15:(1.5,0)}] (1.5,0) -- (2.5,0) ;
			\draw[->,very thick,rotate=180, rotate around={-15:(1.5,0)}] (1.5,0) -- (2.5,0) ;
			\draw[very thick,rotate=180] (0,0) to (1.5,0)node[below]{$-\lambda_0$};
			
			\draw[-,very thick,rotate=240, rotate around={15:(1.5,0)}] (1.5,0) -- (3,0) ;
			\draw[-,very thick,rotate=240, rotate around={-15:(1.5,0)}] (1.5,0) -- (3,0) ;
			\draw[->,very thick,rotate=240, rotate around={15:(1.5,0)}] (1.5,0) -- (2.5,0) ;
			\draw[->,very thick,rotate=240, rotate around={-15:(1.5,0)}] (1.5,0) -- (2.5,0) ;
			\draw[very thick,rotate=240] (0,0) to (1.5,0)node[right]{$\omega^2\lambda_0$};
			
			\draw[-,very thick,rotate=-60, rotate around={15:(1.5,0)}] (1.5,0) -- (3,0) ;
			\draw[-,very thick,rotate=-60, rotate around={-15:(1.5,0)}] (1.5,0) -- (3,0) ;
			\draw[->,very thick,rotate=-60, rotate around={15:(1.5,0)}] (1.5,0) -- (2.5,0) ;
			\draw[->,very thick,rotate=-60, rotate around={-15:(1.5,0)}] (1.5,0) -- (2.5,0) ;
			\draw[very thick,rotate=-60] (0,0) to (1.5,0)node[right]{$-\omega\lambda_0$};
			\draw[<->,very thick] (-1,0) -- (1,0) ;
			\draw[<->,very thick, rotate around={60:(0,0)}] (-1,0) -- (1,0) ;
			\draw[<->,very thick, rotate around={120:(0,0)}] (-1,0) -- (1,0) ;
			
			\draw[very thick,dashed, rotate around={15:(0,0)}] (0,0) -- (3,0) ;
			\draw[very thick,dashed, rotate around={-15:(0,0)}] (0,0) -- (3,0) ;
		\end{tikzpicture}
		\caption{{\protect\small
				The jump contour of $\Sigma^{p}$, and the dashed line denote the jump contours related to $m^p(y;\lambda)$.}}
		\label{MPge0}
	\end{figure}
	\begin{proof}
		Indeed, the RH problem $m^p(y;\lambda)$ is equivalent to $M^P(y;\lambda)$ in the following sense:
		$$
		m^p(y;\lambda)=M^P(y;\lambda)\begin{cases}
			\begin{aligned}
				&(v^P_{1})^{-1},&&\lambda\ \text{on the left~} \text{side of } \Sigma^{p}_{2,3}  \ \text{and inside the dashed line},\\
				&(v^P_{12}),&&\lambda\ \text{on the right~} \text{side of } \Sigma^{p}_{1,3}\  \text{and inside the dashed line}.
			\end{aligned}
		\end{cases}
		$$
		It follows from $v_j^P~(j=1,2,\cdots, 12)$ in (\ref{vp}) that the jump matrices are bounded for $y\ge-c$, where $c$ is some nonnegative constant, moreover, the transformation above is invertible. Consequently, the expansion of $m^p(y;\lambda)$ is the same as $M^P(y;\lambda)$ for $\lambda\to\infty$.
	\end{proof}

\section*{Acknowledgments}
 This work was supported by the National Natural Science Foundation of China (Grant Nos. 12371247 and 12431008) and Beijing Natural Science Foundation Grant No. 1262012. {The last author acknowledge the support of the scholarship provided by the China Scholarship Council (CSC) under Grant No. 202406040149 and the GNFM-INDAM group and the research project Mathematical Methods in NonLinear Physics (MMNLP), Gruppo 4-Fisica Teorica of INFN.} The authors would like to express their sincere gratitude to Prof. Jonatan Lenells for his insightful comments and constructive suggestions on the initial manuscript. Special thanks are also extended to Dr. Lei Liu and Prof. Yuren Shi for their crucial assistance with the numerical simulations.

\bibliographystyle{amsplain}

\begin{thebibliography}{99}

\bibitem{Ablowitz-1977} 
 M. J. Ablowitz, H. Segur,  Asymptotic solutions of the Korteweg‐deVries equation, \emph{Studies  Appl. Math.} \textbf{57(1)} (1977)  13-44.

\bibitem{Ann-Teschl-2009} 
A. Boutet de Monvel, A. Kostenko, D. Shepelsky, G. Teschl, Long-time Asymptotics for the Camassa-Holm Equation, \emph{SIAM J. Math. Anal.} \textbf{41}(4) (2009) 1559-1588.

\bibitem{Beals-1985} 
R. Beals, The inverse problem for ordinary differential operators on the line, \emph{Am. J. Math.} \textbf{107}(2) (1985) 281-366.

\bibitem{Beals-Coifman-1984} 
R. Beals, R. R. Coifman, Scattering and inverse scattering for first order systems, \emph{Commun. Pure Appl. Math.} \textbf{37} (1984) 39-90.

\bibitem{Beals-Deift-1988} 
R. Beals, P. Deift, C. Tomei, \emph{Direct and inverse scattering on the line}, AMS, 1988.

\bibitem{Beals-Deift-Tomei} 
R. Beals, P. Deift, C. Tomei, \emph{Direct and inverse scattering on the line}, Math. Surv. Monogr., No. 28, AMS, Providence, 1988.

\bibitem{Boutet-de-Monvel-1} 
A. Boutet de Monvel, J. Lenells, D. Shepelsky, Long-time asymptotics for the Degasperis-Procesi equation on the half-line, \emph{Ann. de l'Institut Fourier} \textbf{69} (2019) 171-230.

\bibitem{Boutet-de-Monvel-2} 
A. Boutet de Monvel, D. Shepelsky, L. Zielinski, A Riemann-Hilbert Approach for the Novikov Equation, \emph{SIGMA Symmetry Integrability Geom. Methods Appl.} \textbf{12} (2016) 095.

\bibitem{Caudrey-Dodd-1976} 
P. J. Caudrey, R. K. Dodd, J. D. Gibbon, A new hierarchy of KdV equations, \emph{Proc. R. Soc. Lond. Ser. A Math. Phys. Eng. Sci.} \textbf{351} (1976) 407-422.

\bibitem{Tamara}T. Claeys, T. Grava, 
The KdV hierarchy: universality and a Painlev\'e transcendent,
Int. Math. Res. Not. \textbf{22} (2012) 5063-5099.

\bibitem{Charlier-Lenells-2021} 
C. Charlier, J. Lenells, The ``Good'' Boussinesq Equation: a Riemann-Hilbert Approach, \emph{Indiana Univ. Math. J.} \textbf{71}(4) (2022) 1505-1562.

\bibitem{Charlier-Lenells-2022} 
C. Charlier, J. Lenells, On Boussinesq's equation for water waves, arXiv:2204.02365, 2022.

\bibitem{Charlier-Lenells-Wang-2021}
C.~Charlier, J.~Lenells, and D.~S.~Wang,
The ``good'' Boussinesq equation: long-time asymptotics,
\emph{Anal. PDE} \textbf{6} (2023), 1351--1388.
\bibitem{Charlier-Lenells-miura} 
C. Charlier, J. Lenells, Miura transformation for the ``good'' Boussinesq equation, \emph{Stud. Appl. Math.} \textbf{152} (2024) 73.
{
\bibitem{bad4}
C.~Charlier and J.~Lenells,
Boussinesq's equation for water waves: asymptotics in Sector~V,
\emph{SIAM J. Math. Anal.} \textbf{56} (2024), no.~3, 4104--4142.}
{
\bibitem{bad2}
C.~Charlier and J.~Lenells,
The soliton resolution conjecture for the Boussinesq equation,
\emph{J. Math. Pures Appl.} (9) \textbf{191} (2024), Paper No.~103621, 46~pp.}
{
\bibitem{bad3}
{\sc C. Charlier, J. Lenells},
{\em Boussinesq's equation for water waves: asymptotics in Sector~I},
Adv. Nonlinear Anal., 13 (2024), no.~1, Paper No.~20240022, 32~pp.
}
{\bibitem{bad1}
C.~Charlier and J.~Lenells,
Boussinesq's equation for water waves: the soliton resolution conjecture for Sector~IV,
\emph{Adv. Nonlinear Stud.} \textbf{25} (2025), no.~1, 106--151.}


\bibitem{Constantin-0} 
A. Constantin, On the Scattering Problem for the Camassa-Holm Equation, \emph{Proc. Roy. Soc. London A} \textbf{457} (2001) 953-970.

\bibitem{Constantin-1} 
A. Constantin, R. P. Redou, On the inverse scattering approach to the Camassa-Holm equation, \emph{J. Nonl. Math. Phys.} \textbf{10} (2003) 252-255.

\bibitem{Constantin-2} 
A. Constantin, J. Lenells, On the Inverse Scattering Approach for an Integrable Shallow Water Equation, \emph{Phys. Lett. A} \textbf{308} (2003) 432-436.

\bibitem{Constantin-3} 
A. Constantin, V. S. Gerdjikov, R. I. Ivanov, Inverse scattering transform for the Camassa-Holm equation, \emph{Inverse Problems} \textbf{22} (2006) 2197.

\bibitem{Constantin-4} 
A. Constantin, R. I. Ivanov, J. Lenells, Inverse scattering transform for the Degasperis-Procesi equation, \emph{Nonlinearity} \textbf{23} (2010) 2559.

\bibitem{Cosgrove-2000} 
C. M. Cosgrove, Higher-order Painlevé equations in the polynomial class I. Bureau symbol P2, \emph{Stud. Appl. Math.} \textbf{104}(1) (2000) 1-65.

\bibitem{Cosgrove-2006} 
C. M. Cosgrove, Higher-order Painlevé equations in the polynomial class II. Bureau symbol P1, \emph{Stud. Appl. Math.} \textbf{116} (2006) 321-413.

\bibitem{Deift-2008}P. Deift, Some open problems in random matrix theory and the theory of integrable systems, pp. 419-430 in
Integrable systems and random matrices (New York, 2006), edited by J. Baik et al., Contemp. Math. 458, Amer. Math. Soc., Providence, RI, (2008).

\bibitem{Deift-Trubowitz} 
P. Deift, E. Trubowitz, Inverse scattering on the line, \emph{Comm. Pure Appl. Math.} \textbf{32} (1979) 121-251.

\bibitem{Deift-Zhou-1991} 
P. Deift, X. Zhou, Direct and inverse scattering on the line with arbitrary singularities, \emph{Commun. Pure Appl. Math.} \textbf{44}(5) (1991) 485-533.

\bibitem{Deift-Zhou-1993} 
P. Deift, X. Zhou, A steepest-descent method for oscillatory Riemann-Hilbert problem. Asymptotics of the MKdV equation, \emph{Ann. Math.} \textbf{137} (1993) 295-368.

\bibitem{Deift-Zhou-1994} 
P. Deift, S. Venakides, X. Zhou, The collisionless shock region for the long-time behavior of solutions of the KdV equation, \emph{Commun. Pure Appl. Math.} \textbf{47}(2) (1994) 199-206.



\bibitem{Fordy-Gibbons-1980} 
A. P. Fordy, J. Gibbons, Some remarkable nonlinear transformations, \emph{Phys. Lett. A} \textbf{75}(5) (1980) 325.

\bibitem{GGKM-1967} 
C. S. Gardner, J. M. Greene, M. D. Kruskal, R. M. Miura, Methods for solving the Korteweg-de Vries equation, \emph{Phys. Rev. Lett.} \textbf{19} (1967) 1095-1097.

\bibitem{Jimbo-Miwa-1983} 
M. Jimbo, T. Miwa, Solitons and Infinite Dimensional Lie Algebras, \emph{Publ. Res. Inst. Math. Sci.} \textbf{19} (1983) 943-1001.

\bibitem{Kaup-1980} 
D. J. Kaup, On the Inverse Scattering Problem for Cubic Eigenvalue Problems of the Class $\psi_{xxx}+6Q\psi_x+6R\psi=\lambda \psi$, \emph{Stud. Appl. Math.} \textbf{62} (1980) 189-216.

\bibitem{Kudryashov-2014} 
N. A. Kudryashov, D. I. Sinelshchikov, Extended models of non-linear waves in liquid with gas bubbles, \emph{Int. J. Non-Linear Mech.} \textbf{63} (2014) 31-38.

\bibitem{Kupershmidt-1984} 
B. A. Kupershmidt, A super Korteweg-de Vries equation: An integrable system, \emph{Phys. Lett. A} \textbf{102} (1984) 213-215.

\bibitem{Lax-1968} 
P. D. Lax, Integrals of nonlinear equations of evolution and solitary waves, \emph{Commun. Pure Appl. Math.} \textbf{21} (1968) 467-490.

\bibitem{Lenells-2018} 
J. Lenells, Matrix Riemann-Hilbert problems with jumps across Carleson contours, \emph{Monatsh. Math.} \textbf{186} (2018) 111-152.

\bibitem{Lenells-mkdv} 
C. Charlier, J. Lenells, Airy and Painlevé asymptotics for the mKdV equation, \emph{J. Lond. Math. Soc.} \textbf{101} (2020) 194.

\bibitem{Lin-Wang-Zhu} 
L. Huang, D. S. Wang, X. D. Zhu, Long-time asymptotics of the Tzitzéica equation on the line, arXiv:2404.04999.


\bibitem{McKean1981} 
H. P. McKean, Boussinesq's Equation on the Circle, \emph{Comm. Pure Appl. Math.} \textbf{34} (1981) 599-691.

\bibitem{procA} 
D. S. Wang, X. D. Zhu, Direct and inverse scattering problems of the modified Sawada-Kotera equation: Riemann-Hilbert approach, \emph{Proc. R. Soc. A} \textbf{478} (2022) 20220541.

\bibitem{Sawada-Kotera-1974} 
K. Sawada, T. Kotera, A method for finding N-soliton solutions of the K.d.V. equation and K.d.V.-Like Equation, \emph{Progr. Theoret. Phys.} \textbf{51} (1974) 1355-1367.



\bibitem{Zabolotskii} 
A. A. Zabolotskii, Inverse scattering transform for the Yajima-Oikawa equations with nonvanishing boundary conditions, \emph{Phys. Rev. A} \textbf{80} (2009) 063616.


\bibitem{Zakharov1973} 
V. E. Zakharov, On stochastization of one-dimensional chains of nonlinear oscillators, \emph{Zh. Eksp. Teor. Fiz.} \textbf{65} (1973) 219-225.




		
		
		
		
		
	\end{thebibliography}

\end{document}